\documentclass[12pt]{article}
\usepackage[T1]{fontenc}
 \usepackage{float}

\usepackage{amsfonts,amssymb,amsmath,amsthm,epsfig,euscript,bbm,amsthm}
\usepackage{algorithm}
\usepackage{algpseudocode}
\usepackage{amsthm}
\usepackage[margin=1.22in]{geometry}
\usepackage{setspace}
  
\usepackage{amsfonts,amssymb,amsmath,amsthm,epsfig,euscript,bbm,amsthm}

\usepackage{mathtools}
\usepackage[usenames,dvipsnames]{xcolor}
\usepackage{algorithm}
\usepackage[dvipsnames]{xcolor}
\usepackage{algpseudocode}
\usepackage{amsthm}
\usepackage{threeparttable,subcaption}
 \usepackage{booktabs}

 \usepackage{titletoc}
\usepackage{rotating,multirow,makecell,array}
\usepackage[inline]{enumitem}
\usepackage{rotating}
\usepackage{lscape}
\usepackage{float}
\usepackage{graphicx}
\usepackage{subcaption}
\usepackage{tikz}
\usetikzlibrary{arrows,calc}
\usepackage{titletoc}
 
\usepackage{rotating,multirow,makecell,array}
\usepackage{bm}
\usepackage{tikz}
\usetikzlibrary{shapes,positioning}
\usetikzlibrary{arrows.meta,positioning,calc}
\usetikzlibrary{arrows}
\usepackage {graphicx}
\usepackage{enumerate}
\definecolor{darkbrown}{rgb}{0.4, 0.26, 0.13}

\usepackage{subcaption}
\usepackage{caption}
\usepackage[authoryear]{natbib}
\usepackage[titletoc,title]{appendix}
\usepackage[pdftex,colorlinks]{hyperref}
\usepackage{mathtools}
 \usepackage{xcolor}
\usepackage{xr}
\usepackage{algorithm,algpseudocode}

\newcounter{algsubstate}
\renewcommand{\thealgsubstate}{\alph{algsubstate}}

\usepackage{graphicx}
\usepackage{subcaption}
 \usetikzlibrary{graphs, graphs.standard}
 \definecolor{orange}{RGB}{230,170,120}
  \definecolor{green}{RGB}{120,200,120}

 \hypersetup{
 	colorlinks=true,
 	linkcolor=blue,
 	filecolor=blue,      
 	urlcolor=cyan,
 }
 \hypersetup{colorlinks,linkcolor={blue},citecolor={blue}} 
 \usepackage{geometry}                
 \usepackage{multirow}
\usepackage[section]{placeins}
 \usepackage{bbm}
 \usepackage{graphicx}
 \usepackage{amssymb}
 \usepackage{epstopdf}
 \usepackage{tikz} 
 \usepackage{amsmath} 
\usepackage[utf8]{inputenc}
\usepackage{algorithm}
\usetikzlibrary{arrows.meta}
\usetikzlibrary{calc}
\usepackage{setspace}
 \usepackage{epstopdf}
\usetikzlibrary{arrows, chains, calc, positioning, shapes.multipart}
 \theoremstyle{plain}
 \newtheorem{thm}{Theorem}[section]
 \newtheorem{lem}[thm]{Lemma}
 \newtheorem{prop}[thm]{Proposition}

 \theoremstyle{definition}
 \newtheorem{defn}{Definition}[section]
 
 \newtheorem{exmp}{Example}[section]
  \newtheorem{ass}{Assumption}[section]
  \usepackage{multibib}
 \theoremstyle{definition}
 \newtheorem{rem}{Remark}
 
 \makeatletter
 \def\BState{\State\hskip-\ALG@thistlm}
 \makeatother
 
 \usepackage{authblk}
\def\spacingset#1{\renewcommand{\baselinestretch}%
{#1}\small\normalsize} \spacingset{1}
 \usepackage{xr}
\definecolor{coquelicot}{rgb}{1.0, 0.22, 0.0}

\algnewcommand\algorithmicforeach{\textbf{for each}}
\algdef{S}[FOR]{ForEach}[1]{\algorithmicforeach\ #1\ \algorithmicdo}

   \usepackage{titlesec}

\titlespacing*{\section}{0pt}{1.0ex plus .2ex minus .1ex}{0.5ex}
\titlespacing*{\subsection}{0pt}{0.8ex plus .2ex minus .1ex}{0.4ex}
\titlespacing*{\subsubsection}{0pt}{0.6ex plus .2ex minus .1ex}{0.3ex}

\titleformat{\paragraph}[runin]
  {\normalfont\normalsize\bfseries}{\theparagraph}{0.5em}{}
\titlespacing*{\paragraph}{0pt}{0.4ex plus .1ex minus .1ex}{0.6em}
   
\title{Evidence aggregation with ignorance in mind: 
learning what we do (not) know for archetypes discovery\footnote{We thank Isaiah Andrews, Paul Goldsmith-Pinkham, Florencia Hnilo, Guido Imbens, Madeline McKelway, Konrad Menzel, Jose Montiel-Olea, Muriel Niederle, Ashesh Rambachan, Neil Shephard, Jesse Shapiro, Rahul Singh, Nicholas Swanson, Elie Tamer, the editor and referees for helpful comments. We thank Jesus David Martinez Hernandez, Alexander Almeida, and Camilla Cherubini for exceptional research assistance. This project has received financial support by the NSF Grant SES 2447088. Davide Viviano also acknowledges generous support from the Harvard Griffin fund. } 
  }
\author{Emily Breza $\quad$ Arun G. Chandrasekhar $\quad$ Davide Viviano}
\date{First version: January, 2025 \\ This version: \today }

\begin{document}
\maketitle

 \spacingset{1.45}

 \begin{abstract}

When evaluating policy interventions, researchers often pursue two related goals: identifying which individuals or contexts benefit most, and determining whether patterns of treatment effect heterogeneity can be used to aggregate evidence across environments. We develop a framework that aggregates treatment effect heterogeneity, defined over individual and environmental characteristics, into interpretable summaries while setting aside contexts in which extrapolation is unreliable and further evidence is needed. The procedure therefore learns both how to summarize heterogeneous effects and when researchers should admit ignorance. We derive finite-sample regret guarantees, provide data-driven guarantees for selecting the complexity of the summary class, and inference procedures that quantify the value of follow-up data collection. We illustrate the approach by reanalyzing a multifaceted anti-poverty program implemented in six countries.

 \end{abstract}

\newpage

\section{Introduction}

Given the rise of experimental methods and increasingly rich data in social science, researchers can now measure the effects of policy interventions in larger, more representative populations and across a wider range of contexts. In many applications, policymakers want to know where and for whom a promising intervention should be scaled, and when additional evidence is needed. More broadly, social scientists are interested in learning and summarizing heterogeneous effects from the data. Achieving these goals requires understanding patterns of treatment effect heterogeneity across a potentially high-dimensional set of observable characteristics.\footnote{We can interpret the study of heterogeneity both in applications to meta-analysis, where researchers have access to multiple studies, and in applications to treatment effect heterogeneity within a single experiment.} In this paper, we develop an econometric framework and a set of empirical tools to summarize heterogeneous treatment effects with an eye toward both immediate use and future investigation.

As an illustrative example, consider the multifaceted ``Graduation program'' studied experimentally in six countries by \cite{banerjee2015multifaceted}. The program aims to lift individuals out of extreme poverty through income generation and typically combines a large asset transfer, training, savings accounts, and short-term cash transfers. In this setting, heterogeneity is first-order: ex ante, it is unclear which individuals may respond most strongly to the program (for example, by age, relative wealth, or marital status), and how market conditions may shape its effectiveness (for example, through credit access and labor demand).\footnote{While a researcher could explicitly model each of these forces and incorporate them structurally into estimation, the class of plausible models is often very large since different mechanisms might be at play in different places and may be ex-ante unknown.} To summarize such high-dimensional heterogeneity, the researcher must understand when evidence from one context (say widows in Pakistan) can inform another (say job seekers in Peru), and when it cannot. If some contexts are both noisy and substantively different from the rest, it may be preferable to withhold prediction altogether until more data are collected.

Motivated by this observation, our goal is to identify systematic groups of observable characteristics (\textit{archetypes}) that are informative for others under a statistical or economic model, while also detecting the contexts for which such aggregation is not credible given the available data. Rather than forcing conclusions in all environments, we allow the researcher to identify which observations should \textit{not} be pooled, thereby revealing where additional evidence is most valuable (denoted as \textit{basin of ignorance}).


Specifically, researcher observes data from a large number of environments together with individual outcomes and baseline characteristics. Their problem is to aggregate heterogeneity using an interpretable model, while retaining the option to abstain from prediction and recommend additional experimentation in some contexts at a given cost. The framework therefore balances two objectives: using existing evidence to form predictions where credible, and identifying where new data would be most useful. 


Using existing studies, we construct estimators in two steps. First, for each (\textit{small}) group $x$ of the observable individual-level and environmental characteristics (with different age, education, baseline  consumption, site of the experiment, etc.), we form unbiased but possibly noisy estimates of the conditional average effect (CATE) and its variance. Second, we assign each of these units to either an archetype or the basin of ignorance. Assignment to the basin of ignorance incurs a fixed cost/loss.  
The estimated loss for groups comprising an archetype is instead equal to the approximation error of the model used to aggregate evidence across archetypes, estimated by taking its squared prediction error, and subtracting the within sampling variation at $x$.

We justify our approach with three sets of results. We first study the regret as the gap between the researcher’s loss under our estimator and the loss of the oracle rule that observes treatment effects without estimation error, holding fixed the class of admissible models. Under standard moment conditions, we show that this regret vanishes at a $1/\sqrt{N}$-rate in the sample size $N$. The result relies on independent (but not identically distributed) observations together with complexity restrictions on both the model class and the basin of ignorance. Our analysis applies to general model classes with bounded complexity, and we then specialize it to discrete partitions, for which we derive matching minimax upper and lower bounds in the sample size. 

We next allow the researcher to select the model class from the data via sample-splitting procedure, such as the depth of a tree for tree-based partitions. The resulting procedure is compared to an oracle that, absent statistical error, is allowed to use the richest class in the collection to predict across all environments. Our regret bound shows that the selected rule performs nearly as well as the model that minimizes, over all candidate classes, the sum of three terms: a statistical error term, which increases with model complexity; a bias component, which decreases with complexity; and the cost of abstention. A key advantage of our framework relative to standard model-selection results \citep[e.g.][]{mbakop2021model} is that the bias is evaluated only on the subset of types for which the procedure chooses to predict, rather than on the full sample. This clarifies how abstention improves model selection by allowing the researcher to use a simpler and more precise model where it fits well, while withholding prediction in contexts where it performs poorly.

We the turn to the choice of the abstention cost. We give a decision-theoretic interpretation that links this cost to the mean-squared error of a prediction formed after collecting additional data in a follow-up experiment, and, in turn, to the sample size of the follow-up study. This yields transparent benchmark and data-driven calibrations for the cost of abstention through a break-even frontier: the minimum follow-up sample size required for abstention to improve upon the best rule that predicts everywhere with high-probability, identifying when additional experimentation is most valuable.

Finally, we provide exact and approximate optimization procedures where we consider both model classes of discrete partitions, and smooth prediction functions.

We apply our method to the multifaceted Graduation program. We choose a depth-four tree (which finds eleven archetypes) by cross-validation and assigns approximately \(15\%\) of observations to the basin of ignorance through the break-even analysis. The resulting tree reveals substantial heterogeneity: predicted effects range from approximately zero to \(48\%\) of a control-group standard deviation, with effects typically larger for poorer individuals. Several predictive archetypes are shared across countries, showing how the method aggregates evidence across sites. In contrast, the corresponding no-abstention tree primarily splits by country, and creates several small, and higher-variance groups with extreme predictions. Out-of-sample validations and calibrated simulations illustrate large gains in terms of mean-squared error on the predictive set compared to standard procedures such as Empirical Bayes, CART, and generalized random forests that do not account for ignorance.

\paragraph{Related literature} We connect to the literature on meta-analysis and machine-learning methods for treatment effect heterogeneity, both of which are increasingly prevalent in applied work.\footnote{For example, recent meta-analyses tackle topics including deworming, cash transfers, education interventions, the link between democracy and growth, and tests of Allport's contact hypothesis \citep{croke2024meta,angrist2023implementation,crosta2024unconditional,doucouliagos2008democracy,paluck2019contact}. A related empirical literature has also emerged focusing on policy design and targeting \citep{banerjee2025selecting,haushofer2022targeting}.} Across these domains \citep[e.g.][]{meager2022aggregating,spiess2023finding, chernozhukov2018generic, athey2018generalizedrandomforests, venkateswaran2024robustly, bonhomme2015grouped, ishihara2021evidence, menzel2023transfer, kim2026causal, adjaho2022externally, manski2004statistical, athey2021policy, kitagawa2018should}, the objective is typically to estimate heterogeneous effects, or make treatment decisions, for every context in the population of interest. Our departure is to allow the researcher to abstain from prediction, thereby learning where observations should not be pooled and where additional evidence is warranted.

This abstention margin changes the role of the statistical or economic model. As we formalize in Section \ref{sec:model_selection}, assumptions such as sparsity or smoothness \citep[e.g.][]{athey2018generalizedrandomforests, chernozhukov2018generic} need only hold locally on an ex-ante unknown subset of the data, rather than globally as in much of the existing heterogeneity literature. This makes the resulting procedures more robust in practice. Likewise, methods that address statistical noise through power considerations or shrinkage \citep[e.g.][]{spiess2023finding, meager2019understanding} do not accommodate units that remain poorly predicted even \textit{absent} estimation error because of model misspecification. Such units matter not only because they are difficult to predict, but also because they can distort predictions for the remaining units.

We also connect to the robust statistics literature \citep[e.g.][]{huber2011robust,garcia1999robustness}. We complement this literature that typically either posit a robust loss function ex ante, or focus on quantifying robustness through informative robustness metrics \citep[e.g.][]{broderick2020automatic}, here instead by learning the ignorance set from the data. This perspective leads to a different objective: a researcher-facing rule that jointly governs prediction, abstention, and the value of collecting more data and novel (regret) guarantees. 

Our contribution also lies between two adjacent strands of work. Some papers study evidence aggregation for a fixed prediction rule when no further experimentation is possible \citep{deeb2019clustering, bisbee2017local, andrews2022transfer, manski2020toward}. Others use heterogeneity to guide experimental design in the absence of current evidence \citep{gechter2024selecting, olea2024externally}. We instead use existing data both to form predictions where they are credible and to identify where additional experimentation is most valuable. In this sense, the paper also connects to machine-learning work on classification with rejection options \citep{chow1970optimum, cortes2016learning, franc2023optimal} and, more broadly, to \cite{shafer1992dempster}. Unlike that literature, and work on regression with abstention under correct specification/exchangeability \citep{denis2020regression, sokol2024conformalized}, we study the different problem of treatment effect heterogeneity, and solve a joint prediction-and-abstention problem where misspecification and non-exchangeability may naturally occur. This yields a different class of estimators, new regret guarantees (including those for model selection), and a decision-theoretic connection between abstention and experimentation.


\section{Problem description: ignorance-aware predictions} \label{sec:setup}

Consider a setting in which units are partitioned into many observable types \(x \in \mathcal X\). A type may be defined by the cross-product of individual characteristics and study-level characteristics, such as the site or country of an experiment. We allow \(\mathcal X\) to be discrete and high-dimensional, so that the number of types may grow with the overall sample size (in our application $|\mathcal{X}| = 796$). For each type \(x\), the object of interest is a scalar \emph{property} \(\tau(x)\in\mathbb R\), such as a conditional average treatment effect for a given outcome. (Appendix \ref{sec:multiple_properties} presents multivariate $\tau(x)$.) $\tau(x)$ is non-random. 

Researchers observe, for each \(x\in\mathcal X\), noisy estimates
$
\big(\hat\tau(x),\hat\eta(x)^2\big)\sim_{i.n.i.d.}\mathcal D_x,
$ 
independently across \(x\), where \(\mathcal D_x\) may be unknown and heterogeneous, and
\begin{equation} \label{eqn:unbiased1}
\small 
\begin{aligned} 
\mathbb E[\hat\tau(x)] = \tau(x), 
\qquad
\eta(x)^2 := \mathbb E[\hat\eta(x)^2] = \mathbb{V}(\hat\tau(x)).
\end{aligned} 
\end{equation} 
Thus, for each type \(x\), we observe an unbiased estimate of the target \(\tau(x)\) together with an unbiased estimate of its sampling variance, which needs not to be consistent when $N_{\mathcal{X}}$ is high-dimensional. 
We relax independence across \(x\) in Appendix \ref{sec:local}.

\begin{exmp} \label{exmp:ipw}
Consider a randomized experiment with independent treatments \(D_i\in\{0,1\}\), observed outcome
$ 
Y_i=D_iY_i(1)+(1-D_i)Y_i(0).
$ Let
$ 
e(x)=\Pr(D_i=1\mid X_i=x)$, $
n(x)=\big|\{i:X_i=x\}\big|,
$ 
with \(n(x)\ge 2\). For instance, here $x$ may denote individuals with similar baseline consumption, assets, food security, the same site of the study, age, education and gender, with $n(x)$ potentially small. Let
$ 
\tau(x)=\mathbb E\!\left[Y_i(1)-Y_i(0)\mid X_i=x\right],
$
and
$
\tilde Y_i
=
\frac{D_iY_i}{e(X_i)}
-
\frac{(1-D_i)Y_i}{1-e(X_i)}, 
$ 
and let 
\begin{equation}\label{eqn:pseudo_true}
\small 
\begin{aligned} 
\hat\tau(x)
=
\frac{1}{n(x)}
\sum_{i:X_i=x}\tilde Y_i, \qquad \hat\eta(x)^2
=
\frac{1}{n(x)(n(x)-1)}
\sum_{i:X_i=x}\big(\tilde Y_i-\hat\tau(x)\big)^2.
\end{aligned} 
\end{equation}
 Equation \eqref{eqn:unbiased1} holds, and both can be noisy estimates with small $n(x)$. Extensions to continuous covariates can also be accommodated as described in Appendix \ref{sec:matching}.
 \qed 
\end{exmp}

\subsection{Ignorance-aware predictions}

For interpretability and statistical precision, researchers may wish to summarize heterogeneity $\tau(x)$ using a low-complexity prediction rule. Common examples include predictions that group individuals into a finite number of groups via clustering \citep[e.g.][]{chernozhukov2018generic, athey2016recursive, bonhomme2015grouped}, sparse high-dimensional linear models \citep{banerjee2025selecting} or smooth function classes \citep{athey2018generalizedrandomforests, haushofer2022targeting}.  Without imposing restrictions on the complexity of the underlying treatment effects $\tau$, we seek a simpler approximation \(f\) in a candidate class \(\mathcal F\), where \(\mathcal F\) captures economic, communication, or statistical constraints. We start from $\mathcal{F}$ as given, and defer to Section~\ref{sec:model_selection} model selection for $\mathcal{F}$.

A simpler approximation is desirable since it may be easy to communicate and avoid overfitting, but may approximate well effect heterogeneity $\tau(x)$ only for a subset of units $x$. Types $x$ that are poorly approximated within \(\mathcal F\) can distort the choice of \(f\) and increase approximation errors for other units $x$. To address this concern, we introduce a decision problem where researchers can either report a prediction $f(x)$ or abstain. Abstaining at \(x\) defers the prediction until further evidence is collected. 

Formally, let \(\pi(x)\in\{0,1\}\) indicate whether a prediction is reported for type \(x\), and let \(c(x)\ge 0\) denote the opportunity cost of abstaining. Section~\ref{sec:cost_abstention} provides practical recommendations for the choice of $c(x)$ and shows how to interpret $c(x)$ as the corresponding cost of collecting additional data. Define the corresponding loss as 
\begin{equation}\label{eqn:loss}
\small 
\begin{aligned} 
R_c(f,\pi)
=
\sum_{x\in\mathcal X} p(x)
\left[
\left(f(x) - \tau(x)\right)^2\pi(x)+c(x)\bigl(1-\pi(x)\bigr)
\right],
\end{aligned} 
\end{equation}
where \(p(x)\) is the target distribution over types, treated as known by the researcher. (Appendix \ref{app:estimated_p} extends our framework to unknown/estimated $p(x)$.) For given $\pi$, define
$ 
\{x\in\mathcal X:\pi(x)=0\}
$  
as the basin of ignorance, i.e., the set for which we abstain from predictions, and its complement as the prediction region (archetypes). 


Given \(\mathcal F\) and a class of abstention rules \(\Pi\), the \textit{oracle} ignorance-aware prediction problem is
\begin{equation}
\small 
\begin{aligned} \label{defn:generalizable_predictions}
(\pi^\star,f^\star)\in
\arg\min_{\pi\in\Pi, f\in\mathcal F}
R_c(f,\pi).
\end{aligned} 
\end{equation}
Ignorance-aware prediction jointly chooses a simple summary $f \in \mathcal{F}$ of \(\tau\) and the set of types $x$ for which that summary is reliable enough to report. As for restrictions on $\mathcal{F}$, the complexity of $\Pi$ characterizes communication or design constraints for abstention. (It is possible that $\mathcal{F}$ is a function of $\pi$, omitted for notational convenience.)

Existing methods for effect heterogeneity do not allow for abstention, and as such they are nested as the limiting case in which abstention is prohibitively costly. Our formulation instead allows the researcher to withhold prediction when the loss from forcing a low-complexity summary is larger than the opportunity cost of deferral, here capturing the opportunity cost of eliciting further evidence. As we show in Section \ref{sec:theory} this allows researchers to trade-off cost of future data collection,  the bias of $f$ within the prediction set and precision of the chosen prediction.

\begin{exmp}[Supervised clustering] \label{exmp:supervised_clustering}
A leading example summarizes heterogeneity by grouping types into a finite number of archetypes, a common practice in economics \citep{chernozhukov2018generic, athey2016recursive}. Let
$
\alpha:\mathcal X\to\{1,\dots,G\}$, $\alpha\in\mathcal G,
$ 
where \(\alpha(x)=1\) denotes abstention and \(\alpha(x)=g\ge 2\) assigns type \(x\) to archetype \(g\), and \(\mathcal G\) is an admissible class of partition rules. The corresponding prediction rule is a constant prediction within each prediction group $ 
f_\alpha(x)=\mu_{\alpha(x)}$, for $\alpha(x)\ge 2.
$ 

This case is a clustering problem: If \(\tau(x)\) were known (or consistently estimated at each $x$), and \(\mathcal G\) unrestricted, the oracle problem can be solved via dynamic programming, a useful observation recently made by \cite{montielolea_archetype_discovery}, connecting to one-dimensional clustering \citep[e.g.][]{wang2011ckmeans}. 

In applications researchers may also restrict \(\mathcal G\) to interpretable classes, such as tree partitions as in Figure \ref{fig:tree_clustering_combined}(Panel (a)) or clustering over space of covariates \(x \in \mathbb{R}^d\) as in Figure \ref{fig:tree_clustering_combined}(Panel (b)) to both improve interpretability (e.g., to study how baseline consumption, assets and marital status may or may not predict heterogeneous responses) and allow for (many) noisy estimates of $\tau(x)$, see e.g., Section \ref{sec:clustering}. 
\end{exmp}

\begin{exmp}[Linear models (with sparsity).] \label{exmp:economic_example} 
Researchers may posit a linear model 
$ 
f_\theta(x)=q(x)'\theta,
$ 
where \(q(x)\) is a possibly high dimensional vector of individual and site characteristics and $\theta$ is sparse (i.e., $\sum_j 1\{\theta_j \neq 0\} \le s$ for a constant $s$), a requirement often imposed in development applications \citep[e.g.][]{banerjee2025selecting}. 
\end{exmp}

\begin{figure}[!ht]
\centering

\begin{minipage}[t]{0.36\textwidth}
\centering
\vspace{0pt}
{\small (a) Tree}\\[0.4em]

\begin{tikzpicture}[
    grow=down,
    level distance=1.25cm,
    sibling distance=24mm,
    edge from parent/.style={draw,->},
    every node/.style={font=\small},
    decision/.style={
        draw, rounded corners, fill=white,
        minimum width=18mm, minimum height=7mm, align=center
    },
    ign/.style={
        draw, rounded corners, fill=gray!20,
        minimum width=18mm, minimum height=7mm, align=center
    },
    gtwo/.style={
        draw, rounded corners, fill=red!18,
        minimum width=14mm, minimum height=7mm, align=center
    },
    gthree/.style={
        draw, rounded corners, fill=blue!18,
        minimum width=14mm, minimum height=7mm, align=center
    }
]

\node[decision] {split on\\ $x_{j(1)}$}
child {
    node[decision] {split on\\ $x_{j(2)}$}
    child {
        node[ign, yshift=-2mm] {Ignorance}
    }
    child {
        node[gtwo, yshift=2mm] {$g=2$}
    }
}
child {
    node[decision] {split on\\ $x_{j(3)}$}
    child {
        node[ign, yshift=-6mm] {Ignorance}
    }
    child {
        node[gthree, yshift=-1mm] {$g=3$}
    }
};

\end{tikzpicture}
\end{minipage}
\hfill
\begin{minipage}[t]{0.60\textwidth}
\centering
\vspace{0pt}
{\small (b) Clustering on covariate space}\\[0.4em]

\begin{tikzpicture}[scale=0.65, every node/.style={font=\scriptsize}]

    \draw[->] (-0.2,0) -- (8.5,0) node[below] {$x_1$};
    \draw[->] (0,-0.2) -- (0,6.5) node[left] {$x_2$};

    \fill[gray!28] 
        (0,3.8) -- (0,6) -- (2.4,6) -- (3.2,4.8) -- (2.6,3.5) -- cycle;

    \fill[red!18] 
        (0,0) -- (0,3.8) -- (2.6,3.5) -- (3.2,2.2) -- (2.2,0) -- cycle;

    \fill[blue!18] 
        (2.6,3.5) -- (3.2,4.8) -- (5.0,4.5) -- (5.1,2.8) -- (3.2,2.2) -- cycle;

    \fill[gray!28] 
        (2.2,0) -- (3.2,2.2) -- (5.1,2.8) -- (8,2.2) -- (8,0) -- cycle;

    \fill[green!18] 
        (2.4,6) -- (8,6) -- (8,2.2) -- (5.1,2.8) -- (5.0,4.5) -- (3.2,4.8) -- cycle;

    \draw[thick] (0,3.8) -- (2.6,3.5);
    \draw[thick] (2.6,3.5) -- (3.2,4.8);
    \draw[thick] (3.2,4.8) -- (2.4,6);

    \draw[thick] (2.6,3.5) -- (3.2,2.2);
    \draw[thick] (3.2,2.2) -- (2.2,0);

    \draw[thick] (3.2,4.8) -- (5.0,4.5);
    \draw[thick] (5.0,4.5) -- (5.1,2.8);
    \draw[thick] (5.1,2.8) -- (3.2,2.2);

    \draw[thick] (5.1,2.8) -- (8,2.2);

    \draw[thick] (0,0) rectangle (8,6);

    \filldraw[black] (1.2,4.9) circle (2.2pt) node[above left=2pt] {$z_1$};
    \filldraw[black] (1.4,1.5) circle (2.2pt) node[below left=2pt] {$z_2$};
    \filldraw[black] (4.0,3.4) circle (2.2pt) node[below=2pt] {$z_3$};
    \filldraw[black] (5.7,1.1) circle (2.2pt) node[below=2pt] {$z_4$};
    \filldraw[black] (6.0,4.8) circle (2.2pt) node[above right=2pt] {$z_5$};

    \node[gray!70!black, align=center] at (1.3,4.1) {basin of\\ignorance};
    \node[gray!70!black, align=center] at (5.8,1.9) {basin of\\ignorance};

    \node[red!60!black] at (1.5,2.7) {$g=2$};
    \node[blue!60!black] at (4.0,4.2) {$g=3$};
    \node[green!50!black!80] at (6.2,3.7) {$g=5$};

\end{tikzpicture}
\end{minipage}

\caption{Comparison of two ignorance-aware prediction rules. Panel (a) shows a tree-based rule with predictive leaves and ignorance leaves. Panel (b) shows a supervised clustering rule on covariate space, where the basin of ignorance is the union of two cells and the remaining cells form predictive groups.}
\label{fig:tree_clustering_combined}
\end{figure}

\subsection{Estimation} \label{sec:estimation}

For any candidate pair \((f,\pi)\), we estimate the loss function in \eqref{eqn:loss} by
\begin{equation}\label{eqn:R_hat}
\small 
\begin{aligned} 
\widehat R(f,\pi)
=
\sum_{x\in\mathcal X} p(x)
\left[
\left\{\big(f(x)-\hat\tau(x)\big)^2-\hat\eta(x)^2\right\}\pi(x)
+
\hat{c}(x)\bigl(1-\pi(x)\bigr)
\right], 
\end{aligned} 
\end{equation}
where, as for $\hat{\tau}(x)$ and $\hat{\eta}^2(x)$, also $\hat{c}(x)$ is an unbiased but possibly noisy estimator of $c(x)$. 
Our correction by \(\hat\eta(x)^2\) guarantees that
$ 
\mathbb E\!\left[\big(f(x)-\hat\tau(x)\big)^2-\hat\eta(x)^2\right]
=
\big(f(x)-\tau(x)\big)^2,
$ 
so that \(\widehat R(f,\pi)\) is unbiased for \(R_c(f,\pi)\). Intuitively, the decision compares the squared error of the candidate approximation $f(x)$ relative to $\hat{\tau}(x)$ to the within-type variance $\eta(x)^2 = \mathbb{V}(\hat{\tau}(x))$;  types are more likely to be assigned to the basin of ignorance when the prediction error of \(f(x)\) (which depends on the bias of $f(x)$ to approximate $\tau(x)$) is too large to be explained by the variance of $\hat{\tau}(x)$ alone. 

The final estimator chooses a prediction \(\hat{f}\) and an ignorance set
$ 
\{x\in\mathcal X:\hat\pi(x)=0\},
$   
\begin{equation}\label{eqn:erm}
\small 
\begin{aligned} 
(\hat{\pi},\hat f)
\in
\arg\min_{\pi\in\Pi, f \in \mathcal{F}}
\widehat R(f,\pi).
\end{aligned} 
\end{equation}

Data-driven choices for also selecting the complexity of $\mathcal{F}$ are in Section~\ref{sec:model_selection}. 


\section{Statistical guarantees and model selection} \label{sec:theory}

Next, we study the properties of the estimator in terms of the regret, i.e., the distance of the loss under our estimator from  the smallest achievable loss  \citep{manski2004statistical, kitagawa2018should}. We begin with an arbitrary class \(\mathcal F\) of bounded complexity. We then specialize to discrete partitions where we study minimax lower bounds, and finally study model selection over a sequence of classes \(\{\mathcal F_j\}\) of increasing complexity. 

Throughout our analysis, we will impose the following moment conditions.

\begin{ass}[Sampling and moments] \label{ass:general_regret}
For each \(x\in\mathcal X\), researchers observe
$ 
\big(\hat\tau(x),\hat\eta(x)^2,\hat c(x)\big)\sim_{i.n.i.d.}\mathcal D_x,
$ 
independently across \(x\), where \(\mathcal D_x\) may be unknown. The following hold:
\begin{itemize}
    \item[(i)] For all \(x\in\mathcal X\),
    $ 
    \mathbb E[\hat\tau(x)] = \tau(x),
    \mathbb E[\hat\eta(x)^2]  = \mathbb V(\hat\tau(x)),
    \mathbb E[\hat c(x)] = c(x).
    $

    \item[(ii)] For any \(u\in[0,1/2]\), there exists a constant \(0 < M_u<\infty\) such that, for all \(x\in\mathcal X\),
\[
\small 
\begin{aligned} 
\max\Big\{
\sqrt{\mathbb E\big[|\hat\tau(x)-\mathbb{E}[\hat\tau(x)]|^{4-4u}\big]},
\ \sqrt{\mathbb E\big[|\hat\tau(x)-\mathbb{E}[\hat\tau(x)]|^{6}\big]},
\ \mathbb E\big[|\hat\eta(x)^2-\eta(x)^2|^{2-2u}\big],
\end{aligned} 
\]
\[
\small 
\begin{aligned} 
\qquad
\mathbb E\big[|\hat\eta(x)^2-\eta(x)^2|^{3}\big],
\ \mathbb E\big[|\hat c(x)-c(x)|^{2-2u}\big],
\ \mathbb E\big[|\hat c(x)-c(x)|^{3}\big]
\Big\}
\le M_u, 
\end{aligned} 
\]
 with $|\tau(x)| \le B$ bounded by a finite constant $B \in [1,\infty)$. 
    \item[(iii)] There exist a constant $\bar p \in (0,\infty)$ such that, for all \(x\in\mathcal X\),
    $ 
    p(x) \le \frac{\bar p}{|\mathcal X|}.
    $ 
\end{itemize}
\end{ass}

Assumption \ref{ass:general_regret}(i) states that researchers observe, for each \(x\), an unbiased but possibly noisy estimate of \(\tau(x)\), \(\eta(x)^2\), and the abstention cost \(c(x)\). Randomness in the estimator may come from both sampling and treatment assignment, while randomness in \(\hat c(x)\) may reflect uncertainty about the cost of collecting further evidence, see Section \ref{sec:cost_abstention}. These objects do not need to be consistent as \(|\mathcal X|\) grows.  

Assumption \ref{ass:general_regret}(ii) is a mild tail condition for the estimators 
\(\hat\tau(x)\), \(\hat\eta(x)^2\), and \(\hat c(x)\). It requires uniformly bounded sixth moments. Because lower order moments can be bounded with bounds on higher moments, in settings with noisy (inconsistent) $\hat{\tau}(x)$, our leading case, $M_u \le \bar{M}$ for a uniformly bounded constant $\bar{M}$ provided that the estimators have bounded sixth moment; In Example \ref{exmp:ipw}, this condition is attained under standard overlap conditions and sixth moment bound of the potential outcomes.

We state Assumption \ref{ass:general_regret}(ii) also in terms of lower-order moments because in settings where initial estimates are more precise (for example, because many observations are available within each type \(x\), e.g., large $n(x)$ in Example \ref{exmp:ipw}) $M_u$ also captures the precision of each estimator, for arbitrary small $u$. This would reflect in sharper bounds for the regret, which therefore also allows for precise $\hat{\tau}$ when available.  Condition (ii) also requires that the population $\tau(x)$ is uniformly bounded to avoid infinite values.\footnote{The condition that $B \ge 1$ is without loss, since the bound trivially holds when $|\tau(x)| \le 1$.}

Assumption \ref{ass:general_regret}(iii)  states that no single unit receives disproportional large weight.

\paragraph{Example \ref{exmp:ipw} continued.}
Our framework encompasses two possible cases: 
\begin{itemize} 
\item Large $|\mathcal{X}|$, i.e., many imprecise $\hat{\tau}(x)$:
Suppose researchers collect a sample size of size $2 N$ and construct $N$ estimates of $(\hat{\tau}(x), \hat{\eta}(x)^2)$, with two units having (approximately) the same value of $x$ ($n(x) = 2$). 
Suppose that there exist constants \(C<\infty\) and \(\delta\in(0,1)\) such that, almost surely,
$ 
\max\Big\{
\mathbb E\big[|Y_i(1)|^6\mid X_i\big],
\mathbb E\big[|Y_i(0)|^6\mid X_i\big]
\Big\}\le C,
\delta \le e(X_i)\le 1-\delta,
$ 
and that \(c(x)\) is known, so that \(\hat c(x)=c(x)\) for simplicity. Then Assumption \ref{ass:general_regret}(ii) holds with \(M_u\le \bar M\) for all \(u\in[0,1/2]\), where \(\bar M<\infty\) is a finite constant. In this example $N_{\mathcal{X}}:=|\mathcal{X}| = N$. 

\item More precise $\hat{\tau}(x)$: If instead $n(x) \ge \underline{n} \ge 2$ for all $x \in \mathcal{X}$, then we can find a constant $\bar{M}' < \infty$, so that $M_u \le \bar{M}'\underline{n}^{-(1 - u)}$ for any $u\in [0,1/2]$.\footnote{Standard moment bounds for sample means imply
$ 
\mathbb E\big[|\hat\tau(x)-\tau(x)|^4\big]\le \bar{M}\,\underline n^{-2},
\mathbb E\big[|\hat\tau(x)-\tau(x)|^6\big]\le \bar{M}\,\underline n^{-3},
$ 
for a finite constant \(\bar{M}\) depending only on \(C\) and \(\delta\). Hence, for every \(u\in[0,1/2]\),
$ 
\sqrt{\mathbb E\big[|\hat\tau(x)-\tau(x)|^{4-4u}\big]}
\le
\big(\mathbb E[|\hat\tau(x)-\tau(x)|^4]\big)^{(1-u)/2}
\le
\bar{M}\,\underline n^{-(1-u)}.
$ 
Likewise, since \(\hat\eta(x)^2\) is the usual sample-variance estimator divided by \(n(x) \ge \underline n\), and the squared pseudo-outcomes have uniformly bounded third moment, standard bounds give
$ 
\mathbb E\big[|\hat\eta(x)^2-\eta(x)^2|^2\big]\le \bar{M}\,\underline n^{-3},
\mathbb E\big[|\hat\eta(x)^2-\eta(x)^2|^3\big]\le \bar{M}\,\underline n^{-9/2}.
$ Therefore,
$ 
\mathbb E\big[|\hat\eta(x)^2-\eta(x)^2|^{2-2u}\big]
\le
\big(\mathbb E[|\hat\eta(x)^2-\eta(x)^2|^2]\big)^{1-u}
\le
\bar{M}\,\underline n^{-3+3u}
\le
\bar{M}\,\underline n^{-(1-u)}.
$ 
}  In this case $M_u$ is smaller for larger effective sample size for estimating each individual $\hat{\tau}(x)$.
\end{itemize} 

We impose no restrictions  on $\tau(x), \eta(x)^2$ other than being uniformly bounded.
\qed 

\medskip

\begin{ass}[Bounded complexity of \(\Pi\)] \label{ass:bounded_Pi}
Assume that the class of abstention rules \(\Pi\) has VC dimension at most \(v_\Pi<\infty\).\footnote{The VC dimension of a class is the size of the largest set of points it can shatter. It is a standard measure of the capacity of a class \citep{devroye2013probabilistic}.}
\end{ass}

Assumption \ref{ass:bounded_Pi} requires that the complexity of the basin of ignorance, as induced by \(\Pi\), is finite. This rules out highly irregular abstention regions and holds for a broad range of interpretable specifications, including trees as in Figure \ref{fig:tree_clustering_combined}(Panel (a)), clustering as in Figure \ref{fig:tree_clustering_combined}(Panel (b)) (see Section \ref{sec:clustering}), and maximum score rules \citep{zhou2023offline, kitagawa2018should, mbakop2021model}. Assumption \ref{ass:bounded_Pi} ensures that the basin of ignorance is interpretable and avoids overfitting.

\subsection{Regret guarantees for general models with abstention}

The first result characterizes the properties of our estimator relative to the oracle choice of basin of ignorance and estimator in Equation \eqref{defn:generalizable_predictions}. To characterize the complexity of the prediction function class $\mathcal{F}$ (without restrictions on the complexity of the target estimand $\tau$) we introduce the following assumption. 

\begin{ass}[Bounded complexity prediction function-class] \label{ass:entropy_F} Let $\mathcal{F}$ a function class with $\sup_{f \in \mathcal{F}, x \in \mathcal{X}}|f(x)| \le K$ for a constant $K \in (0,\infty)$. 
For any probability measure \(Q\) on \(\mathcal X\), define for two functions $f,g \in \mathcal{F}$
$ 
\|f-g\|_{Q,1}
:=
\int_{\mathcal X} |f(x)-g(x)|\,dQ(x),
$ 
and let \(\mathcal{N}(\varepsilon,\mathcal F,L_1(Q))\) denote the \(\varepsilon\)-covering number of \(\mathcal F\) under \(\|\cdot\|_{Q, 1}\).\footnote{The covering number \(\mathcal{N}(\varepsilon,\mathcal F,L_1(Q))\) is the smallest number of \(L_1(Q)\)-balls of radius \(\varepsilon\) needed to cover \(\mathcal F\); equivalently, it measures how many functions are needed to approximate every element of \(\mathcal F\) up to error \(\varepsilon\). It is a common measure of complexity, see for instance \cite{devroye2013probabilistic}. It is possible to replace conditions on the covering number with weaker conditions on the conditional Rademacher complexity \citep{bartlett2002rademacher}, omitted for expositional convenience.} Assume there exists a constant \(C(\mathcal{F})<\infty\) such that
$ 
\sup_{Q\in\mathcal Q}
\int_0^{2K}
\sqrt{
\log \mathcal{N}\!\left(\varepsilon,\mathcal F,L_1(Q)\right)
}\,d\varepsilon
\le
\sqrt{C(\mathcal{F})},
$ 
where \(\mathcal Q\) denotes the collection of finitely supported probability measures on \(\mathcal X\).
\end{ass}

Assumption \ref{ass:entropy_F} imposes two restrictions on the \textit{prediction} function class \(\mathcal F\) (but, importantly, not on the population function $\tau(x)$). First, the uniform bound rules out pathological specifications with arbitrarily large fitted values, a standard restriction in empirical process arguments \citep{devroye2013probabilistic}.  
The second condition instead restricts the complexity of $\mathcal{F}$ to both control the statistical error of the estimator. 

The covering number is a common measure of complexity \citep{devroye2013probabilistic}. 
Examples satisfying the complexity bound in Assumption \ref{ass:entropy_F} include the partition classes in Example \ref{exmp:supervised_clustering}, the sparse linear classes in Example \ref{exmp:economic_example}, forest-based methods, and low-dimensional parametric families under mild regularity conditions \citep[e.g., Sections 2.6-2.7][]{van1996weak}. We provide examples in Sections \ref{sec:clustering}, \ref{sec:regression_models}. 

In our first theorem we characterize the regret in full generality and then specialize its implication in our leading Examples \ref{exmp:supervised_clustering} and \ref{exmp:economic_example}. 

\begin{thm} \label{thm:regret} Let Assumptions \ref{ass:general_regret}, \ref{ass:bounded_Pi}, \ref{ass:entropy_F} hold. Denote $N_{\mathcal{X}} = |\mathcal{X}|$ the number of types $x$. Then for any $u \in (0,1/2]$
$$
\small 
\begin{aligned} 
\mathbb{E}\left[R_c(\hat{f}, \hat{\pi}) - \min_{\pi \in\Pi, f \in \mathcal{F}} R_c(f, \pi) \right] \le c_0 \sqrt{\tilde{M}_u \frac{(C(\mathcal{F}) + v_\Pi)}{ N_{\mathcal{X}}} }
\end{aligned} 
$$
for a constant $c_0 \le q_0 \bar{p} \max\{1,K\} B^{3/2}$ with universal constant $q_0 < \infty$, and $\tilde{M}_u = \frac{(M_u + M_u^2)}{u^2}$.  
\end{thm} 

\begin{proof} See Appendix \ref{proof:thm:regret}.  
\end{proof}

Theorem \ref{thm:regret} gives a finite-sample regret guarantee for the ignorance-aware estimator relative to the oracle rule that optimally chooses both the prediction function and the basin of ignorance. The result holds uniformly over all data-generating processes satisfying Assumption \ref{ass:general_regret}. In particular, beyond the moment conditions above, the theorem imposes no restriction on \(\tau(x)\) or $\eta(x)$.

The bound depends on a constant \(c_0\) and three main components. The first component decreases at rate \( N_{\mathcal{X}}^{-1/2}\), where we think of $N_{\mathcal{X}}$ as proportional to the overall sample size rate in our leading examples where the potential heterogeneity $|\mathcal{X}|$ is large. For instance, in Example \ref{exmp:ipw} with $n(x) = 2$, $N_{\mathcal{X}}$ corresponds to half the overall sample size.  The second component depends on the moment bound
$ 
\tilde M_u=\frac{M_u+M_u^2}{u^2},
u\in(0,1/2].
$

Because the bound holds for any $u$, different values of $u$ correspond to different terms $\tilde M_u$. For example, taking \(u=1/2\), the regret depends on a uniformly bounded \textit{constant} \(\tilde M_{1/2} < \infty\) that captures a mild bound on the sixth moments of the relevant estimators (see Example \ref{exmp:ipw} continued, first case). This leading case corresponds to the regret converging at rate $N_{\mathcal{X}}^{-1/2}$. When more units are observed within each type $x$, improving the precision of each estimate, $\tilde{M}_u$ also captures gains from more precise estimates. In this case, we can take smaller values of \(u\) (e.g., $u = \varepsilon > 0$ for a small positive $\varepsilon$), so that the bound involves essentially a second moment of the estimators (see Example \ref{exmp:ipw} continued, second case). In this sense, the result is general as it accommodates both settings with many noise estimates for $x$ (high dimensional $|\mathcal{X}|$) as in our leading case and settings with fewer but more precisely estimated types. 

Finally, the complexity term enters through \(C(\mathcal F)+v_\Pi\): \(C(\mathcal F)\) measures the richness of the prediction class, while \(v_\Pi\) measures the complexity of the basin of ignorance.

Unlike standard heterogeneity procedures \citep[e.g.][]{chernozhukov2018generic, athey2018generalizedrandomforests}, where evidence aggregation is imposed globally across all types $x$ through restrictions on the prediction function, our framework uses such restrictions only where the data support reliable prediction. The estimator must therefore jointly learn both a prediction rule and the region in which that rule should not be applied. The proof accordingly requires controlling not only the estimation error in the estimated properties, but also the complexity of the induced classification problem.

  \subsection{Implications for clustering methods} \label{sec:clustering}

Next, we focus on clustering methods, given their large applicability. 

\begin{defn}[Clustering partitions] \label{defn:clustering}
Fix integers \(G\ge 2\) and let \(\mathcal G\) be a class of partition rules
$ 
\alpha:\mathcal X\to\{1,\dots,G\},
$ 
where label \(g = 1\) denotes the basin of ignorance and labels \(g\in\{2,\dots,G\}\) denote predictive groups. For each \(\alpha\in\mathcal G\), define the induced prediction rule
$ 
\pi^\alpha(x)=1\{\alpha(x)\ge 2\},
$ 
and, for any \(\mu=(\mu_2,\dots,\mu_G)\in[-K,K]^{G-1}\), and arbitrary finite constant $K$, define the prediction rule
$ 
f_{\alpha,\mu}(x)=\sum_{g=2}^G \mu_g\,1\{\alpha(x)=g\}.
$ 
We denote by
$ 
\Pi =\{\pi^\alpha:\alpha\in\mathcal G\}$, and
$\mathcal F=\{f_{\alpha,\mu}:\alpha\in\mathcal G,\ \mu\in[-K,K]^{G-1}\},
$ 
the induced classes of abstention rules and prediction rules.

We consider the following class of partitions: 
\begin{itemize}
    \item[(i)] For every \(\alpha\in\mathcal G\) and every \(g\in\{2,\dots,G\}\), either
    $  
    \sum_{x\in\mathcal X}1\{\alpha(x)=g\}=0
    $ 
    or
    $ 
    \sum_{x\in\mathcal X}1\{\alpha(x)=g\}\ge \underline\kappa |\mathcal X|
    $ 
    for some constant \(\underline\kappa \in (0,1]\).

    \item[(ii)] For each \(g\in\{1,\dots,G\}\), the indicator class
    $ 
    \Bigl\{x\mapsto 1\{\alpha(x)=g\}:\alpha\in\mathcal G\Bigr\}
    $ 
    has VC dimension at most \(v<\infty\).
\end{itemize}
\end{defn}

Definition \ref{defn:clustering} imposes two main restrictions on clustering classes. First, each predictive group outside the basin of ignorance must be empty or contain at least a fraction \(\underline\kappa\) of the observable types. This condition is intuitive: a group containing only one or a few different types should not be interpreted as part of an aggregated summary. Second, the indicator classes \(x\mapsto 1\{\alpha(x)=g\}\) must have bounded VC dimension. This controls the geometric complexity of both the predictive groups and the basin of ignorance, since label \(g=1\) corresponds to abstention. It improves interpretability of the clustering and control the statistical error that would otherwise lead to overfitting. The bounded VC complexity is attained both by tree-based partitions and $G$-means clustering over the space of covariates, as we discuss below. 
Finally, we also require that the predictions $\mu$ are uniformly bounded to avoid pathological cases. 

\begin{exmp}[Tree-based classes] \label{exmp:tree_entropy}
Consider a decision rule as in Figure \ref{fig:tree_clustering_combined}(Panel (a)). That is, suppose \(x\in\mathbb R^d\) and the partition rule \(\alpha\) is induced by a tree: each predictive group indicator
$ 
x\mapsto 1\{\alpha(x)=g\}
$ 
belongs to the class of axis-aligned rectangles in \(\mathbb R^d\), whose VC dimension is at most \(2d\). By letting the basin of ignorance to be the union of at most \(G - 1\) leaves as in Figure \ref{fig:tree_clustering_combined}(a), then the class of basin indicators
$ 
x\mapsto 1\{\alpha(x)=1\}
$ 
has VC dimension of order \(dG\log G\), satisfying Definition \ref{defn:clustering}(ii).\footnote{The class of axis-aligned rectangles in \(\mathbb R^d\) has VC dimension at most \(2d\) \citep{blumer1989learnability}. Therefore, if each predictive group corresponds to a single terminal leaf, then for each \(g\ge 2\),
$
\Bigl\{x\mapsto 1\{\alpha(x)=g\}:\alpha\in\mathcal G\Bigr\}
$ 
has VC dimension at most \(2d\).
The basin of ignorance consists of unions of at most \(G-1\) sets from a base class of VC dimension \(2d\) \citep[][Page 220]{devroye2013probabilistic}. By the standard VC bound for unions of \(G-1\) sets from a VC class of dimension \(v\), the resulting class has VC dimension of order \(v G\log G\), where $v = 2d$ \citep[Lemma 3.2.3][]{blumer1989learnability}. } 
\end{exmp}

\begin{exmp}[\(G\)-means on covariates]
\label{exmp:kmeans_entropy}
Consider the decision rule in Figure \ref{fig:tree_clustering_combined} (Panel (b)).
Formally, let \(x\in\mathbb R^d\), and let \(z=(z_2,\dots,z_G)\in(\mathbb R^d)^{G-1}\) be a collection of centroids for the predictive groups.  
Given \(z\), define 
$
j_z(x):=\arg\min_{g\in\{2,\dots,G\}}\|x-z_g\|^2,
$ 
and write
$ 
V_g(z):=\{x:j_z(x)=g\}, g=2,\dots,G.
$ 
For each fixed \(g\), \(V_g(z)\) is an intersection of \(G-2\) halfspaces.\footnote{Formally,
$ 
V_g(z)
=
\bigcap_{h\neq g}
\left\{
x\in\mathbb R^d:
\|x-z_g\|^2\le \|x-z_h\|^2
\right\}.
$} 
Standard VC bounds for intersections of halfspaces \citep{blumer1989learnability} imply that the class
$ 
\Bigl\{x\mapsto 1\{x\in V_g(z)\}: z\in(\mathbb R^d)^{G-1}\Bigr\}
$ 
has VC dimension of order \(dG\log G\). 
To allow for abstention, let \(S\subseteq\{2,\dots,G\}\) denote the collection of cells assigned to the basin of ignorance, and define
\begin{equation} \label{eqn:S}
\small 
\begin{aligned} 
\alpha_{z,S}(x)
=
\begin{cases}
1, & j_z(x)\in S,\\
j_z(x), & j_z(x)\notin S.
\end{cases}
\end{aligned} 
\end{equation} 
Under this definition,
$ 
\{x:\alpha_{z,S}(x)=1\}
=
\bigcup_{g\in S} V_g(z).
$ 
Since the basin of ignorance is the union of at most \(G-1\) such groups, the class of basin of ignorance $1\{\alpha(x) = 1\}$
has VC dimension of order \(dG^2\log^2 G\) by Lemma 3.2.3 in \cite{blumer1989learnability}. This shows that the $G$-means partition over the covariate space satisfies Definition \ref{defn:clustering}(ii). 
\end{exmp}

The examples above illustrate how the VC-dimension of the class assignment captures the complexity of the partition.

\begin{thm}[Regret bound for clustering partitions] \label{thm:clustering_regret}
Let Assumption \ref{ass:general_regret} hold, and suppose \(\mathcal G\) is a class of clustering partitions in the sense of Definition \ref{defn:clustering}. Let
\[
\small 
\begin{aligned} 
(\hat\alpha,\hat\mu)\in
\arg\min_{\alpha\in\mathcal G,\ \mu\in[-K,K]^{G-1}}
\widehat R\!\left(f_{\alpha,\mu},\pi^\alpha\right),
\end{aligned} 
\]
and define
$ 
\hat f:=f_{\hat\alpha,\hat\mu},
\hat\pi:=\pi^{\hat\alpha}.
$ 
Denote $N_{\mathcal{X}} = |\mathcal{X}|$ the number of types $x$.
Then for any \(u\in(0,1/2], G \ge 2\),
\[
\small 
\begin{aligned} 
\mathbb E\!\left[
R_c(\hat f,\hat\pi)
-
\inf_{\alpha\in\mathcal G,\ \mu\in[-K,K]^{G-1}}
R_c\!\left(f_{\alpha,\mu},\pi^\alpha\right)
\right]
\le
c_0 
\sqrt{
\tilde M_u\,
\frac{(v+1)\log G}{\underline\kappa\,N_{\mathcal{X}}}
},
\end{aligned} 
\]
for $c_0 \le \bar{p} \max\{1, K^2\} B^{3/2} q_0$ for a universal constant $q_0 < \infty$ and
 \(\tilde M_u=(M_u+M_u^2)/u^2\).
\end{thm}

\begin{proof}
See Appendix \ref{proof:thm:clustering_regret}.
\end{proof}

Theorem \ref{thm:clustering_regret} shows that our may results hold for discrete partitions, with the rate corresponding to the one in Theorem \ref{thm:regret} for finite $G$. The result also clarifies the value of combining a low-complexity partition with a basin of ignorance. If the number of predictive groups is large (where typically $\underline{\kappa} \asymp 1/G$), or if groups are allowed to become arbitrarily small, the partition can overfit the estimates, with the bound becoming less informative. By restricting attention to a small number of sufficiently large groups, and through complexity restrictions on the basin of ignorance, we collect the types for which evidence cannot be aggregated while avoiding overfitting. 

We formalize this intuition in the following theorem where we provide a lower bound for trees and $G$-clustering presented in the examples above.

\begin{thm}[Lower bound for $G$-means partitions]
\label{thm:gmeans_lower_union_ignorance}
Fix an integer \(G\ge 4\), \(d\ge 1\), and \(K_0\in(0,K]\). 
Fix $\underline{\eta}$ be a finite constant with
$ 
0 < \underline\eta \le \frac{K_0}{32}.  
$ 
Let
\(\mathcal G\) denote either the nearest-centroid partition class in
Example~\ref{exmp:kmeans_entropy} or the class of trees in Example \ref{exmp:tree_entropy}, with 
$ 
0<\underline\kappa\le \frac{1}{4(G-1)}, 
$ 
and $G$ as in Definition \ref{defn:clustering}. In the tree case, each leaf node either corresponds to an individual archetype or to the basin of ignorance.  
Then for some universal constants \(c_0,C_0>0\), for every integer
$ 
N_{\mathcal X}\ge
C_0\max\left\{
(G-1)^2,\,
\frac{\underline\eta^2}{K_0^2}(G-1)^3
\right\},
$ 
there exists a fixed design
$ 
\mathcal X=(x_1,\dots,x_{N_{\mathcal X}})\in (\mathbb R^d)^{N_{\mathcal X}}
$ 
of \(N_{\mathcal X}\) distinct points and $p(x_i) = \frac{1}{N_{\mathcal{X}}}$ such that 
for every measurable estimator \((\widehat\alpha,\widehat\mu) \in \mathcal G \times [-K,K]^{G-1}\), 
there exists a data-generating process $P$ with \((\hat{\tau}(x),\hat{\eta}(x)^2, \hat{c}(x))_{x \in \mathcal{X}} \sim P\) satisfying
Assumption~\ref{ass:general_regret}, with
$ 
\eta(x)^2=\underline\eta^2$ and $c(x)=K_0^2$
for all $x \in\mathcal X$,
such that 
\begin{equation} \label{eqn:lower_bound} 
\small 
\begin{aligned} 
\mathbb E_P\!\left[
R_c\!\left(f_{\widehat\alpha,\widehat\mu},\pi^{\widehat\alpha}\right)
-
\inf_{\alpha\in\mathcal G,\,
\mu\in[-K,K]^{G-1}}
R_c\!\left(f_{\alpha,\mu},\pi^\alpha\right)
\right]
\ge
c_0\,K_0\,\underline\eta
\sqrt{\frac{G-2}{N_{\mathcal X}}}.
\end{aligned} 
\end{equation} 
\end{thm}

\begin{proof} See Appendix \ref{proof:lower_bound}.  
\end{proof}

In the leading case with many noisy estimates $\hat\tau(x)$, Theorem \ref{thm:gmeans_lower_union_ignorance} provides a $N_{\mathcal{X}}^{-1/2}$ minimax rate, matching the upper bound, for both the tree method and the $G$-means clustering method with abstention with fixed $G$. (In addition, the dependence on the variance $\underline\eta$ enters the bound similarly to $\sqrt{M_u}$ in the upper bound.)

The lower bound also depends on the complexity of the clustering (and of the basin of ignorance), which in Examples \ref{exmp:tree_entropy} and \ref{exmp:kmeans_entropy} depends on $G$. Intuitively, a larger number of predictive groups enlarges the set of partitions and as such the flexibility of both the prediction function class and basin of ignorance. The bound therefore clarifies that arbitrary clustering methods with a large number of clusters may lead to overfitting. In contrast, by imposing complexity restrictions on the partition function, we can control statistical error, achieving minimax rates of convergence on the regret.

\subsection{Implications for other common regression models} \label{sec:regression_models}

Our results also apply to function classes beyond partitions. For instance, consider linear prediction rules of the form
$
f_\theta(x)=q(x)^\top\theta.
$ 
Under a uniform bound condition on the regressors and coefficients, the entropy complexity \(\sqrt{C(\mathcal F)}\) in Assumption~\ref{ass:entropy_F} is of order \(O(\sqrt{d})\), where \(d\) denotes the dimension of \(q(x) \in \mathbb{R}^d\), see Appendix Lemma \ref{lem:linear_entropy}. Thus, Theorem~\ref{thm:regret} delivers regret bounds for low-dimensional linear summaries.

In high-dimensional settings, Appendix Theorem~\ref{thm:l0_sparse} shows that sharper bounds are available under sparsity restrictions, as in Example~\ref{exmp:economic_example}. In particular, for Example~\ref{exmp:economic_example} we obtain the analogue of Theorem~\ref{thm:regret} with
$
\sqrt{C(\mathcal F)} \le c_0' \sqrt{s\log(d)},
$ 
for a finite constant \(c_0'<\infty\), where \(s\) denotes the number of nonzero coefficients and \(d\) is the total number of available regressors. This yields the familiar reduction in statistical complexity relative to unrestricted high-dimensional linear models.\footnote{We note that for lasso estimator ($l_1$ instead of $l_0$ sparsity), the Dudley's entropy integral bound scales at rate $\sqrt{d}$. It is possible to replace restrictions on the Dudley's entropy integral in Assumption \ref{ass:entropy_F} with weaker requirements on the conditional Rademacher complexity of the function, which for lasso would instead scale at a rate $\sqrt{\log(d)}$ \citep{bartlett2002rademacher} for fixed $l_1$-sparsity. }

A similar logic applies to random forests. For example, standard covering-number calculations together with Appendix Lemma~\ref{lem:devroye_sum} imply that, for a forest with \(T\) trees, a  crude bound for 
$
\sqrt{C(\mathcal F)}$ is (up to constant terms)  
$\sqrt{T\,C(\mathcal F_{\mathrm{tree}})},
$ 
where \(C(\mathcal F_{\mathrm{tree}})\) denotes the complexity of the underlying regression-tree class.\footnote{Sharper complexity bounds are available by exploiting the convex-hull structure \citep{bartlett2002rademacher}, omitted for brevity.} 

\subsection{Model selection} \label{sec:model_selection}

In many applications, researchers may be interested in choosing the model class $\mathcal{F}$ from the data, such as choosing the depth of a tree-based rule. In this section we show that introducing the possibility of abstention we can improve statistical guarantees for model selection through a notion of \textit{ignorance-free} approximation error, i.e., an approximation error only within the predictive set. Let
$$ 
\small 
\begin{aligned} 
\mathcal F_1, \mathcal F_2, \cdots, \mathcal F_J, \qquad \Pi_1, \cdots, \Pi_J
\end{aligned} 
$$ 
be a sequence of function classes (of possibly increasing flexibility), with corresponding function classes for prediction sets as $\Pi_j$ that (in full generality) may or may not vary with $j$. For each prediction class \(\mathcal F_j\), let \(C(\mathcal F_j)\) denote the corresponding complexity bound in Assumption \ref{ass:entropy_F}.  
In Example \ref{exmp:supervised_clustering} this may depend on the  number of predictive groups, and in Example \ref{exmp:economic_example}, this may depend on the number of non-zero coefficients. 

For \(\widehat R(f,\pi)\) as defined in \eqref{eqn:R_hat}, let \((\hat \pi_j,\hat f_j)\) the ignorance-aware estimator obtained by restricting the prediction function to the class \(\mathcal F_j\), and $\pi_j^\star, f_j^\star$ their oracle counterpart, 
\[
\small 
\begin{aligned} 
(\hat \pi_j,\hat f_j)
\in
\arg\min_{\pi\in\Pi_j,\ f\in\mathcal F_j} \widehat R(f,\pi), \qquad (\pi_j^\star,f_j^\star)\in\arg\min_{\pi\in\Pi_j,\ f\in\mathcal F_j}R_c(f,\pi). 
\end{aligned} 
\]

To select the class complexity \(j\), we use an independent out-of-sample evaluation. Specifically, suppose that for each type \(x\) we observe independent validation estimates
\[
\small 
\begin{aligned} 
\big(\hat\tau_{oos}(x),\hat\eta_{oos}(x)^2,\hat c_{oos}(x)\big),
\end{aligned} 
\]
independent of the training estimates \(\big(\hat\tau(x),\hat\eta(x)^2,\hat c(x)\big)\). For the IPW estimator of Example \ref{exmp:ipw}, it is possible to construct such estimators when four independent units for each type $x$ are available to researcher; see Remark \ref{sec:sample_splitting}. We estimate 
\[
\small 
\begin{aligned} 
\widehat R_{oos}(f,\pi)
=
\sum_{x\in\mathcal X} p(x)
\left[
\left\{\bigl(f(x)-\hat\tau_{oos}(x)\bigr)^2-\hat\eta_{oos}(x)^2\right\}\pi(x)
+
\hat c_{oos}(x)\bigl(1-\pi(x)\bigr)
\right].
\end{aligned} 
\]
We then select the model class by out-of-sample loss minimization:
\[
\small 
\begin{aligned} 
\hat j
\in
\arg\min_{1\le j\le J}\widehat R_{oos}(\hat f_j,\hat\pi_j),
\qquad
(\hat\pi,\hat f):=(\hat\pi_{\hat j},\hat f_{\hat j}).
\end{aligned} 
\]

\paragraph{Theoretical guarantees} Our goal is to study the trade-off between abstention and approximation error by comparing our procedure to an oracle that observes \(\tau(x)\) and is allowed to use the \textit{richest} prediction class. Specifically, we study  
\[
\small 
\begin{aligned} 
\mathcal R_c
:=
\mathbb E\!\Big[
R_c(\hat f,\hat \pi)
-
\min_{f\in \bar{\mathcal{F}}}
\sum_{x\in\mathcal X} p(x)\bigl(f(x)-\tau(x)\bigr)^2
\Big],\qquad \bar{\mathcal{F}} = \cup_{j=1}^J \mathcal{F}_j
\end{aligned} 
\]
that is, the performance gap of the ignorance-aware estimator relative to the best no-abstention predictor in the most flexible class defined by the union of all feasible function classes. This benchmark makes explicit the extent to which abstention can balance statistical error and approximation error for using a simpler prediction rule.

To characterize this comparison, for any abstention rule \(\pi\), define
\[
\small 
\begin{aligned} 
L_{\mathcal{F}_j}^{\mathrm{loc}}(\pi)
:=
\min_{f\in\mathcal F_j}
\sum_{x\in\mathcal X} p(x)\bigl(f(x)-\tau(x)\bigr)^2\pi(x),
\qquad
A_c(\pi)
:=
\sum_{x\in\mathcal X} p(x)c(x)\bigl(1-\pi(x)\bigr),
\end{aligned} 
\]
where \(L_j^{\mathrm{loc}}(\pi)\) is the smallest approximation error achievable by class \(\mathcal F_j\) \textit{local} to the prediction region \(\{x:\pi(x)=1\}\), and \(A_c(\pi)\) is the corresponding cost of abstention on the complement. We also define
\[
\small 
\begin{aligned} 
\Delta_j^{\mathrm{loc}}(\pi)
:=
L_{\mathcal{F}_j}^{\mathrm{loc}}(\pi)-L_{\bar{\mathcal{F}}}^{\mathrm{loc}}(\pi),
\end{aligned} 
\]
which measures the approximation error of the smaller class \(\mathcal F_j\) relative to the richest class \(\cup_{j=1}^J \mathcal F_j\), localized to the region $\{x: \pi(x) = 1\}$.

\begin{thm}[Model selection] \label{thm:model_selection}
Consider independent training sample \((\hat\tau(x),\hat\eta(x)^2,\hat c(x))_{x\in\mathcal X}\) and validation sample \((\hat\tau_{oos}(x),\hat\eta_{oos}(x)^2,\hat c_{oos}(x))_{x\in\mathcal X}\), both satisfying Assumption \ref{ass:general_regret}.
Let each $\mathcal{F}_j$ satisfying Assumption \ref{ass:entropy_F} with corresponding complexity $C(\mathcal{F}_j)$, and each \(\Pi_j\) satisfying Assumption \ref{ass:bounded_Pi} with VC-complexity $v_{\Pi_j}$. 

Then for any $u \in (0,1/2]$, for $N_{\mathcal{X}} = |\mathcal{X}|$ denoting the number of types,  
$$
\small 
\begin{aligned} 
\mathcal{R}_c \le \underbrace{c_0 J^{1/2} \sqrt{\frac{M_0 + M_0^2}{N_{\mathcal{X}}}}}_{\text{model selection}} + \inf_{1 \le j \le J} \Big\{\underbrace{ c_0   \sqrt{\frac{\tilde{M}_u(C(\mathcal{F}_j) + v_{\Pi_j})}{N_{\mathcal{X}}}}}_{\text{Estimation noise} } + \underbrace{\Delta_j^{\mathrm{loc}}(\pi_j^\star)}_{\text{approximation error on predictive set}} + \underbrace{A_c(\pi_j^\star)}_{\text{abstention cost}} \Big\}
\end{aligned} 
$$
for $c_0 \le q_0 B^{3/2} \bar{p} \max\{1, K\}$ for a universal constant $q_0 < \infty$ and $\tilde{M}_u = \frac{M_u + M_u^2}{u^2}$. 
\end{thm} 

\begin{proof} See Appendix \ref{proof:thm:model_selection}. 
\end{proof} 

Theorem \ref{thm:model_selection} compares the loss of the ignorance-aware model-selection procedure to the smallest achievable prediction loss of an oracle that is forced to predict everywhere using the richest class \(\bar{\mathcal{F}}\). The result is finite-sample and holds for any number of candidate classes \(J\) and any number of observable types \(N_{\mathcal{X}}\).

The bound contains four components. The first is a model-selection term of order \(\sqrt{J/N_{\mathcal{X}}}\), which is the cost of comparing \(J\) candidate classes. With many types $x$, and a finite number of models, this term is of order $N_{\mathcal{X}}^{-1/2}$.\footnote{From an inspection of the proof \ref{proof:thm:model_selection}, Step 5, it is possible to replace the second moment bound with concentration inequalities with exponential tail bounds \citep{devroye2013probabilistic} for sub-gaussian estimators. In this case, the dependence on $J^{1/2}$ can be improved to $\sqrt{\log(J)}$.}  The component $M_0 < \infty$ is a finite moment bound from Assumption \ref{ass:general_regret}(ii), with $u = 0$. The second is the statistical error within class \(j\), governed by the complexity of the prediction class \(\mathcal F_j\) and the complexity of the abstention class \(\Pi_j\) as in Theorem \ref{thm:regret}. Intuitively, simpler function classes lead to a smaller variance.

The third component is the \textit{ignorance-free} approximation error \(\Delta_j^{\mathrm{loc}}(\pi_j^\star)\), which measures how much is lost by using class \(j\) rather than the richest class \(\bar{\mathcal{F}}\), but only on the subset of types for which class \(j\) chooses to predict. This innovation shows how abstention may improve error on the predictive set. The fourth is the abstention cost \(A_c(\pi_j^\star)\), which captures the loss from withholding predictions.

This decomposition differs from standard model-selection problems \citep[e.g.][]{mbakop2021model}, where approximation error is evaluated globally (this would correspond here to the special case \(\Delta_j(\pi = 1)\)). Here, abstention allows us to compare classes through their approximation error only on the subset of types where prediction is actually made, while paying an explicit cost of abstention.

\begin{rem}[Sample splitting] \label{sec:sample_splitting} 
When at least four observations are available within each type \(x\), researchers can use half of the observations to construct the in-sample estimates and the remaining half to construct the out-of-sample estimates. In this case, the sample splitting is performed separately within each type.

An alternative strategy we recommend is to split the types \(x\) at random and assign different types to training and validation set. Formally, we can introduce sampling indicators \(S_x\in\{0,1\}\), where \(S_x=0\) if type \(x\) is assigned to the main sample and \(S_x=1\) if it is assigned to the auxiliary sample, with $P(S_x = 1) = p$ for constant $p$. Because our loss function averages over types \(x\), by integrating over the randomness induced by \(S_x\) we preserve unbiasedness of the empirical criterion. The corresponding guarantees continue to hold following similarly to Appendix Theorem \ref{app:estimated_p}. \qed 
\end{rem}

\section{Practical implementations}

In this section we discuss practical choices of the abstention cost and optimization.

\subsection{Cost of abstention}
\label{sec:cost_abstention}


As a first step, we interpret the cost of abstention as the continuation loss incurred when additional evidence is collected to help us providing explicit calibrations.

\paragraph{Continuation cost interpretation} Fix a prediction rule \(f\) and an abstention rule \(\pi\), and suppose that for a given researcher's report $(f,\pi)$, whenever \(\pi(x)=0\), an audience does not form a prediction \(f(x)\) but instead collects new data  to form a prediction \(f^{new}(x)\) (possibly by combining existing with new evidence). When instead \(\pi(x)=1\), the audience forms a prediction using \(f(x)\) without collecting more data. Therefore, given a researcher's report $(f,\pi)$, the audience forms a prediction $\pi(x) f(x) + (1 -\pi(x))f^{new}(x)$. Define the expected loss
\[
\small 
\begin{aligned}
\mathcal L(f,\pi)
:=
\sum_{x\in\mathcal X} p(x)\,
\mathbb E_{f^{new}}
\left[
\Big(
\tau(x)-\pi(x)f(x)-(1-\pi(x))f^{new}(x)
\Big)^2
\right], 
\end{aligned} 
\]
where the expectation is taken over the $f^{new}$ generated by the (future) audience action. 

\begin{prop}[Continuation-loss interpretation of \(c(x)\)] \label{prop:continuation_loss}
Fix \(f\) and \(\pi\). 
If
$ 
c(x)=\mathbb E_{f^{new}}\!\left[\bigl(f^{new}(x)-\tau(x)\bigr)^2\right],
$ 
then
$ 
\mathcal L(f,\pi)=R_c(f,\pi).
$ 
\end{prop}

 Proposition \ref{prop:continuation_loss} provides an interpretation of $c(x)$ as the mean-squared error of a  prediction formed after collecting more evidence (without having to restrict $f^{new}$).\footnote{The proof is immediate: because $\pi(x) \in \{0,1\}$, 
$
\mathcal{L}(f,\pi) = \sum_x p(x) \left\{ (\tau(x) - f(x))^2 \pi(x) + (1 - \pi(x))\mathbb{E}[(\tau(x) - f^{new}(x))^2]\right\}, 
$ where $c(x)$ is defined such that  $\mathbb{E}[(\tau(x) - f^{new}(x))^2] = c(x)$. }

\paragraph{Calibration} The proposition suggests a natural way to calibrate \(c(x)\). Recall that \(\eta^2(x)\) is the variance of the current estimator \(\hat\tau(x)\) for type \(x\). Suppose that for observations in the basin of ignorance the audience collects $m$ additional unbiased signals $\hat{\tau}_i(x)$ of $\tau(x)$ and construct an estimator as $\hat{\tau}^{new}(x) = \frac{1}{m+1} \sum_{i=0}^m \hat{\tau}_i(x)$, with $\hat{\tau}_0(x)$ corresponding to the estimated effect in the existing study. Therefore \(m=1\) corresponds to combining the existing evidence with a follow-up study that has the same sample size for type \(x\) as the current study, while larger values of \(m\) correspond to proportionally more type-\(x\) observations in the follow-up study. Using $f^{new} = \hat{\tau}^{new}$ would lead to a mean-squared error \(\eta^2(x)/(m+1)\). Any new prediction rule \(f^{new}(x)\) that is allowed to combine the new and existing evidence more efficiently should perform at least as well. It is therefore natural to view
$$
\small 
\begin{aligned} 
c_m(x) = \frac{\eta^2(x)}{m + 1}, \qquad \hat c_m(x)=\frac{\hat\eta^2(x)}{m +1}
\end{aligned} 
$$ 
as an upper bound on \(c(x)\) and $\hat{c}_m$ as the corresponding estimator available using existing data (we extend Theorem \ref{thm:model_selection} when using a cost upper bound in \ref{sec:upper_c}).

\paragraph{Simple benchmark choice of $m$} 
A natural benchmark is \(m=1\), with $c(x) = \eta(x)^2/2$ which corresponds to combining the current evidence with a follow-up study that has the same effective sample size for type \(x\) as the original experiment.

More generally, for given constraint $b$ on the size of the next study, researchers can estimate each policy $\hat{\pi}^{(m)}$ as a function of $m$ and report the largest $m$, and corresponding estimators, so that the budget is attained, i.e., $\sum_x m (1 - \hat{\pi}^{(m)}(x)) \le b$.\footnote{In this case, although the choice of $m$ is data-adaptive, one can directly extend our results in Section \ref{sec:theory} by providing regret bounds for a finite grid of of values $m \in \mathcal{M}$ that mimic those in Section \ref{sec:theory} up to an additional term that scales at rate $\log(|\mathcal{M}|)$ due to the enlargement of the function class of policies, $\pi^m \in \Pi, m \in \mathcal{M}$ each satisfying the budget constraint. We omit this for brevity only. One can also consider $\sum_x \kappa(x) m (1 - \hat{\pi}^{(m)}(x)) \le b$ for a given type-specific known cost $\kappa(x)$.}

\subsubsection{Data-adaptive choice via break-even analysis} \label{app:breakeven_other}

When the sample size in a future study is unknown, we recommend a data-adaptive choice of $m$ via a break-even analysis. Specifically, given a grid of plausible values of $m$, we compute the smallest value of $m$ for which the expected cost of future data collection on the selected basin of ignorance is below a benchmark loss $\omega_m$, with probability at least $95\%$. 
For a given training sample $\mathcal T$ used to construct $\hat{\pi}, \hat{f}$ as described below, we let $\omega_m(x)$ be any uniformly bounded function measurable with respect to $\mathcal{T}$, and defer an explicit recommendation in Remark~\ref{rem:omega_m}.

\paragraph{Estimation} First, for each candidate $m\in\mathcal M$, using the training sample estimate 
$$  
\small 
\begin{aligned} 
\widehat R^{(m)}(f,\pi)
&=
\sum_{x\in\mathcal X} p(x)
\left[
\left\{\bigl(f(x)-\hat\tau(x)\bigr)^2-\hat\eta(x)^2\right\}\pi(x)
+
\hat c_m(x)\bigl(1-\pi(x)\bigr)
\right],
\end{aligned} 
$$ 
and $(\hat\pi^{(m)},\hat f^{(m)})
\in
\arg\min_{\pi\in\Pi,\ f\in\mathcal F}
\widehat R^{(m)}(f,\pi)$. 

Using sample splitting as in Section~\ref{sec:model_selection}, we then take an independent validation sample. Let $\hat\eta_{oos}^2(x)$ denote the out-of-sample estimate of $\eta^2(x)$. Let 
$$  
\small 
\begin{aligned} 
\Delta_\eta(m)
&=
\sum_{x\in\mathcal X}
p(x)
\left\{
\frac{\eta^2(x)}{m+1}
-
\omega_m(x)
\right\}
\bigl(1-\hat\pi^{(m)}(x)\bigr),
\qquad  m \in \mathcal{M}, 
\end{aligned} 
$$ 
and $\widehat\Delta_\eta(m) = \sum_{x\in\mathcal X}
p(x)
\left\{
\frac{\hat\eta_{oos}^2(x)}{m+1}
-
\omega_m(x)
\right\}
\bigl(1-\hat\pi^{(m)}(x)\bigr)$ its out-of-sample estimate. 

Negative values of $\Delta_\eta(m)$ indicate that, under $\hat\pi^{(m)}$, future data collection on the basin of ignorance would yield an expected loss below the target benchmark $\omega_m$. 

\paragraph{Inference}
For any $m\in\mathcal M$, let $N_{\mathcal X}:=|\mathcal X|$ and define
\begin{equation*}
\small 
\begin{aligned} 
q_m^2
=
\mathbb V\left(
\sqrt{N_{\mathcal X}}
\left\{
\widehat\Delta_\eta(m)-\Delta_\eta(m)
\right\}
\mid \mathcal T
\right).
\end{aligned} 
\end{equation*}
Let $Q_m^2$ be a variance upper bound estimate satisfying
$ 
P(Q_m^2\ge q_m^2)\rightarrow 1
$ 
for any $m\in\mathcal M$. Formal expressions for $Q_m^2$ are in Remark~\ref{sec:inference_detailsa}. 
 
Let $\mathcal M=\{m_1,\ldots,m_{|\mathcal M|}\}$ so that $m_1<\cdots<m_{|\mathcal M|}$. For any $j\le |\mathcal M|$, define
\begin{equation} \label{eqn:S_la}
\small 
\begin{aligned} 
U_j
=
\max_{l\ge j}
\left\{
\widehat\Delta_\eta(m_l)
+
\Phi^{-1}(1-\alpha)
\frac{Q_{m_l}}{\sqrt{N_{\mathcal X}}}
\right\},
\end{aligned} 
\end{equation}
where $\Phi(\cdot)$ is the CDF of a standard normal random variable and $\alpha\in(0,1/2)$. The maximization over $l\ge j$ imposes the natural convention that, once a sufficiently large value of $m$ breaks even, larger values of $m$ should also be considered acceptable.

We choose
\begin{equation*}
\small 
\begin{aligned} 
\hat j^\star
=
\min\{j:U_j<0\},
\qquad
\hat m^\star
=
m_{\hat j^\star},
\end{aligned} 
\end{equation*}
with the convention that if $U_j\ge 0$ for all $j$, then $\hat m^\star=\infty$. Thus, $\hat m^\star$ is the smallest value of $m$ for which the ordered upper confidence bound is below zero.

\begin{thm}[Inference for break-even analysis] \label{thm:breakevena}
Let Assumptions~\ref{ass:general_regret} and~\ref{ass:entropy_F} hold for the in-sample and out-of-sample estimates, and assume that the two samples are independent. Suppose that $q_m^2>\underline l$ almost surely for all $m\in\mathcal M$, for some constant $\underline l>0$, and that $|\mathcal M|<\infty$. Take any $Q_m^2$ such that
$ 
\lim_{N_{\mathcal X}\rightarrow\infty}
P(Q_m^2\ge q_m^2)=1
$ 
for all $m\in\mathcal M$. Then for any $\omega_m(x)$ uniformly bounded and measurable with respect to $\mathcal{T}$, 
$ 
\limsup_{N_{\mathcal X}\rightarrow\infty}
P\left(
\hat m^\star<\infty,\,
\Delta_\eta(\hat m^\star)\ge 0
\right)
\le \alpha.
$ 
\end{thm}

\begin{proof}
See Appendix~\ref{proof:thm:breakevena}.
\end{proof}

Theorem~\ref{thm:breakevena} shows that the reported break-even value $\hat m^\star$ controls the probability of selecting a finite $m$ for which the expected loss from future data collection on the selected basin of ignorance is not below the benchmark $\omega_m$. The lower bound on the asymptotic variance rules out degenerate cases. (Also, it is possible to use a similar strategy for other objectives in the break-even analysis, see Appendix \ref{app:breakeven_otherb}.)

\begin{rem}[Recommended $\omega_m$] \label{rem:omega_m}
We recommend choosing $\omega_m$ by asking how much a future unbiased signal would reduce the variance component of the no-abstention rule under a simple Bayesian updating rule (where the focus on the variance of the no-abstention rule provides us with a simple lower bound of its error). 

Formally, let $\hat f^{(0)}$ denote the estimator obtained by forcing no abstention. Conditional on the fitted no-abstention rule, define
$ 
V_0(x)
=
\mathbb V\left(\hat f^{(0)}(x)\right), 
$ 
and $\hat{V}_0(\cdot)$ denote an estimate using in-sample observations. 
Our recommended choice is 
$$
\omega_m(x) = \frac{(2 - w_m)}{w_m} \hat{V}_0(x), \qquad w_m = \frac{V_m(\mathcal{T})}{V_m(\mathcal{T}) + E_m(\mathcal{T})}
$$
where $V_m(\mathcal{T}) = \frac{\sum_x (1 - \hat{\pi}^{(m)}(x)) p(x) \hat{V}_0(x)}{\sum_x p(x) (1 - \hat{\pi}^{(m)}(x))}$ and $E_m(\mathcal{T}) = \frac{\sum_x (1 - \hat{\pi}^{(m)}(x)) p(x) \hat{\eta}(x)^2/(m+1)}{\sum_x p(x) (1 - \hat{\pi}^{(m)}(x))}$, both estimated using in sample observations, so that Theorem \ref{thm:breakevena} holds.\footnote{One can also truncate $\omega_m$ at a sufficiently large value to guarantee uniformly bounded $\omega_m$.} 

Here is why: suppose that for each $x$ in the basin of ignorance a future study produces an independent unbiased signal $\tilde{\tau}_m$ with variance $\eta(x)^2/(m+1)$. A natural updating estimator is $(1-w_m)\hat f^{(0)}(x)+w_m\tilde\tau_m(x)
$,  where 
where $w_m\in[0,1]$ is a shrinkage weight. 
It follows that the excess variance from the updated signal is \textit{proportional} to 
$\frac{\eta(x)^2}{m+1} - \omega_m(x)$, where $\omega_m(x) = \frac{(2-w_m)}{w_m}V_0(x)$. When $w_m = 1$ this corresponds to simply comparing to the variance of $\hat{f}^{(0)}$, while smaller $w_m$ favor abstention more. Our proposed choice mimics a naive (uniform) Bayesian update and provides a \textit{lower bound} on the performance of collecting new data compared to optimal Bayesian updating.  \qed 
\end{rem}

\begin{rem}[Choice of $Q_m$]
\label{sec:inference_detailsa}

We are left to derive an expression for a uniform bound $Q_m^2$. 
Let $a_m(x)=p(x)(1-\hat{\pi}^{(m)}(x))$. Conditional on the training sample, the only randomness in $\widehat\Delta_\eta(m)$ comes from the out-of-sample estimates $\hat{\eta}_{oos}^2(x)$. A natural approach is to take
$ 
Q_m^2
=
\frac{N_{\mathcal X}}{(m+1)^2}
\sum_x a_m(x)^2
\mathbb E\left[
\left(\hat{\eta}_{oos}^2(x)-t_m\right)^2
\mid \mathcal T
\right],
$ 
where
$ 
t_m
=
\frac{
\sum_x p(x)(1-\hat{\pi}^{(m)}(x))\mathbb{E}[\hat{\eta}_{oos}^2(x)|\mathcal{T}]
}{
\sum_x p(x)(1-\hat{\pi}^{(m)}(x))
}
$.  
By replacing the conditional expectation with the empirical analogue, this approach provides a consistent estimate of a valid upper bound on the variance, where the relevant moments are consistently estimated by the law of large numbers.\footnote{Specifically, the reason why this is a valid bound is that
$ 
\mathbb V(\hat{\eta}_{oos}^2(x)\mid \mathcal T)
\le
\mathbb E\left[
\left(\hat{\eta}_{oos}^2(x)-t\right)^2
\mid \mathcal T
\right]
$ 
for any constant $t$ that is fixed conditional on the training sample. We therefore pick $t=\hat t_m$, the average out-of-sample variance estimate on the selected basin of ignorance. } \qed
\end{rem}


\subsection{Optimization routines} \label{sec:optimization}

In this section, we discuss practical optimization routines for several function classes.

\paragraph{Tree based rules}
A simple and interpretable rule are tree based rules
following Example \ref{exmp:tree_entropy}, with at most \(G-1\) predictive groups and depth 
$ 
L:=\log_2(G-1)\in\mathbb N,
$ 
so that it has exactly \(G-1\) terminal nodes. The basin of ignorance is allowed to be the union of any subset of these leaves.\footnote{As for common tree partitions \citep{athey2018generalizedrandomforests}, each terminal node classified as a predictive group must contain at least \(\underline\kappa |\mathcal X|\) types, for a user-specified \(\underline\kappa\in(0,1]\); otherwise the node is assigned to the basin of ignorance. This restriction rules out trivial partition that would define an achetype a group with only one or few different types.} 

This formulation corresponds to growing a tree of depth \(L\) and, at each terminal branch, deciding whether the corresponding cell should be labeled as a predictive archetype or assigned to the basin of ignorance, as illustrated in Figure \ref{fig:tree_clustering_combined}(Panel (a)). As in the literature on policy learning without abstention \citep{zhou2023offline}, the optimization problem can be solved exactly by recursion with appropriate modifications.

The exact solution is described in Algorithm \ref{alg:alg3} and in online Appendix Figure \ref{fig:tree_recursion}. Starting from the root, the algorithm searches over all admissible splits (choice of variable and threshold), evaluates the reward generated by the resulting left and right children, and then calls itself recursively on each child until depth \(L\) is reached. At the terminal stage, each leaf is compared against the basin-of-ignorance option.\footnote{Whenever $L > \log_2(G - 1)$, so that the basin of ignorance can be the union of more than $G - 1$ groups, we provide a greedy pruning algorithm in Appendix \ref{app:algorithm} and formal regret bounds.}

\begin{prop}[Exact solution and complexity guarantees for trees] \label{prop:complexity}
Let \(x\in\mathbb R^d\). Consider an axis-aligned tree partition as in Example \ref{exmp:tree_entropy} with with depth
$ 
L:=\log_2(G-1)\in\mathbb N_+
$. Denote the class of such rules by \(\mathcal P_L\).

Let \(\widehat R(f,\pi)\) denote the empirical criterion in \eqref{eqn:R_hat}. Then Algorithm \ref{alg:alg3}  with $\lambda = 0$ returns a pair \((\hat f,\hat\pi)\) satisfying
$ 
(\hat f,\hat\pi)\in
\arg\min_{(f,\pi)\in\mathcal P_L}
\widehat R(f,\pi).
$ 
In addition, the computational complexity of the algorithm is
$ 
O\!\left((2 d|\mathcal X|)^{L+1}\right),
$ 
so the algorithm is polynomial in \((d,|\mathcal X|)\) for fixed \(L\), and quasi-polynomial in \((d,|\mathcal X|,G)\).\footnote{An algorithm is quasi-polynomial in \(n\) if its worst-case complexity is \(O(n^{(\log n)^{O(1)}})\).}
\end{prop}

\begin{proof} See Appendix \ref{proof:prop:complexity}.  
\end{proof} 

Proposition \ref{prop:complexity} shows that the complexity of Algorithm \ref{alg:alg3} is quasi-polynomial in the joint number of groups, covariates and number of observations. In our application with $|\mathcal X| = 796$, when $G = 9$ the tree takes  a few seconds to run, and for $G=17$ it takes a couple of minutes on a personal computer (Appendix Figure \ref{fig:runtime}).\footnote{An alternative approach is to study clustering in which an unconstrained clustering solution is first computed in a low-dimensional space and then projected onto an interpretable tree structure \citep{moshkovitz2020explainable}. However, these have different objectives and their optimal value can be far from optimum \citep{moshkovitz2020explainable} making these methods not applicable here.}

\begin{rem}[Tuning parameters] \label{rem:lambda} Algorithm \ref{alg:alg3} introduces a regularization parameter $\lambda$ that penalizes the variance of the tree estimator, which in practice we recommend to set $\lambda = 1$. This penalty is not needed for the regret guarantees, but it can improve finite-sample performance, where $\lambda = 1$ captures the finite sample variance of the estimator.\footnote{Since $\widehat{V}_{\mathcal{A}} = O_p(N_{\mathcal{X}}^{-1})$ the regularization parameter does not affect rates in the regret bounds. However, to see why $\lambda = 1$ correctly adjust  for the finite sample variance fix a leaf \(A\) and let $V_A = \sum_{x \in A} p(x)^2 \eta(x)^2/(\sum_{x \in A} p(x))$. Conditional on the leaf, 
$ 
\mathbb E\!\left[\sum_{x\in A}p(x)\{(\hat\tau(x)-\hat\mu_A)^2-\hat\eta(x)^2\}\right]
=
\sum_{x\in A}p(x)(\tau(x)-\mu_A^\star)^2-V_A,
$ 
whereas 
$ 
\mathbb E\!\left[\sum_{x\in A}p(x)(\tau(x)-\hat\mu_A)^2\right]
=
\sum_{x\in A}p(x)(\tau(x)-\mu_A^\star)^2+V_A.
$ 
Thus adding \(2V_A\), corresponding to \(\lambda=1\), captures the variance of the prediction. }  The tree requires $\underline{\kappa}$ to avoid overly-small partitions; in our application even with $\underline{\kappa} = 0.5\%$ or $1\%$ we do not find overly-small partitions when combined with the $\lambda$ penalization and abstention option.
 \qed 
\end{rem}

 
\paragraph{G-means clustering}  Following Example \ref{exmp:kmeans_entropy}, a more general partition rule consists of a $G$-means clustering through the intersection of \(G-1\) halfspaces. The basin of ignorance is obtained from the union of $G$ groups, see e.g., Figure \ref{fig:tree_clustering_combined}(Panel (b)). 

Whenever we cluster on baseline covariates $x \in \mathbb{R}^d, d > 1$, Appendix \ref{app:clustering} provides two algorithms, an exact solution via mixed-integer quadratically constrained programming that can be solved using off-the-shelf software and a smooth approximation via gradient descent. We recommend to implement the exact solution with a hard stopping time and compare the objective to the smooth relaxation to choose.

In the special one-dimensional case with $x \in \mathbb{R}$, the clustering problem admits an exact dynamic-programming formulation \citep[][]{wang2011ckmeans}, with polynomial complexity, a useful observation recently made by \cite{montielolea_archetype_discovery}. See Appendix Algorithm \ref{alg:onedim_clustering}. This can be adopted when either $\tau(x)$ is known or precisely estimated at each $x$ (e.g., when the number of observations for each $x$ is large), in which case one can cluster on $\tau(x)$ directly at the expense of lower interpretability, or when researchers may wish to cluster only using a one-dimensional proxy for heterogeneity. 

\begin{algorithm}[!ht]
\footnotesize
\caption{Exact ignorance-aware tree search}
\label{alg:alg3}
\begin{algorithmic}[1]
\Require Node \(\mathcal S\subseteq\mathcal X\), depth \(\ell\), depth \(L\), minimum leaf size \(\underline\kappa |\mathcal X|\), variance-regularization $\lambda$

\Function{LeafLoss}{$\mathcal{A}$}
    \State Let $V_{\mathcal A}
    :=
    \frac{\sum_{x\in\mathcal A}p(x)^2\hat\eta(x)^2}
    {\sum_{x \in \mathcal{A}} p(x)}$. Compute the archetype and ignorance loss on \(\mathcal A\):
    \[
    \widehat L^{\text{arch}}(\mathcal A)
    :=
   \sum_{x\in\mathcal A} p(x)\Big\{(\hat{\tau}_{\mathcal{A}} - \hat\tau(x))^2-\hat\eta(x)^2\Big\} + 2 \lambda \widehat V_{\mathcal A}, \qquad \hat{\tau}_{\mathcal{A}} := \frac{\sum_{x \in \mathcal{A}} p(x) \hat{\tau}(x)}{\sum_{x \in \mathcal{A}} p(x)}, \qquad  \widehat L^{\text{ign}}(\mathcal A) = \sum_{x \in \mathcal{A}} p(x) \hat{c}(x)
    \]
    \If{\(|\mathcal A|<\underline\kappa |\mathcal X|\)}
        Set \(\widehat L^{\text{arch}}(\mathcal A)=+\infty\)
    \EndIf
    \If{\(\widehat L^{\text{arch}}(\mathcal A)\le \widehat L^{\text{ign}}(\mathcal A)\)}
        \Return archetype assignment on \(\mathcal A\), loss \(\widehat L^{\text{arch}}(\mathcal A)\)
    \Else
        \quad \Return ignorance assignment on \(\mathcal A\), loss \(\widehat L^{\text{ign}}(\mathcal A)\)
    \EndIf
\EndFunction

\Function{TreeSearch}{$\{\mathcal S, \ell\}$}
    \For{each admissible split \((j,t)\) on node \(\mathcal S\)}
        \State Define left and right children
        $ 
        \mathcal S_L(j,t):=\{x\in\mathcal S: x_j\le t\},
        \mathcal S_R(j,t):=\{x\in\mathcal S: x_j> t\}
        $ 
        \If{\(\ell=L\)}
            \State \((a_L,\;E_L)\gets\Call{LeafLoss}{\mathcal S_L(j,t)}\),  \((a_R,\;E_R)\gets\Call{LeafLoss}{\mathcal S_R(j,t)}\)
        \Else
            \State \((a_L,\;E_L)\gets\Call{TreeSearch}{\mathcal S_L(j,t),\ell+1}\), \((a_R,\;E_R)\gets\Call{TreeSearch}{\mathcal S_R(j,t),\ell+1}\)
        \EndIf
        \State Set total loss \(E(j,t):=E_L+E_R\)
        and store corresponding assignments \(A(j,t):=(a_L,a_R)\)
    \EndFor
    \State Let \((j^\star,t^\star)\in\arg\min_{(j,t)} E(j,t)\)
    \State \Return split \((j^\star,t^\star)\), assignments \(A(j^\star,t^\star)\), loss \(E(j^\star,t^\star)\)
\EndFunction

\State \Return \Call{TreeSearch}{$\mathcal{X},1$}
\end{algorithmic}
\end{algorithm}

\paragraph{Smooth regression problems}
It is also possible to generalize the procedure to smooth prediction models,
such as linear or ridge regressions. Let \(f_h(x)\) denote a continuously
differentiable prediction rule indexed by parameters \(h\). Let
\(\pi(x)\in\{0,1\}\) denote the decision to issue a prediction for type \(x\),
with \(\pi(x)=1\) corresponding to prediction and \(\pi(x)=0\) corresponding
to assignment to the basin of ignorance. 
To obtain a differentiable criterion, let \(s_\theta(x)\in\mathbb R\) be a
continuously differentiable score indexed by parameters \(\theta\). The hard
prediction rule is
$ 
\pi_\theta(x)
=
\mathbf 1\left\{s_\theta(x)\ge 0\right\}.
$ 
As for smoothed maximum score principle \citep{horowitz1992smoothed}, we  approximate this hard rule by
$ 
\pi_{\theta,\kappa}(x)
=
K\left(\frac{s_\theta(x)}{\kappa}\right),
$ 
where \(K\) is a smooth cumulative distribution function, such as the standard
normal or logistic cumulative density function and $\kappa$ is typically a small bandwidth parameter (e.g., $\kappa = 0.1$). The corresponding smooth objective is
\begin{equation}
\small 
\begin{aligned} 
\hat R_{\theta,h,\kappa}
=
\sum_x p(x) \pi_{\theta,\kappa}(x)
\left[
\Big(f_h(x)-\hat{\tau}(x)\Big)^2
-
\hat{\eta}^2(x)
-
\hat c(x)
\right], \qquad \pi_{\theta,\kappa}(x):= K\left(\frac{s_\theta(x)}{\kappa}\right). 
\end{aligned} 
\end{equation}
This criterion can be optimized over \((\theta,h)\) using gradient-based
methods, which takes few seconds to run on a personal laptot for standard sized economic datasets, with guarantees to reach a local optima (in practice we recommend using multiple starting points). In addition, the objective directly allows to include regularization in the final objective, as for ridge-type regression. After optimization, the final binary prediction rule is obtained by
thresholding the score 
$ 
\hat\pi(x)
=
\mathbf 1\left\{s_{\hat\theta}(x)\ge 0\right\}, 
$ 
with one more estimation round for $f_h$ under the discrete assignment rule.

 \section{Empirical application and numerical studies}\label{sec:empirics}

In this section, we illustrate the properties of our method by re-analyzing the six experimental evaluations of a multifaceted antipoverty (``Graduation'') program, first described in \cite{banerjee2015multifaceted}. The core intervention consists of providing a bundle of asset transfer, consumption support, training, and access to financial and health services.  The specific implementation was adjusted to each of the six local contexts (Ethiopia, Ghana, Honduras, India, Pakistan, and Peru). The goal is to give poor households the tools to generate a sustained improvement in living standards. Across all six pilot experiments, researchers enrolled 10,495 households spanning more than 500 villages. 
The randomization was conducted at the individual (household) level across the three countries, and at the village and household level in three countries, with household-level randomization within treated villages.\footnote{We abstract from within-village spillovers, consistent with the spillover
checks reported by \citet[Supplementary Text 4]{banerjee2015multifaceted}.
Under this no-spillover interpretation, the IPW pseudo-outcomes in
Example~\ref{exmp:ipw}, constructed using the relevant design probabilities, are
unbiased for household-level direct effects. In sites with village-level assignment,
these probabilities incorporate both village-level treatment assignment and
household-level assignment within treated villages. For simplicity, we approximate
the design as if observations were independent conditional on these design
probabilities; extensions to locally dependent observations are discussed in
Appendix~\ref{sec:local}.}
\cite{banerjee2015multifaceted} conclude that this ``big push'' program has large impacts after pooling across experimental sites, with positive effects on total consumption, an index measuring food security and an index measuring assets.  We focus on these three outcomes one year after the intervention. Following verbatim \cite{banerjee2015multifaceted}, we standardize the outcomes by the country-specific control group mean and standard deviation to make these comparable across countries. 
We use as covariates $x$ the country (experiment), baseline outcomes (total consumption, the food security and asset index measured at baseline) and the total amount of individual loan measured at baseline.  

\paragraph{Estimation} We estimate $\hat{\tau}$ via IPW as in Example \ref{exmp:ipw}: we run Algorithm \ref{alg:ipw_matching} to construct types $x$ with in median ten observations each within each country via non-parametric matching. We estimate $\hat{\tau}(x)$ as the sample mean of each group and $\hat{\eta}(x)^2$ as the sample variance appropriate divided by $n(x)$. 
 In total, we have $|\mathcal{X}| = 796$. 

The archetype structure is the same across different outcomes, whereas the predictions are different for each outcome (see Appendix \ref{sec:multiple_properties}). 
We select a depth-four tree, with the depth four minimizing out-of-sample loss compared to depth-three tree with cross-validation (uniformly over $m \in \{1, \cdots, 7\}$). For both cases, we use Algorithm \ref{alg:alg3} with $\lambda = 1$ to capture the variance of the tree estimator (Remark \ref{rem:lambda}), and a minimum leaf size equal to five unit types $x$ for an archetype (i.e., in median fifty observations, $\approx 0.5\%$ of the sample) and ten types $x$ for the basin of ignorance to avoid overfitting.  

\paragraph{Cost of ignorance via breakeven analysis} Following Section~\ref{sec:cost_abstention}, we set the cost of ignorance to
$ 
\hat c_m(x)=\frac{\hat\eta(x)^2}{m+1},
$ 
where $m$ indexes the effective size of a follow-up experiment relative to the current evidence. We select $m= 3$ through a data-driven choice with our breakeven analysis (Figure \ref{fig:breakeven}). The breakeven analysis identifies the minimum sample size at which we can conclude, with $95\%$ confidence, that the expected loss from additional data collection improves upon the no-abstention variance benchmark
(Section \ref{app:breakeven_other} and Remark \ref{rem:omega_m}), with sample splitting under our recommended choice in Remark \ref{sec:sample_splitting}.   
The selected $m = 3$ corresponds to approximately 15\% units in the basin of ignorance and about 1.5-fold in-sample prediction improvement on the predictive set over the no-ignorance tree, see online Appendix Figure \ref{fig:estimated_tree}.

\paragraph{Estimation results}
 Figure~\ref{fig:ignorance_tree} reports the estimated ignorance-aware tree for \(m=3\), with effects reported both for the average of the standardized consumption, food-security, and asset outcomes, and separately for each outcome. The estimated tree has eleven predictive archetypes and five terminal regions assigned to the basin of ignorance (corresponding in total to \(14.6\%\) of observations). The first split is on baseline loan amount. Subsequent splits are primarily on baseline food security, consumption, assets, and loan amount. Country indicators enter only at lower levels of the tree, for Ethiopia and India. Thus, the rule first pools observations according to baseline economic conditions; it introduces country-specific predictive structure only where these splits provide a better summary of treatment-effect heterogeneity.

The estimated archetypes reveal substantial heterogeneity. Predicted effects range from near zero to \(48\%\) of a control-group standard deviation. Several high-effect archetypes are driven especially by asset accumulation. The tree also identifies large groups with significant but smaller effects, between \(6\%\)--\(14\%\) (more than half of the sample), and two archetypes with effects close to zero ($\approx 8\%$ of the overall sample).

The left panel of Figure~\ref{fig:country_composition} further shows how the estimated ignorance-aware archetypes aggregate evidence across countries. The moderate effect archetypes, appear in every country and account for large shares of Ghana, Honduras, and Peru. Other archetypes are more concentrated: the \(25\%\)-effect archetype is concentrated in Ethiopia, while the \(48\%\)-effect archetype is concentrated in India. The basin of ignorance is present in every country, but its size varies across sites, with larger shares in India and Ethiopia.

The right panel of Figure~\ref{fig:country_composition} reports the share of individuals assigned to the same group (either archetype or basin of ignorance) between two countries relative to the size of each country. Large overlaps indicate that the method often assigns units from different countries to the same archetype groups. 

To understand the structure of the different archetypes in terms of poverty level, 
Figure~\ref{fig:predictions_quartile} summarizes the ignorance-aware predictions by country and by quartiles of the baseline poverty score. The poverty score is the first principal component of baseline assets, consumption, and food security, after standardizing each variable by the within-site baseline control-group mean and standard deviation. Higher values of the score correspond to \textit{poorer} individuals. Countries are ordered by baseline consumption of the experimental population, from lowest to highest. 

The figure illustrates three broad patterns. First, predicted effects are largest in India and Ethiopia, the two countries with the lowest baseline consumption. Second, outside these two countries, predicted effects tend to increase for poorer individuals. In Peru, predicted effects rise from \(6\%\) in the two lowest poverty-score quartiles to \(11\%\) and \(20\%\) in the two highest quartiles, and in Pakistan from \(9\%\)--\(10\%\) to \(21\%\). Third, the procedure recommends more caution in higher-poverty regions: the ignorance share reaches \(43\%\) for the poorest quartile in Peru and \(30\%\) for the poorest quartile in Pakistan, and is also large for India and Ethiopia.
These are individuals where multiple mechanisms may be consistent with the observed evidence and where further experimentation would be especially valuable.

\paragraph{Comparison with the same non-ignorance tree}
The comparison with the tree without ignorance in Figure~\ref{fig:tree_noignorance} illustrates the role of the abstention option. Once every observation must be assigned to a predictive archetype, the tree without ignorance primarily organizes the sample by country. The first splits isolate Ghana, Pakistan, India, and Ethiopia, leaving a large residual archetype that pools the remaining countries and has an index effect of only \(2\%\). Only after these country splits the tree uses baseline economic variables to form additional subgroups. This structure suggests that, without the possibility of abstention, the procedure mostly learns country-specific averages rather than a more portable structure. Figure~\ref{fig:predictions_quartile} confirms this intuition and shows that the no-ignorance tree produces predictions that on average are essentially stable within each country as a function of the poverty score. 

At the same time, for some units, the no-ignorance tree also creates several small leaves with extreme predictions. For instance, it includes leaves with index effects of \(64\%\), \(63\%\), and \(-35\%\), as well as component effects such as an asset effect of \(249\%\), a consumption effect of \(-103\%\), and an asset effect of \(-96\%\). These small leaves are also associated with large standard errors. For instance, two groups have standard errors of $14.5\%$ and the average standard errors per archetype is about $5\%$. By contrast, the ignorance-aware tree sets aside several terminal regions rather than forcing them into predictive archetypes, and present standard errors that do not exceed $4.5\%$, with an average standard error per archetype equal to 2.17\%, about half the no-abstetion tree.

\begin{figure}[!ht]
\centering 
\includegraphics[scale = 0.3]{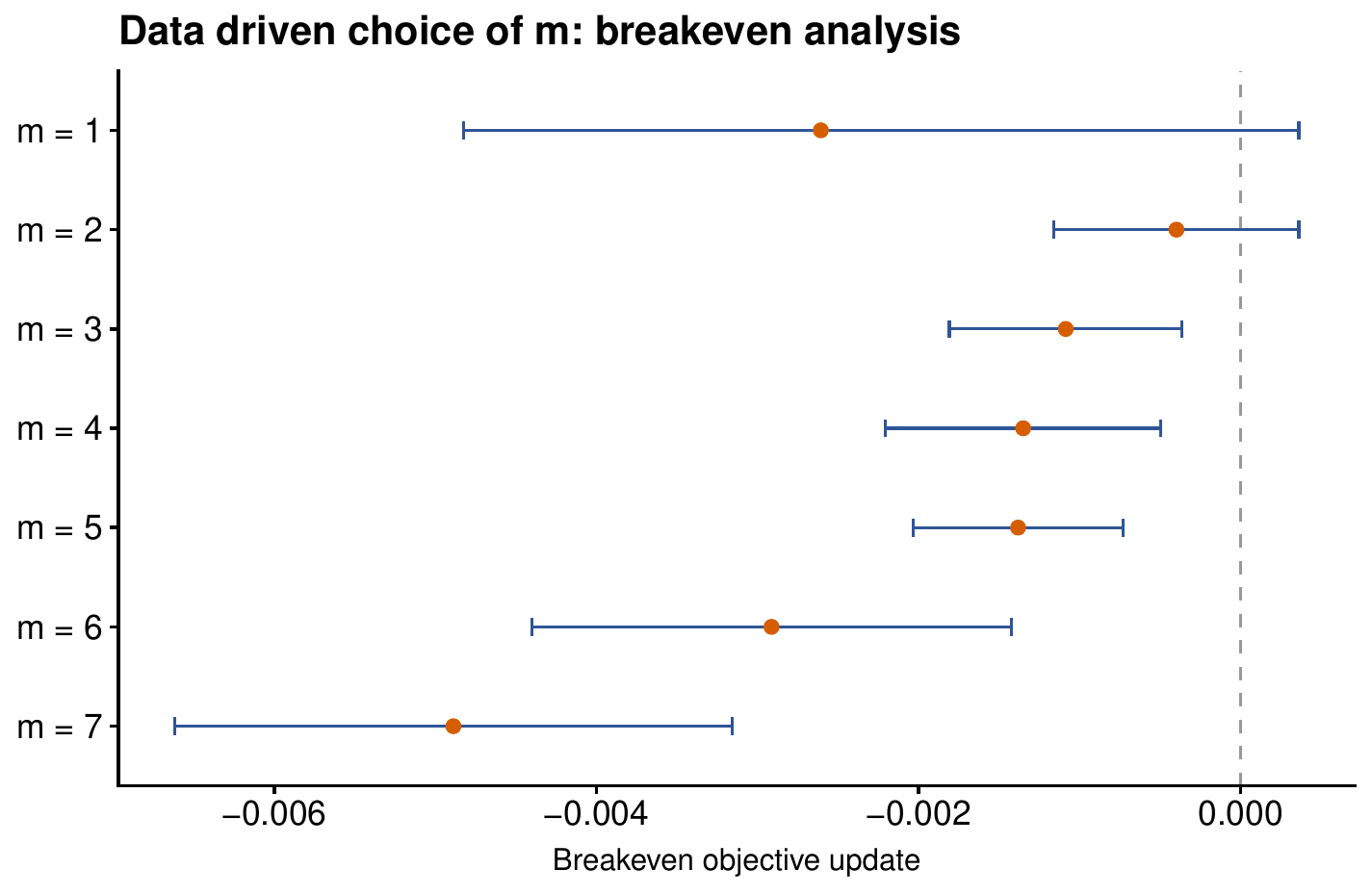}
\caption{\textbf{Data-driven choice of \(m\) via break-even analysis.} We choose \(m\) as the smallest value whose confidence-interval upper bound is strictly below zero. The figure suggests \(m=3\) as the smallest value for a statistically significant precision improvement after collecting more data. Specifically, the figure reports, for each value of \(m\), the average break-even objective on the out-of-sample abstention set, with two-sided \(90\%\) confidence intervals, equivalently one-sided \(95\%\) confidence intervals, as described in Theorem~\ref{thm:breakevena}. The comparison evaluates the loss improvement between collecting more data, as described in Remark~\ref{rem:omega_m}. Confidence intervals are constructed out of sample via three-fold cross-validation with sample splitting generated as recommended in Remark~\ref{sec:sample_splitting}, and confidence interval adjusted for uniform coverage as in Equation \eqref{eqn:S_la}.}

\label{fig:breakeven} 
\end{figure}

\begin{figure}[!ht]
\centering
\includegraphics[scale = 0.4]{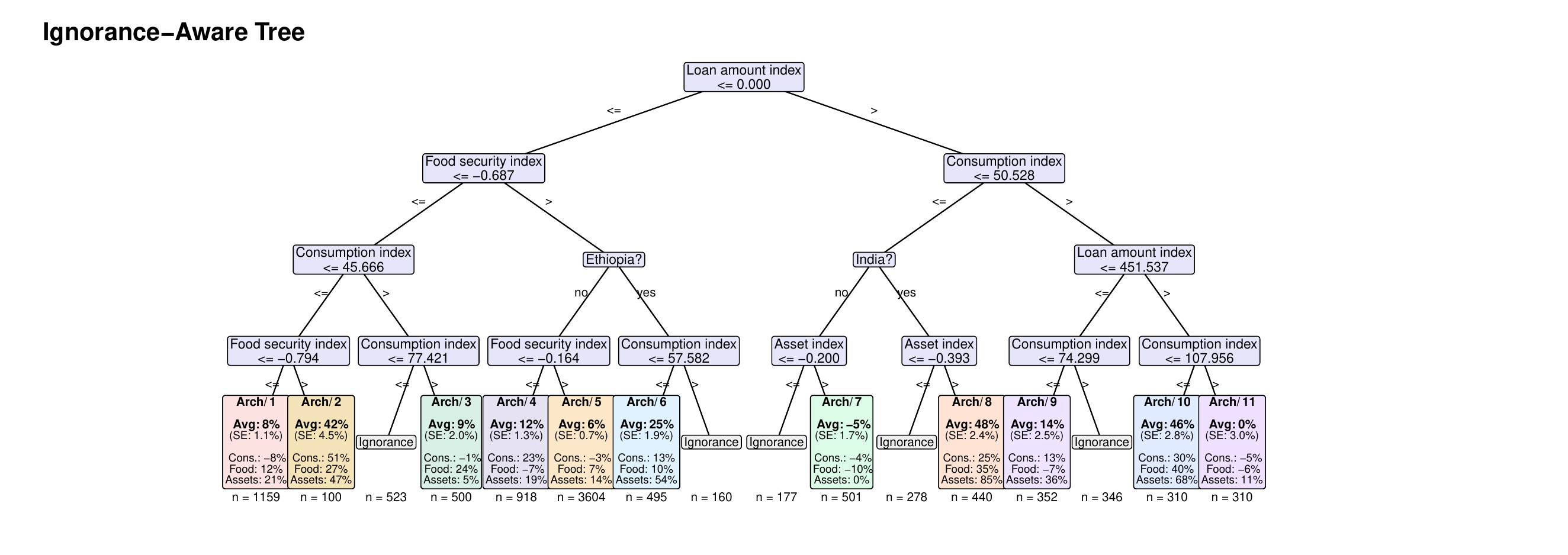}
\caption{\textbf{Full-sample ignorance-aware tree.}
The figure reports the estimated ignorance-aware tree for \(m=3\) and at most sixteen archetypes. Terminal nodes are either assigned to the basin of ignorance or to predictive archetypes. For each predictive archetype, the figure reports the predicted effect on the index outcome, defined as the average of the standardized consumption, food-security, and asset outcomes, as well as the predicted effects on each component outcome separately. Effects are reported as percentages of a control-group standard deviation. The value \(n\) below each terminal node denotes the number of underlying observations assigned to that leaf, computed as the number of types \(x\) in the leaf times the number of units assigned to each type. Treatment effects are measured in standard deviation measures.}
\label{fig:ignorance_tree} 
\end{figure}

\begin{figure}[!ht]
\centering
\includegraphics[scale = 0.4]{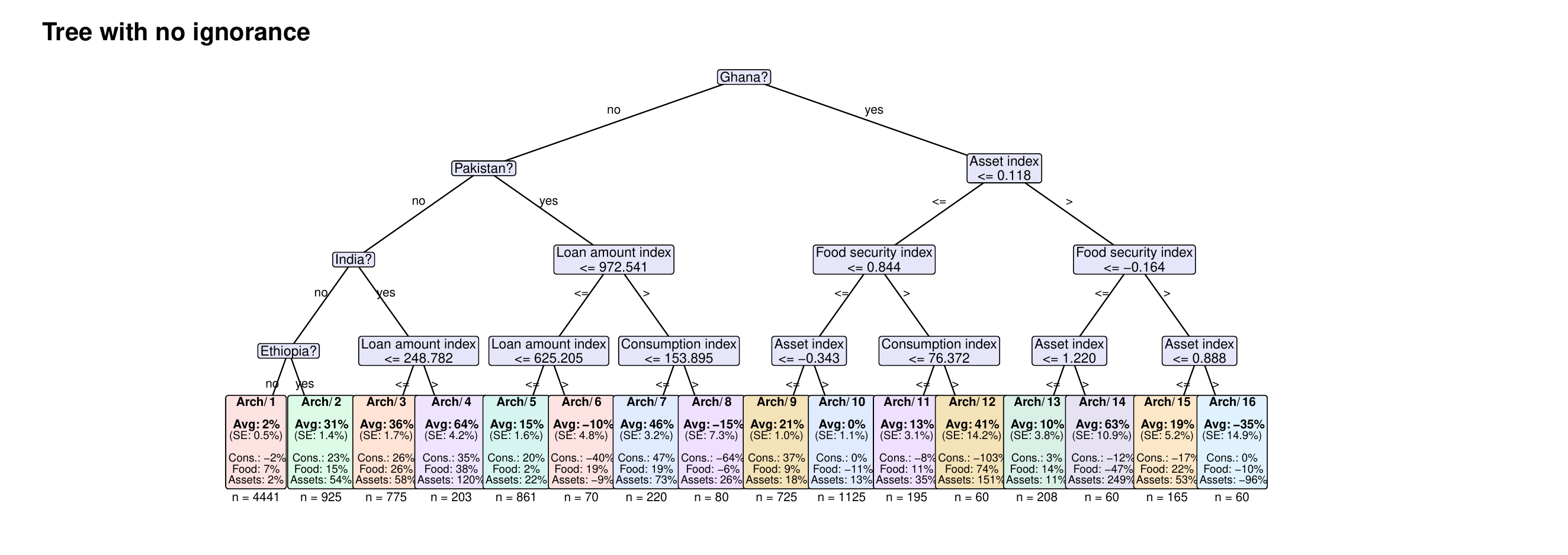}
\caption{\textbf{Full-sample tree without ignorance.}
The figure reports the estimated tree without allowing for ignorance, with at most sixteen archetypes. For each archetype, the figure reports the predicted effect on the index outcome, defined as the average of the standardized consumption, food-security, and asset outcomes, as well as the predicted effects on each component outcome separately. Effects are reported as percentages of a control-group standard deviation. The value \(n\) below each terminal node denotes the number of underlying observations assigned to that leaf, computed as the number of types \(x\) in the leaf times the number of observations within each type. Effects are measured in standard deviation measures.}
\label{fig:tree_noignorance} 
\end{figure}

\begin{figure}[!ht]
\centering
\includegraphics[scale = 0.35]{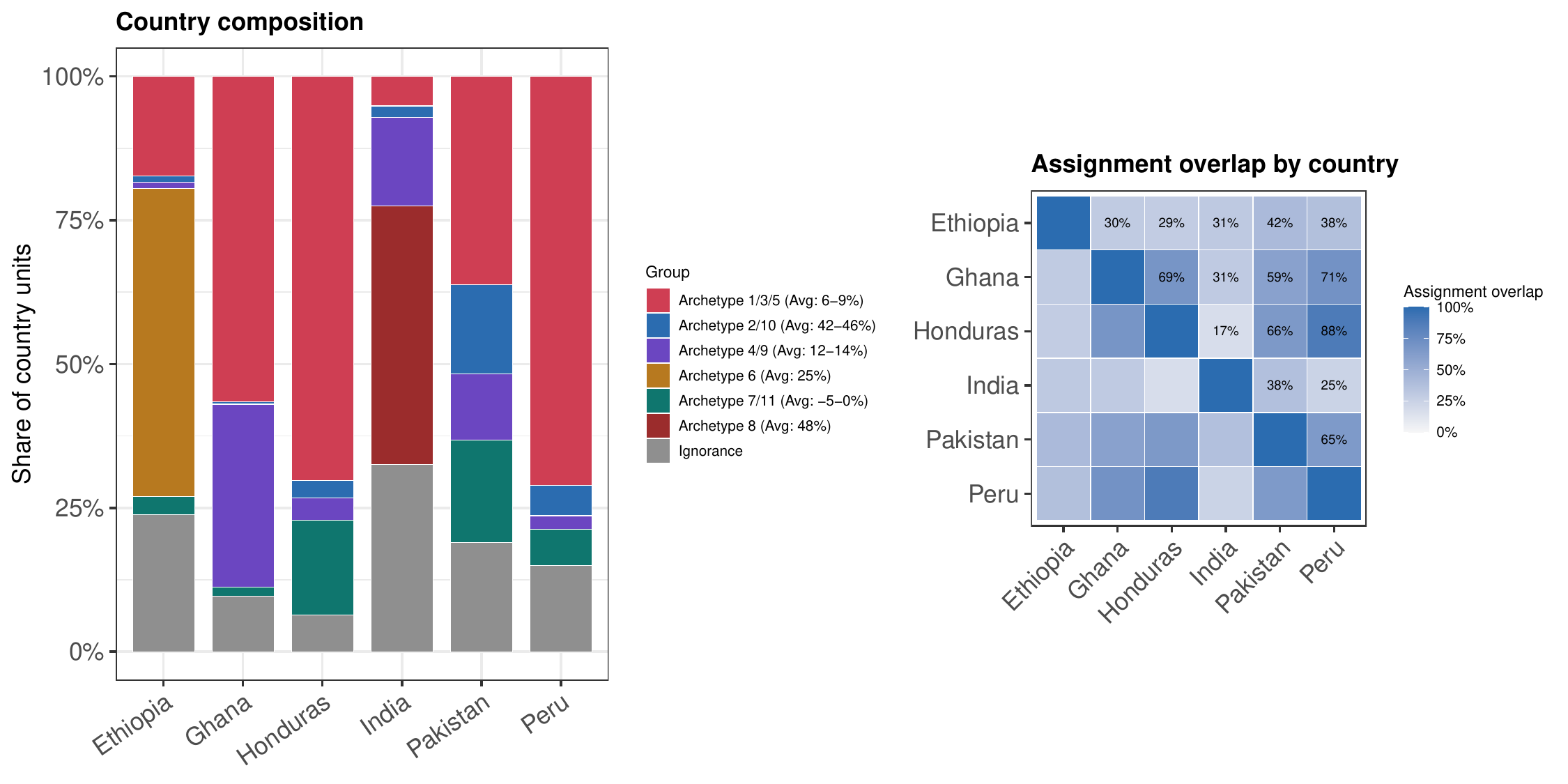}
\caption{\textbf{Country composition of archetypes and assignment overlap.}
The left panel reports, for each country, the share of underlying observations assigned to each predictive archetype or to the basin of ignorance under the ignorance-aware tree in Figure~\ref{fig:ignorance_tree}. For readability, predictive archetypes with similar index effects are grouped in the legend. The right panel reports the pairwise assignment overlap between countries, defined for countries $c$ and $d$ as
$ 
\mathrm{Overlap}(c,d)
=
\sum_{g \in \mathcal{G}}
\min\left\{ s_{cg}, s_{dg} \right\},
s_{cg}
=
\frac{\sum_{i:\, C_i=c,\, G_i=g} w_i}
{\sum_{i:\, C_i=c} w_i},
$ 
where $G_i$ is the tree assignment group of type $i$, $w_i$ is the number of matched individuals represented by type $i$, and $\mathcal{G}$ includes the predictive archetype groups and the basin of ignorance. Darker cells indicate larger overlap in the estimated archetype assignments. The figure shows that several archetypes contain observations from multiple countries. Effects are measured in standard deviation measures.}
\label{fig:country_composition} 
\end{figure}

\begin{figure}[!ht]
\centering 
\includegraphics[scale = 0.5]{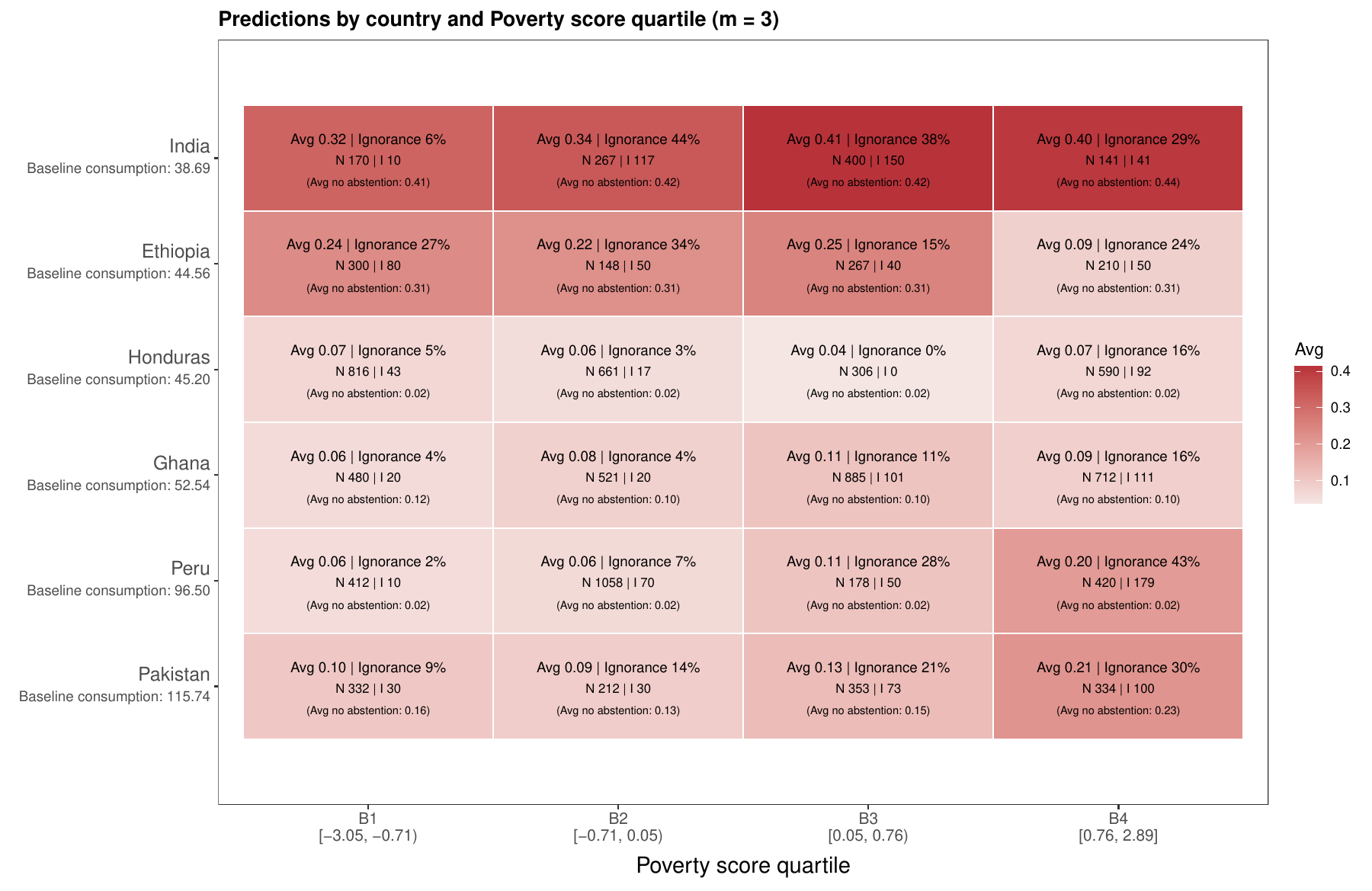}
\caption{\textbf{Ignorance-aware predictions by country and baseline poverty-score quartile.}
The figure reports average predicted effects on the index outcome from the ignorance-aware tree, separately by country and by quartile of the baseline poverty score. The poverty score is constructed as the first principal component of baseline assets, consumption, and food security, after standardizing each variable by the within-site control-group mean and standard deviation. Higher values correspond to poorer individuals. Each cell also reports the share assigned to the basin of ignorance. The labels \(N\) and \(I\) denote, respectively, the number of observations in the cell and the number assigned to ignorance. Treatment effects are measured in standard deviation measures. In parenthesis the prediction of the tree that does not allow for abstention, averaging over all units in a given cell.}
\label{fig:predictions_quartile} 
\end{figure}


To further compare our procedure to the corresponding no-abstention tree, Figure~\ref{fig:comparison_simple_tree} reports three out-of-sample comparisons each computed using three-fold cross validation. The left panel reports the difference in the estimated squared-bias proxy between the depth-4 no-abstention tree and the depth-4 ignorance-aware tree at \(m=3\). Positive values therefore indicate lower squared bias for the ignorance-aware tree. The bias differences are small and not statistically distinguishable from zero in Q1, Q2, and Q4. The main exception is Q3, where the no-abstention tree has a significantly larger squared-bias proxy, of about \(0.02\) (in square-root units, $\sqrt{0.02} \approx 0.14$).

The middle panel reports the ratio of average standard errors on the prediction set, with values above one indicating lower prediction variance for the ignorance-aware tree. In Q1, the standard error of the no-abstention tree is 2.5-times larger than that of the ignorance-aware tree, more than 1.5 times larger for Q2, and smaller but significant and positive for Q3. 
Thus, even with a potentially small abstention set on these units, gains in precisions can be substantial. 

The Q4 result highlights the trade-off induced by abstention. For the highest-poverty
quartile, the standard-error ratio is below one, so the ignorance-aware tree has higher
average standard error among the types on which it continues to predict. This occurs
because a large share of Q4 observations is assigned to the basin of ignorance, reducing
the effective sample size of the retained prediction set. At the same time, the right
panel shows that the no-abstention tree performs poorly on the types that our procedure
abstains from predicting. Thus, the cost of higher variance on the retained Q4 prediction
set is paired with avoiding predictions in precisely the region where the no-abstention
benchmark has high held-out error.

Finally, in online Appendix  \ref{sec:oos_comparisons} (Figure \ref{fig:comparison_grf}) we also present out-of-sample comparisons on the same predictive set with generalized random forest and Bayes meta-regression procedures. We find significant improvements across all quartiles of poverty score, illustrating benefits of the method for targeting heterogeneity with abstention.

\begin{figure}
\centering
\includegraphics[scale = 0.3]{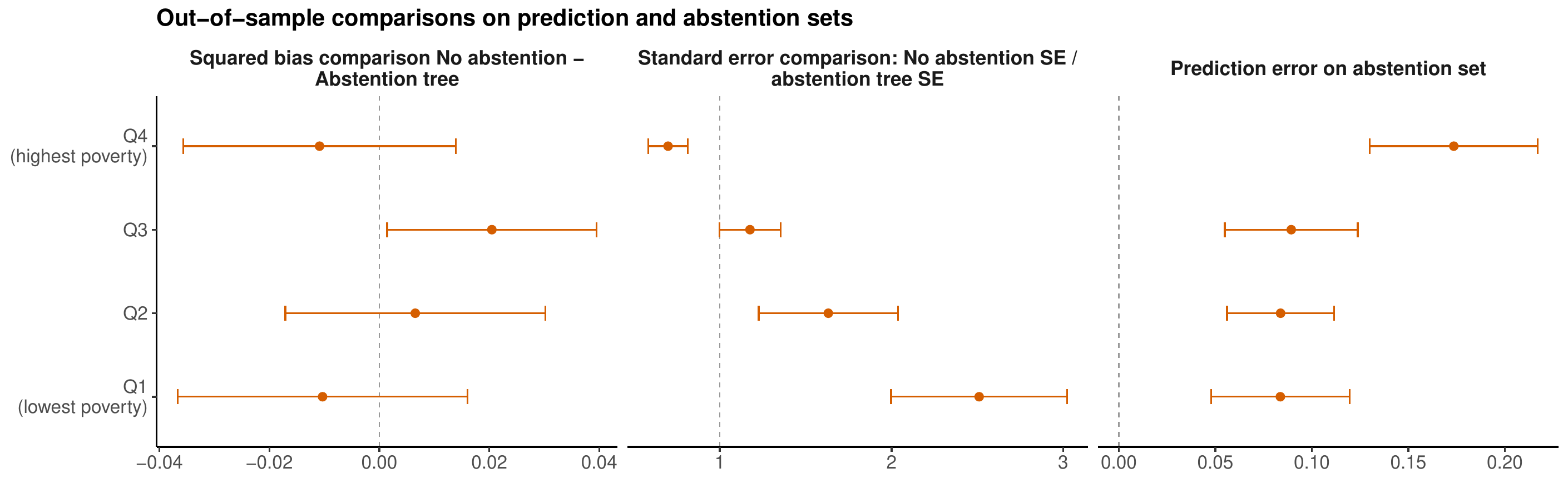}
\caption{\textbf{Out-of-sample comparisons between the depth-4 no-abstention tree and the depth-4 ignorance-aware tree at \(m=3\), shown by baseline poverty-score quartile.} The left panel reports the out-of-sample difference in the squared-bias proxy; positive values indicate lower squared bias for the ignorance-aware tree. This is computed as the out-of-sample sum over the predictive set of $(\hat{f}_0(x) - \hat{\tau}_{oos}(x))^2 - (\hat{f}(x) - \hat{\tau}_{oos}(x))^2 -  (\mathbb{V}(\hat{f}_0(x))) - (\mathbb{V}(\hat{f}(x)))$, with $\mathbb{V}(\cdot)$ as the variance estimate and $\hat{f}_0$ the no-abstention prediction. The middle panel reports the ratio of average standard errors on the prediction set, with values above one indicating lower prediction variance for the ignorance-aware tree. The right panel reports the no-abstention tree's held-out prediction error on the \(m=3\) abstention set, aggregated within each quartile and normalized by the overall abstention-set size. Horizontal bars are \(90\%\) confidence intervals. }
\label{fig:comparison_simple_tree}
\end{figure}

\paragraph{Calibrated numerical study on the prediction set.}
To complement the empirical comparisons above, we report a calibrated Monte Carlo exercise where a small set of types with arbitrary heterogeneity can contaminate predictions. The exercise is deliberately simpler than the main empirical analysis. We focus on the average of the standardized consumption, food-security, and asset outcomes, and use the depth-two calibration tree in Figure~\ref{fig:trees}.
 This tree is estimated from the empirical data.\footnote{Specifically, we compute IPW pseudo-outcomes directly at the empirical support points, use the active covariates displayed in the figure---an indicator for Peru, baseline log consumption, and the baseline asset index---and estimate a depth-two ignorance-aware tree with at most four archetypes, a minimum terminal-node size of twenty observations, five candidate split points per variable, and a constant calibration cost \(c(x)=2.5\). This calibration step is not used as an additional empirical result; it is only a transparent way to define a simulation design with a small region of misspecification.} The calibration tree has two terminal regions assigned to the basin of ignorance. In the Monte Carlo design, however, only the smaller Peru region is treated as truly non-predictive. For all types in the predictive leaves, and for the larger non-Peru ignorance region, we set the treatment effect equal to the corresponding stable archetype prediction from the calibration tree. Thus, the larger non-Peru ignorance region is folded into the stable component of the design. For the small Peru region, which represents approximately \(4\%\) of the sample, we instead draw treatment effects \(\tau(x)\) from a Cauchy distribution with scale parameter \(s\in[0.1,3]\). The parameter \(s\) therefore controls the severity of the heterogeneous contamination. Conditional on the realized \(\tau(x)\), outcomes are drawn as homoskedastic Gaussian with variance one, and treatment assignments are drawn from a Bernoulli design. The resulting simulated IPW signal is an unbiased estimate of \(\tau(x)\); conditional on the realized treatment-effect function, the relevant sampling moments are finite even though the cross-type distribution of effects in the contaminated Peru region is heavy-tailed.

\begin{figure}[!ht]
\centering
\includegraphics[scale = 0.4]{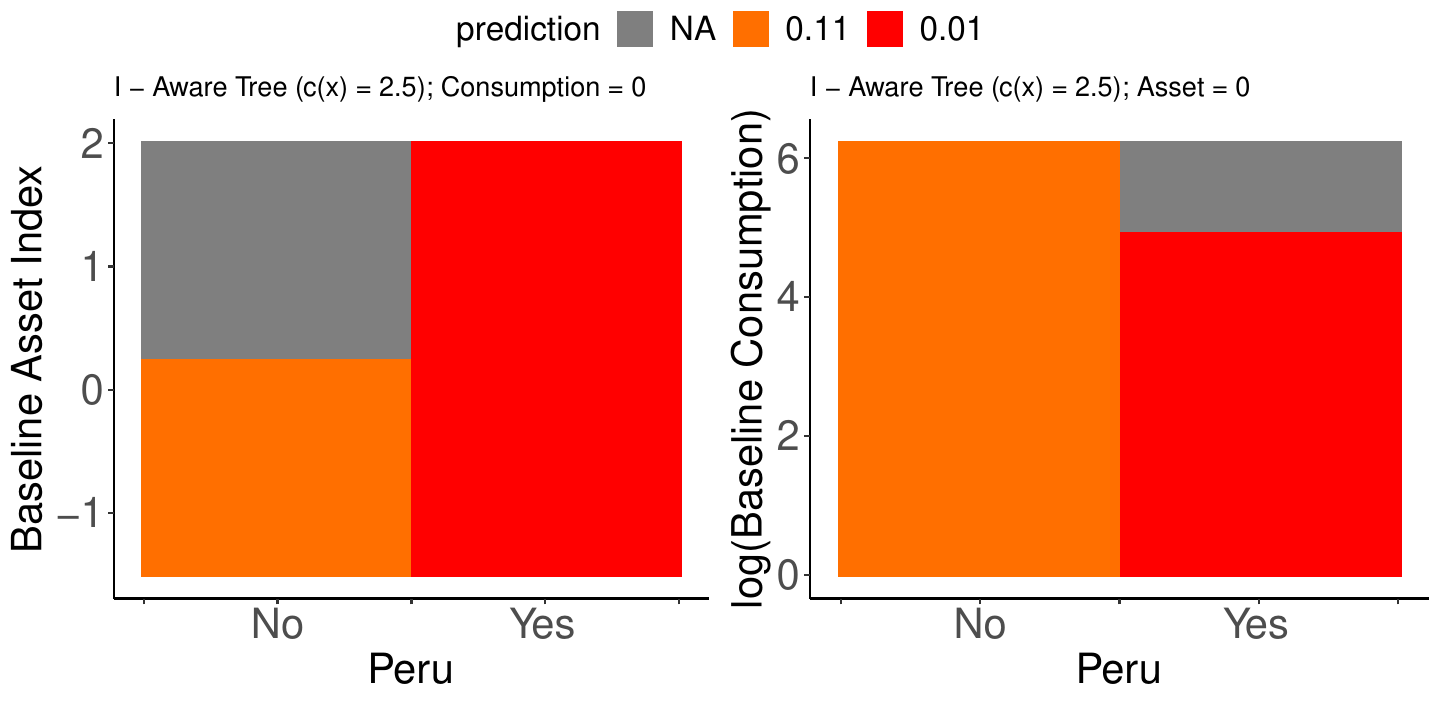}
\caption{\textbf{Depth-two calibration tree used in the numerical study.}
The figure reports the simplified ignorance-aware tree used only to define the Monte Carlo data-generating process. The tree is estimated on empirical covariates and IPW pseudo-outcomes for the index outcome, using the active covariates shown in the figure: an indicator for Peru, baseline log consumption, and the baseline asset index. The calibration uses a depth-two tree with at most four archetypes, a minimum terminal-node size of twenty observations, five candidate split points per variable, and constant abstention cost \(c_{\mathrm{cal}}=2.5\). The calibration tree is estimated using the individual-level IPW signal
\(\hat\phi(X_i)=\widetilde Y_i\), where \(\widetilde Y_i\) is the outcome reweighted by the inverse empirical treatment probability within country; \(\hat\eta(X_i)^2\) is estimated using a linear regression with Lasso and cross-validation. Gray regions denote the calibration basin of ignorance. In the simulation design, only the small Peru ignorance region is assigned arbitrary Cauchy heterogeneity; the larger non-Peru ignorance region is pooled with the stable component of the design outside Peru. Treatment effects are measured in standard deviation measures.}
\label{fig:trees}
\end{figure}

\begin{figure}[!ht]
\centering
\includegraphics[scale = 0.35]{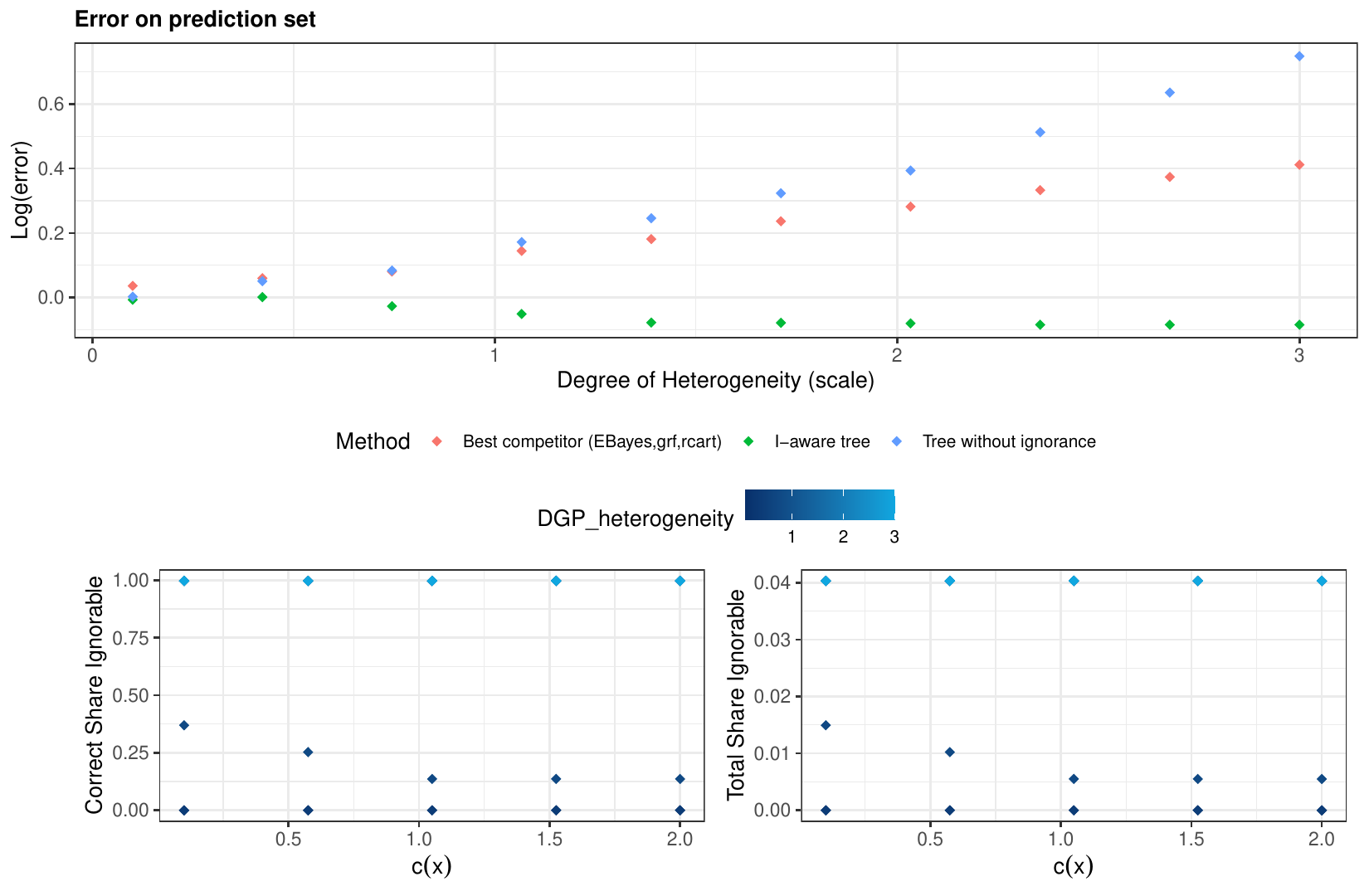}

\caption{\textbf{Calibrated numerical study.}
The top panel reports the median, over one hundred replications, of mean squared prediction error in logarithmic scale on the stable prediction set of the data-generating process. The plotted series compare the best benchmark among empirical Bayes, generalized random forests, and CART; the no-abstention tree, reported as a reference; and the worst performance of the ignorance-aware tree over \(c\le 2\). Errors are normalized by the error of the baseline depth-two tree at the smallest heterogeneity scale. The bottom-left panel reports the share of truly contaminated types correctly assigned to the basin of ignorance, and the bottom-right panel reports the total share of observations assigned to the basin of ignorance. In all panels, the \(x\)-axis is the scale parameter of the Cauchy distribution from which treatment effects are drawn in the contaminated Peru region.}
\label{figure:error}
\end{figure}

For each value of the Cauchy scale \(s\), we generate one hundred Monte Carlo replications. In each replication, we re-estimate the ignorance-aware tree using the simulated \(\hat\tau(x)\) and \(\hat\eta(x)^2\), with a constant cost \(c(x)=c\) varied over
$ 
c\in\{0.1,0.5,1,1.5,2\}.
$ 
We compare the resulting to a depth-two CART tree, generalized random forests with default tuning parameters, and empirical Bayes procedures centered around the estimated CART tree. For empirical Bayes, we report versions that use either the true homoskedastic variance from the simulation design or an estimated variance. The complete summary is provided in Appendix Table \ref{tab:calibrated_simulation_design}. 

Figure~\ref{figure:error} reports the results. The top panel plots the mean squared prediction error, in logarithmic scale, on the stable prediction set of the data-generating process. The figure reports the best-performing benchmark among CART, generalized random forests, and empirical Bayes, the no-abstention version of the tree as a reference, and the worst performance of the ignorance-aware tree over \(c\le 2\). Each statistic is normalized by the error of the baseline depth-two tree at the smallest heterogeneity scale and summarized by its median over the one hundred replications. When the Cauchy scale is small, the ignorance-aware tree performs similarly to the best benchmark. Once the scale is at least \(0.5\), however, allowing for abstention substantially reduces prediction error on the stable prediction set: the error is up to about \(50\%\) lower than that of the best benchmark and about \(80\%\) lower than that of the no-abstention tree.

The bottom panels of Figure~\ref{figure:error} show why. As heterogeneity in the contaminated Peru region increases, the ignorance-aware tree assigns that region to the basin of ignorance with high probability, while keeping the total size of the basin of ignorance close to the true contaminated share. When heterogeneity is small and the cost of abstention is also small, the estimated rule behaves similarly to a standard prediction tree, as expected. The exercise therefore illustrates that even a small amount of arbitrary heterogeneity can distort predictions on regions where aggregation is otherwise valid, and that explicitly learning a basin of ignorance can protect prediction accuracy.

Appendix Figure \ref{fig:comparison_k_means} presents comparisons as we consider one-dimensional $K$-means clustering with potentially many (ten or twenty) small groups to study whether these may ``absorb'' the basin of ignorance; also in this case we observe large improvement of our method on the predictive set.

 \section{Discussion and some practical lessons}\label{sec:discussion}

The growing availability of experiments across different environments (and with heterogeneous individuals) has motivated a large literature on effect heterogeneity and meta-analysis. Estimators in this literature typically aim to learn treatment effects by pooling information across individuals through, e.g., shrinkage or sparsity restrictions. This paper instead focuses on the task of learning when (and how) information from different individuals can be pooled together and when it cannot. To that end, we provide a framework to study evidence aggregation and introduce a class of prediction functions that jointly estimate when and how to form predictions across different observable characteristics and environments. We give the researcher the option to admit ignorance at a given (opportunity) cost. We provide a decision-theoretic foundation of this problem, derive strong finite sample regret guarantees, asymptotic theory for inference and discuss numerical properties of the procedure. An application analyzing a multifaceted program by \cite{banerjee2015multifaceted} illustrates our approach. 

The results of the paper provide practical guidance for an applied researcher interested in treatment effect heterogeneity within a single study and meta-analysis across studies. We study a regime where researchers do not have strong priors on (i) which covariates matter and, most importantly, (ii) when and whether the set of models posed by the researchers is predictive of treatment effects observed in the data. 
Therefore, our method can be used both to inform where to collect further evidence and to detect anomalies in the data, which can be relevant for model discovery. Our method applies beyond looking at environment-by-agent characteristic heterogeneity in the sense that one can interpret the environment much more broadly. For instance,
it also provides a vocabulary to study heterogeneity in research teams, methods, or implementation. 

We leave the reader with many open directions for future work. First, implementing our method may often require harmonizing both outcomes and covariates across studies, and we need better methods to make those comparable. Second, if we seek to learn about mechanisms, rather than predictions, the variables predictive of heterogeneity might not be main drivers of economic phenomena but predictive proxies, raising questions of how to search for mechanisms.  
Third, there are likely deeper implications of our method for how to design future experiments: once we learn the basin of ignorance, there may be ways to prioritize where to run the next experiment. This raises the question of how to combine our method in a dynamic research process, where researchers may sequentially collect data, while leveraging techniques for site selection similar to \cite{olea2024externally},  \cite{gechter2024selecting}.


  \bibliography{bibliography}
\bibliographystyle{chicago}

\counterwithin{figure}{section}
\counterwithin{table}{section}
\counterwithin{rem}{section}        
\counterwithin{algorithm}{section}
\numberwithin{equation}{section}
\appendix 

\newpage 

\setcounter{page}{1}

\section{Extensions} 

We present several extensions, whose proofs are included in online Appendix \ref{app:more_proofs_extensions}. 
\subsection{Estimated $p(x)$} \label{app:estimated_p}

In this section, we consider settings in which \(p(x)\) is estimated from the data. Formally, for each type \(x\in\mathcal X\), let \(S_x\in \mathbb{R}_+\) denote a random sampling count. We interpret randomness in $S_x$ from a sampling based perspective as in \cite{abadie2020sampling} and \cite{viviano2024policy}. We use weights 
$ 
\hat p(x) \propto S_x 
$ (see Equation \eqref{eqn:estimator})
proportional to the target count $S_x$. 
A leading example is when the researcher targets a super-population distribution \(p(x)=P(X=x)\) in a population larger than the observed sample, but only (randomly) sample units of each type \(x\). We formalize this framework below.  

\begin{ass}[Sampling] \label{ass:sampling2}
Suppose that the following hold:
(i) For random variables \(S_x\), denote \(N_S:=\sum_{x\in\mathcal X}\mathbb E[S_x]\). Researchers observe \(S_x\), satisfying
    $
    \mathbb E[S_x] = N_S p(x)
    $. In addition, let \(S_x\le \bar p\) for a finite constant $\bar{p} < \infty$ almost surely and for all $x$, $\mathbb{E}[S_x] > \gamma$ for a constant $\gamma > 0$. 
(ii) The sample is representative:
    $
    \mathbb E[\hat\tau(x)\mid S_x]=\tau(x), 
    \mathbb E[\hat\eta(x)^2\mid S_x]=\mathbb V(\hat\tau(x)\mid S_x), 
    \mathbb E[\hat c(x)\mid S_x]=c(x)
    $ 
    almost surely.
(iii) Independent sampling:
    $
    \bigl(S_x, \hat\tau(x),\hat\eta(x)^2,\hat c(x)\bigr)_{x\in\mathcal X}
    $ 
    is independent across \(x\).
(iv) There exists a finite constant \(M \in [1,\infty)\) such that, for all \(x\in\mathcal X\),
\[
\small 
\begin{aligned} 
\max\Big\{
\mathbb E\big[|\hat\tau(x)-\tau(x)|^3\mid S_x\big],\ 
\mathbb E\big[|\hat\tau(x)^2-\mathbb E[\hat\tau(x)^2\mid S_x]|^3\mid S_x\big],
\end{aligned} 
\]
\[
\small 
\begin{aligned} 
\qquad\qquad
\mathbb E\big[|\hat\eta(x)^2-\mathbb E[\hat\eta(x)^2\mid S_x]|^3\mid S_x\big],\ 
\mathbb E\big[|\hat c(x)-c(x)|^3\mid S_x\big], c(x)^3
\Big\}
\le M
\end{aligned} 
\]
almost surely, with \(|\tau(x)|\le B\) for a finite constant \(B<\infty\).
\end{ass} 

Condition (i) formalizes the sampling framework, where \(S_x\) denotes the empirical count for type \(x\). Here, (i) imposes that sampling is ``unbiased'', i.e., $S_x$ is in expectation proportional to the target density $p(x)$; we also impose a standard sampling overlap condition that requires that the probability of sampling each individual type $x$ in the population of interest is bounded away from zero. Finally, counts are uniformly bounded across types, which matches the leading case of interest in this paper where many types and few observations per type are sampled, as in our applications.

Condition (ii) states that, sampling does not introduce bias on the target estimators. Condition (iii) imposes the same independence structure as in the main text, now for the augmented array including the sampling counts. Condition (iv) requires a standard third moment bound conditional on $S_x$, attained by sub-gaussian or sub-exponential random variables (for simplicity we omit the sharper bounds as a function of the second moment for the sake of brevity); for the IPW estimator below it requires at least two units for each $x$.  

\begin{exmp}[Inverse probability weights with sampling] \label{exmp:ipw2}
Consider a randomized experiment with binary treatment \(D_i\in\{0,1\}\), individual sampling indicator \(R_i\in\{0,1\}\), and observed outcome
$ 
Y_i=D_iY_i(1)+(1-D_i)Y_i(0).
$ 
Let
$ 
e(x)=\Pr(D_i=1\mid X_i=x,\ R_i=1),
S_x=\sum_i R_i\,1\{X_i=x\},
$ 
with \(S_x\ge 2\). 
The target property is the conditional average treatment effect
$ 
\tau(x)=\mathbb E\!\left[Y_i(1)-Y_i(0)\mid X_i=x\right].
$ 
Define
$ 
\tilde Y_i
=
\frac{D_iY_i}{e(X_i)}
-
\frac{(1-D_i)Y_i}{1-e(X_i)},
$ 
and let
$ 
\hat\tau(x)
=
\frac{1}{S_x}
\sum_{i:X_i=x} R_i \tilde Y_i, 
\hat\eta(x)^2
=
\frac{1}{S_x(S_x-1)}
\sum_{i:X_i=x} R_i \bigl(\tilde Y_i-\hat\tau(x)\bigr)^2, 
\hat p(x) \propto S_x.
$
Here condition (ii) in Assumption \ref{ass:sampling2} requires that the sampling indicators be independent of potential outcomes, a standard representativeness condition in experiments.
\end{exmp}

The corresponding population criterion remains unchanged, and is given by \(R_c\) in Equation \eqref{eqn:loss}. We estimate the target loss function as 
\begin{equation} \label{eqn:estimator} 
\footnotesize 
\begin{aligned} 
\widehat{R}_S(f,\pi) = \frac{1}{N_S} \sum_{x\in\mathcal X} S_x \left[\Big(f(x)-\hat\tau(x)\Big)^2 \pi(x) - \hat\eta(x)^2 \pi(x) + \hat c(x)(1-\pi(x))\right].
\end{aligned} 
\end{equation} 
Because the multiplicative factor \(1/N_S\) does not affect the minimizer, we estimate \((\hat\pi,\hat f)\) as 
$ 
(\hat{\pi}_S, \hat{f}_S) = \mathrm{arg} \min_{f \in \mathcal{F}, \pi \in \Pi} N_S \widehat{R}_S(f,\pi).  
$

\begin{lem}[Unbiasedness under estimated sampling weights]
\label{lem:estimated_p_unbiased}
Let Assumption \ref{ass:sampling2} hold.
Then, for every \(f\in\mathcal F\) and \(\pi\in\Pi\),
$ 
\mathbb E[\widehat R_S(f,\pi)] = R_c(f,\pi).
$ 
\end{lem}

\begin{thm}[Regret under estimated target weights] \label{thm:regret_sampling}
Let Assumptions \ref{ass:sampling2}, \ref{ass:bounded_Pi}, and \ref{ass:entropy_F} hold. Then by letting $N_{\mathcal{X}}:=|\mathcal{X}|$, 
$ 
\mathbb{E}\left[
R_c(\hat{f}_S,\hat{\pi}_S)
-
\min_{\pi\in\Pi,\ f\in\mathcal F} R_c(f,\pi)
\right]
\le
c_0
\sqrt{
M\,\frac{C(\mathcal F)+v_\Pi}{\gamma^2 N_{\mathcal{X}} }
},
$ 
for a finite constant \(c_0 \le q_0 \bar{p} \max\{(K + B)^2, 1\}\) for a universal costant $q_0 < \infty$. 
\end{thm}

\subsection{Sparse linear models}

\begin{defn}[High-dimensional linear models with exact \(\ell_0\) sparsity] \label{defn:l0_sparse}
Let \(q(x)\in\mathbb R^d\) be a high-dimensional vector of observed features and consider the class (for $s\ge 1$)
$ 
\mathcal F_{\ell_0}
=
\{f_\theta(x)=q(x)'\theta:\theta\in\Theta_0\},
\Theta_0
=
\{\theta\in\mathbb R^d:\|\theta\|_0\le s,\ \|\theta\|_2\le R\}.
$ 
\end{defn}

Definition \ref{defn:l0_sparse} introduces linear model with sparsity constraints, often of interest in applied work \citep[e.g.][]{banerjee2025selecting, venkateswaran2024robustly}. We impose that the parameters are uniformly bounded (and also $q(x)$), to avoid degenerate predictions. 

\begin{thm}[Sparse high-dimensional models] \label{thm:l0_sparse}
Let Assumptions \ref{ass:general_regret} and \ref{ass:bounded_Pi} hold, and consider the class \(\mathcal F_{\ell_0}\) in Definition \ref{defn:l0_sparse}. Suppose that
$ 
\sup_{x\in\mathcal X}\|q(x)\|_2\le U
$ 
for some finite constant \(U\). Then for any \(u\in(0,1/2]\), letting $N_{\mathcal{X}}:=|\mathcal{X}|$, \\  
$ 
\mathbb E\!\left[
R_c(\hat f,\hat\pi)
-
\inf_{\pi\in\Pi,\ f\in\mathcal F_{\ell_0}} R_c(f,\pi)
\right]
\le
c_0 
\sqrt{
\frac{\tilde M_u\Big(\,s\log\!\left(ed/s\right)+v_\Pi\Big)}
{N_{\mathcal{X}}}
}, 
$ 
for $c_0 \le q_0 \bar{p} \max\{1,RU\}^2 B^{3/2}$ for a universal constant $q_0 < \infty$, and $\tilde{M}_u = \frac{M_u + M_u^2}{u^2}$. 
\end{thm}

\subsection{Regret guarantees when using an upper bound $\bar{c}$ on $c$} \label{sec:upper_c}

In this section we illustrate how our theoretical guarantees extend in settings where the estimated cost is biased upward relative to the true cost. We provide regret guarantees with respect to the oracle evaluated at the \textit{true} cost. 

\begin{prop}[Conservative abstention costs] \label{prop:conservative_cost_model_selection}
Suppose that $\mathbb{E}[\hat{c}(x)] = \bar c(x)\ge c(x)$ for all \(x\in\mathcal X\). Consider the model-selection procedure in Section \ref{sec:model_selection}. 
Assume that the training and validation samples satisfy the same conditions as in Theorem \ref{thm:model_selection}, with \(c\) replaced by \(\bar c\). Then for any \(u\in(0,1/2]\), with $N_{\mathcal{X}} := |\mathcal{X}|$, \\ 
$
\footnotesize 
\mathbb E\!\left[
R_c(\hat f,\hat\pi)
-
\min_{f\in\bar{\mathcal F}}
\sum_{x\in\mathcal X} p(x)\bigl(f(x)-\tau(x)\bigr)^2
\right] \le$  \\ $
\qquad \qquad c_0  \,\sqrt{\frac{J(M_0 + M_0^2)}{N_{\mathcal{X}}}}
+
\inf_{1\le j\le J}
\left\{
c_0
\sqrt{\tilde M_u\,\frac{C(\mathcal F_j)+v_{\Pi_j}}{N_{\mathcal{X}}}}
+
\Delta_{j}^{\mathrm{loc}}(\pi_{j,\bar c}^\star)
+
A_{\bar c}(\pi_{j,\bar c}^\star)
\right\}, 
$ \\
where $\pi_{j,\bar{c}}^\star \in \mathrm{arg} \min_{\pi \in \Pi_j} \min_{f \in\mathcal{F}_j} R_{\bar{c}}(f,\pi)$
and $c_0 \le q_0 B^{3/2} \bar{p} \max\{1,K\}$ for a universal constant $q_0 < \infty$, $\tilde{M}_u = \frac{M_u + M_u^2}{u^2}$. 
\end{prop}

Proposition \ref{prop:conservative_cost_model_selection} shows that our guarantees continue to hold when we use an upper bound on the true cost. 

\subsection{Local dependence} \label{sec:local}

Next, we extend Theorem \ref{thm:regret} to the case where observations are locally dependent. 

\begin{thm} \label{thm:regret_local} Let Assumptions \ref{ass:general_regret}, \ref{ass:bounded_Pi}, \ref{ass:entropy_F} hold, where $(\hat{\tau}(x), \hat{\eta}(x)^2, \hat{c}(x)) \sim \mathcal{D}_x$ and  are not necessarily independent, but rather form a local dependency graph with maximum degree $L$ in the sense of Definition 3.1 in \cite{ross2011fundamentals}. Then for any $u \in (0,1/2]$, letting $|\mathcal{X}| := N_{\mathcal{X}}$, 
$ 
\mathbb{E}\left[R_c(\hat{f}, \hat{\pi}) - \min_{\pi \in\Pi, f \in \mathcal{F}} R_c(f, \pi) \right] \le c_0 (L + 1) \sqrt{\tilde{M}_u \frac{(C(\mathcal{F}) + v_\Pi)}{ N_{\mathcal{X}}} }
$ 
for a constant $c_0 \le q_0 \bar{p} \max\{1,K\} B^{3/2}$ for a universal constant $q_0 < \infty$, and $\tilde{M}_u = \frac{(M_u + M_u^2)}{u^2}$.  
\end{thm}

\subsection{Break-even analysis on prediction set accuracy} \label{app:breakeven_otherb}

An alternative break-even analysis is to compare the loss on the predictive set against the loss on the same estimated predictive set from the prediction function with no ignorance. The analysis is formalized in online Appendix \ref{app:more_more_breakeven}.

\subsection{Additional algorithms}

\subsubsection{Construction of pseudo-true outcomes} \label{sec:matching}

Algorithm \ref{alg:ipw_matching} constructs type-level treatment effect estimates when covariates are continuous by combining inverse-probability weighting with nonparametric matching. For each unit \(i\), it first forms the standard IPW pseudo-outcome
$ 
\tilde Y_i=\frac{D_iY_i}{e_i}-\frac{(1-D_i)Y_i}{1-e_i},
$ 
where \(e_i\) denotes the treatment propensity score (known from the design). These pseudo-outcomes are unbiased for the conditional treatment effect at the unit's covariate value. The algorithm then groups units into small local neighborhoods based on proximity in covariate space. Starting from the set of unassigned observations, it selects a unit and matches it to its \(\nu\) nearest neighbors in Euclidean distance, including all exact ties when several observations share the same covariate profile. 

Each matched group defines a local support point for the subsequent analysis. Its representative covariate value is taken to be the coordinate-wise median of the covariates of the units in the group. The algorithm targets the local group-average effect, which we denote as $\tau(\bar{x})$. Here, $\bar{x}$ simply serves as the representative covariate label for the group (note that our theory does not impose restrictions on covariates). We construct the estimator $\hat{\tau}(\bar{x})$ as the sample average of the pseudo-outcomes within the group, and estimate the corresponding sampling variance $\hat{\eta}(\bar{x})^2$ by the sample variance of the pseudo-outcomes divided by the group size (which is approximately unbiased under the assumption that true treatment effects are roughly constant within the tight local matching neighborhood).

 \begin{algorithm}[!ht]
\footnotesize
\caption{IPW pseudo-outcomes with continuous covariates and variance estimation via matching}
\label{alg:ipw_matching}
\begin{algorithmic}[1]
\Require Matching size \(\nu\); observed data \(\{(Y_i,D_i,X_i)\}_{i=1}^n\); propensity scores \(e_i\) for each unit \(i\)
\For{each \(i\in\{1,\ldots,n\}\)}
    \State Construct the IPW pseudo-outcome
    $ 
    \tilde Y_i
    :=
    \frac{D_iY_i}{e_i}
    -
    \frac{(1-D_i)Y_i}{1-e_i}.
    $ 
\EndFor
\State Initialize the set of unassigned indices \(\mathcal R:=\{1,\ldots,n\}\)
\While{\(\mathcal R\neq\varnothing\)}
    \State Select one index \(i\in\mathcal R\)
    \State Form a group \(\mathcal U\subseteq \mathcal R\) consisting of unit \(i\) and the \(\nu\) closest units to \(i\) in Euclidean distance based on the covariates \(X\)
    \State If more than \(\nu\) units are at exactly zero distance from \(i\), include all such units in \(\mathcal U\)
    \State Remove the indices in \(\mathcal U\) from \(\mathcal R\)
    \State Define the representative covariate value \(\bar x\) as the coordinate-wise median of \(\{X_j\}_{j\in\mathcal U}\)
    \State Define the group-level effect estimate
    $
    \hat\tau(\bar x)
    :=
    \frac{1}{|\mathcal U|}
    \sum_{j\in\mathcal U}\tilde Y_j
    $ 
    \State Define the corresponding variance estimate
    $ 
    \hat\eta(\bar x)^2
    :=
    \frac{1}{|\mathcal U|(|\mathcal U|-1)}
    \sum_{j\in\mathcal U}\bigl(\tilde Y_j-\hat\tau(\bar x)\bigr)^2
    $ 
\EndWhile
\State \Return \(\hat\tau(\bar x)\) and \(\hat\eta(\bar x)^2\) for all matched groups \(\bar x\)
\end{algorithmic}
\end{algorithm}

\subsubsection{Greedy tree optimization with $L > \log_2(G- 1)$}\label{app:algorithm}

Whenever \(L>\log_2(G-1)\), so that a depth-\(L\) tree may generate more than \(G-1\) predictive leaves, we use the greedy pruning procedure described in Algorithm \ref{alg:alg4}. This algorithm is relevant only when we want to introduce additional flexibility to construct the basin of ignorance. The algorithm runs Algorithm \ref{alg:alg3} on the unrestricted depth-\(L\) tree and then, holding the resulting partition fixed, greedily reassigns leaves to the basin of ignorance until the constraint on the number of predictive groups is satisfied; Algorithm \ref{alg:alg4} provides an approximate solution to the constrained problem.

To quantify the quality of this approximation, the algorithm also reports the difference in empirical loss between the unrestricted solution returned by Algorithm \ref{alg:alg3} and the final pruned tree. By construction, this quantity is nonnegative and provides an observable upper bound on the optimization error induced by the greedy pruning.

\begin{algorithm}[!ht]
\footnotesize
\caption{Ignorance-aware tree with greedy pruning}
\label{alg:alg4}
\begin{algorithmic}[1]
\Require Maximum number of predictive groups \(G-1\), tree depth \(L\), grid size \(S\), minimum leaf size \(\underline\kappa |\mathcal X|\)

\State Run Algorithm \ref{alg:alg3} with depth \(L\), grid size \(S\), and root node \(\mathcal X\)
\State Let \((\tilde \alpha,E^\star)\) denote the resulting depth-\(L\) tree partition and its loss
\State Let \(P\) denote the number of terminal leaves in \(\tilde \alpha\)

\For{each leaf \(\ell\in\{1,\dots,P\}\)}
    \State Define the leaf \(\mathcal X_\ell:=\{x\in\mathcal X:\tilde\alpha(x)=\ell\}\)
    \State Compute the best archetype loss and ignorance loss on leaf \(\ell\):
    \[
    \widehat L^{\mathrm{arch}}(\ell)
    :=
    \min_{\mu\in\mathbb R}
    \sum_{x\in\mathcal X_\ell}
    p(x)\Big\{(\mu-\hat\tau(x))^2-\hat\eta(x)^2\Big\}, \qquad \widehat L^{\mathrm{ign}}(\ell)
    :=
    \sum_{x\in\mathcal X_\ell} p(x)\hat c(x)
    \]
    \State Define the leaf score
    $ 
    S_\ell:=\widehat L^{\mathrm{arch}}(\ell)-\widehat L^{\mathrm{ign}}(\ell)
    $ 
\EndFor

\State Initialize all leaves with \(S_\ell\ge 0\) as ignorance leaves
\State Among leaves with \(S_\ell<0\), retain the \(G-1\) leaves with the smallest scores \(S_\ell\) (or all such leaves if there are fewer than \(G-1\))
\State Assign all remaining leaves to the basin of ignorance
\State Relabel the retained predictive leaves as groups \(2,\dots,\tilde G\), with \(\tilde G\le G\), and denote the resulting partition by \(\hat\alpha^t\)

\State Construct the corresponding prediction rule \(\hat f_{\hat\alpha^t}\) and abstention rule \(\pi^{\hat\alpha^t}\)
\State Compute the final loss
$ 
\hat E:=\widehat R(\hat f_{\hat\alpha^t},\pi^{\hat\alpha^t})
$ 
\State Define the optimization gap
$ 
\varepsilon:=\hat E-E^\star
$ 
\State \Return \(\hat\alpha^t\) and \(\varepsilon\)
\end{algorithmic}
\end{algorithm}

\begin{prop}[Approximate tree search] \label{cor:error}
Let $\mathcal{P}_{L,G}$ denote the class of depth-$L$ trees with at most $G-1$ predictive leaves. 
Let Assumption \ref{ass:general_regret} hold, and let \((\hat f^t,\hat\pi^t)\) denote the prediction and abstention rule induced by the partition \(\hat\alpha^t\) returned by Algorithm \ref{alg:alg4}.  Then, for any \(u\in(0,1/2]\), with probability at least $1- \delta$, letting $N_{\mathcal{X}} :=|\mathcal{X}|$, 
$ 
R_c(\hat f^t,\hat\pi^t)
-
\inf_{(f,\pi) \in \mathcal{P}_{L,G}}
R_c\!\left(f,\pi\right)
\le
\frac{c_0}{\delta} 
\sqrt{
\tilde M_u\,
\frac{d G \log^2 G}{ \underline\kappa N_{\mathcal{X}}}
}
+\varepsilon,
$ 
for a constant $c_0 \le \bar{p} \max\{1, K^2\} B^{3/2} q_0$ for a universal constant $q_0 < \infty$ and
 \(\tilde M_u=(M_u+M_u^2)/u^2\). Whenever \(L=\log_2(G-1)\) and Algorithm \ref{alg:alg4} coincides with the exact search in Algorithm \ref{alg:alg3}, we have \(\varepsilon=0\).
\end{prop}

\subsubsection{Optimization with $G$-means clustering} \label{app:clustering}

\paragraph{One-dimensional clustering} Algorithm \ref{alg:onedim_clustering} describes exact $G$-means clustering in one dimension. The algorithm exploits the fact that, after ordering the score \(x\) increasingly, any interval partition on the real line induces contiguous groups. A feasible ignorance-aware rule therefore partitions \(\{1,\dots,|\mathcal X|\}\) into consecutive blocks, each assigned either to a predictive archetype or to the basin of ignorance. For any candidate block \(\{a,\dots,b\}\), the algorithm computes the corresponding predictive and ignorance costs, denoted \(C^{\mathrm{pred}}(a,b)\) and \(C^{\mathrm{ign}}(a,b)\). The dynamic program then compares these two options and recursively combines them across all possible endpoints \(b\ge a\) through the Bellman recursion, selecting the partition and assignments that minimize the total empirical loss. The basin of ignorance corresponds to the complement of the set of predictive archetypes. 
\begin{algorithm}[!ht]
\footnotesize
\caption{Exact one-dimensional interval partition (with $p(x) = 1/N_{\mathcal{X}}$)}
\label{alg:onedim_clustering}
\begin{algorithmic}[1]
\Require $N$ ordered scores \(x_{(1)}\le \cdots \le x_{(N)}\), where \(x_{(i)}\) is any out-of-sample scalar predictor of the conditional effect for type \(i\); in-sample estimates \(\hat\tau_{(i)}\), variance estimates \(\hat\eta_{(i)}^2\), abstention costs \(\hat c_{(i)}\), maximum number of predictive groups \(G-1\), minimum group size \(\underline\kappa N\)

\State Precompute cumulative sums
$ 
S_\tau(j):=\sum_{i=1}^j \hat\tau_{(i)},
S_q(j):=\sum_{i=1}^j \big(\hat\tau_{(i)}^2-\hat\eta_{(i)}^2\big),
S_c(j):=\sum_{i=1}^j \hat c_{(i)}.
$ 

\For{all intervals \(1\le a\le b\le N\)}
    \State Set $ T(a,b):=S_\tau(b)-S_\tau(a-1),
        Q(a,b):=S_q(b)-S_q(a-1)$ and 
    \[
    C^{\mathrm{ign}}(a,b):=S_c(b)-S_c(a-1), \qquad  C^{\mathrm{pred}}(a,b):=\left[Q(a,b)-\frac{T(a,b)^2}{b-a+1}\right]
    \]
    \If{\(b-a+1<  \underline\kappa N\)}
        \State Set \(C^{\mathrm{pred}}(a,b):=+\infty\)
    \EndIf
\EndFor

\State Initialize
$
V(N+1,g):=0 \quad \text{for all } g\in\{0,\dots,G-1\},
V(i,0):=C^{\mathrm{ign}}(i,N) \quad \text{for } i=1,\dots,N.
$ 

\For{\(g=1,\dots,G-1\)}
    \For{\(i=N,\dots,1\)}
        \State Compute
        \[
        V(i,g)
        :=
        \min_{b\ge i}
        \left\{
        C^{\mathrm{ign}}(i,b)+V(b+1,g),\;
        C^{\mathrm{pred}}(i,b)+V(b+1,g-1)
        \right\}.
        \]
        \State Store the minimizing endpoint \(b^\star(i,g)\) and whether block \(\{i,\dots,b^\star(i,g)\}\) is assigned to a predictive group or to ignorance
    \EndFor
\EndFor

\State Recover the optimal partition by iteratively following the stored minimizers \(b^\star(i,g)\), starting from \((i,g)=(1,G-1)\), until all units $i$ have been assigned.
\end{algorithmic}
\end{algorithm}

\paragraph{Multi-dimensional clustering} Algorithm \ref{alg:miqp_gmeans} presents the exact mixed-integer quadratically constrained formulation that can be solved off-the-shelf for multi-dimensional clustering. 
Online Algorithm \ref{alg:gmeans_abstention} describes the approximate (faster) clustering algorithm via gradient descent using a smooth approximation to exact clustering.

\begin{algorithm}[!ht]
\footnotesize
\caption{Exact \(G\)-means clustering with cell-level abstention via MIQCP}
\label{alg:miqp_gmeans}
\begin{algorithmic}[1]
\Require Types \(x_1,\dots,x_{|\mathcal X|}\in\mathbb R^d\), estimates \(\hat\tau_i:=\hat\tau(x_i)\), variance estimates \(\hat\eta_i^2:=\hat\eta(x_i)^2\), abstention costs \(\hat c_i:=\hat c(x_i)\), weights \(p_i:=p(x_i)\), number of Voronoi cells \(G-1\), minimum predictive-cell size \(\underline\kappa |\mathcal X|\), compact centroid set \(\mathcal Z\subset\mathbb R^d\)

\State Introduce binary Voronoi-assignment variables \(y_{ig}\in\{0,1\}\) for \(i=1,\dots,|\mathcal X|\), \(g=2,\dots,G\), where \(y_{ig}=1\) means that type \(x_i\) belongs to Voronoi cell \(g\)

\State Introduce binary cell-status variables \(\delta_g\in\{0,1\}\) for \(g=2,\dots,G\), where \(\delta_g=1\) means that cell \(g\) is predictive and \(\delta_g=0\) means that the entire cell \(g\) belongs to the basin of ignorance

\State Introduce auxiliary binary variables \(r_{ig}\in\{0,1\}\), where
$ 
r_{ig}=y_{ig}\delta_g,
$ 
so that \(r_{ig}=1\) if and only if \(x_i\) is assigned to cell \(g\) and cell \(g\) is predictive

\State Introduce continuous centroid variables \(z_g\in\mathcal Z\), predictive means \(\mu_g\in[-K,K]\), and auxiliary variables \(v_{ig}\in[-K,K]\), where
$ 
v_{ig}=r_{ig}\mu_g.
$ 

\State Formulate the objective
$ 
\min
\sum_{i=1}^{|\mathcal X|}
\sum_{g=2}^G
p_i
\left[
\mu_g v_{ig}
-2\hat\tau_i v_{ig}
+
(\hat\tau_i^2-\hat\eta_i^2)r_{ig}
+
\hat c_i\bigl(y_{ig}-r_{ig}\bigr)
\right].
$ 
\State Add assignment constraints:
$ 
\sum_{g=2}^G y_{ig}=1,
\qquad i=1,\dots,|\mathcal X|.
$ 

\State Add cell-level abstention and predictive-cell size constraints:
$ 
\sum_{i=1}^{|\mathcal X|} y_{ig}
\ge
\underline\kappa |\mathcal X|\,\delta_g,
\qquad g=2,\dots,G.
$ 
Thus, if cell \(g\) is predictive, it must contain at least \(\underline\kappa |\mathcal X|\) types; if \(\delta_g=0\), the entire cell is assigned to the basin of ignorance.

\State Linearize \(r_{ig}=y_{ig}\delta_g\) by imposing
$ 
r_{ig}\le y_{ig},
\qquad
r_{ig}\le \delta_g,
\qquad
r_{ig}\ge y_{ig}+\delta_g-1,
\qquad
i=1,\dots,|\mathcal X|,\quad g=2,\dots,G.
$ 

\State Linearize \(v_{ig}=r_{ig}\mu_g\) using the bound \(\mu_g\in[-K,K]\):
$ 
-Kr_{ig}\le v_{ig}\le Kr_{ig},
$ 
$ 
\mu_g-K(1-r_{ig})\le v_{ig}\le \mu_g+K(1-r_{ig}),
\quad 
i=1,\dots,|\mathcal X|,\quad g=2,\dots,G.
$ 

\State Add Voronoi-consistency constraints:
$
\|x_i-z_g\|^2
\le
\|x_i-z_h\|^2
+
M_i(1-y_{ig}),
i=1,\dots,|\mathcal X|,\quad g,h\in\{2,\dots,G\},\quad g\neq h,
$ 
where
$ 
M_i
\ge
\sup_{z,z'\in\mathcal Z}
\left\{
\|x_i-z\|^2-\|x_i-z'\|^2
\right\}.
$ 

\State Return the partition
\[
\hat\alpha(x_i)
=
\begin{cases}
1, & \text{if } y_{ig}=1 \text{ for some } g \text{ with } \delta_g=0,\\
g, & \text{if } y_{ig}=1 \text{ and } \delta_g=1,
\end{cases}
\]
together with the predictive means \(\hat\mu_g\) for all predictive cells \(g\) with \(\delta_g=1\).
\end{algorithmic}
\end{algorithm}

\subsection{Multiple properties} \label{sec:multiple_properties}

Next, suppose that $\tau(x) \in \mathbb{R}^Q$ for $Q > 1$. In the presence of multivariate properties, we may consider assuming that the archetypical structure (groups) are the same for each property, whereas the predictions can be different. There is a conceptual advantage of considering all of the outcomes simultaneously: configurations are not clustered together as an archetype unless they exhibit similar patterns across different dimensions. Our results directly extend to this case via Theorem \ref{thm:regret_local} by treating the outcome identity as an additional discrete ``covariate type'' (e.g., stacked predictions for consumption, assets, and food security). In this stacked representation, the estimates for different outcomes within the same original type $x$ are locally dependent, which is exactly the structure Theorem \ref{thm:regret_local} accommodates.


\section{Proofs} 

In this section we present the main proofs. We will be using auxiliary lemmas, whose proofs are included in the online additional Appendix \ref{sec:auxiliary_lemmas}. 

\subsection{Proof of Theorem \ref{thm:regret}} \label{proof:thm:regret}

Let
$ 
(\pi^\star,f^\star)\in\arg\min_{\pi\in\Pi, f\in\mathcal F} R_c(f,\pi).
$ 
Since \((\hat f,\hat\pi)\) minimizes \(\widehat R(f,\pi)\) over \((\mathcal F, \Pi)\),
$ 
\widehat R(\hat f,\hat\pi)\le \widehat R(f^\star,\pi^\star).
$ 
Therefore, 
\begin{equation} \label{eqn:basic_in}
\small 
\begin{aligned} 
R_c(\hat f,\hat\pi)-R_c(f^\star,\pi^\star)
\le
\bigl(R_c-\widehat R\bigr)(\hat f,\hat\pi)
+
\bigl(\widehat R-R_c\bigr)(f^\star,\pi^\star)
\le
2\sup_{\pi\in\Pi, f \in \mathcal{F}}\bigl|(\widehat R-R_c)(f,\pi)\bigr|.
\end{aligned} 
\end{equation} 
Therefore it suffices to bound
$ 
\mathbb E\left[\sup_{\pi\in\Pi, f \in \mathcal{F}}\bigl|(\widehat R-R_c)(f,\pi)\bigr|\right].
$ 
We will index units by $i \in \{1,\cdots, N_{\mathcal{X}}\}$ and write
$$  
\small 
\begin{aligned} 
p_i:=p(x_i),
\tau_i:=\tau(x_i),
\hat\tau_i:=\hat\tau(x_i),
\hat\eta_i^2:=\hat\eta(x_i)^2,
\hat c_i:=\hat c(x_i),  c_i:=c(x_i), \eta_i^2:=\eta(x_i)^2
\end{aligned} 
$$ 

\paragraph{Step 1: basic bounds} We state three preliminary bounds below.  \\ 
\textit{Basic bound 1:}
 By Assumption \ref{ass:general_regret}(i), it follows that 
$
\mathbb{E}[\widehat{R}(f,\pi)] = R_c(f,\pi), 
$
which implies that we can write 
$$
\footnotesize 
\begin{aligned} 
& \widehat{R}(f,\pi) - R_c(f,\pi) = \widehat{R}(f,\pi) - \mathbb{E}[\widehat{R}(f,\pi)] = \\ 
& \sum_i p_i \left\{\Big(\Big(f(x_i) - \hat{\tau}_i\Big)^2 - \mathbb{E}\Big[\Big(f(x_i) - \hat{\tau}_i\Big)^2 \Big]\Big) \pi(x_i) - \Big(\hat{\eta}(x_i)^2 - \eta(x_i)^2\Big) \pi(x_i)  + \Big(\hat{c}(x_i) - c(x_i)\Big)(1 - \pi(x_i)) \right\}. 
\end{aligned} 
$$ 
By expanding the squares, we can equivalently write 
$$
\footnotesize 
\begin{aligned} 
& \widehat{R}(f,\pi) - R_c(f,\pi) = \\ 
& \sum_i p_i \left\{\Big(\hat{\tau}_i^2 - \mathbb{E}[\hat{\tau}_i^2]\Big) \pi(x_i) - 2 f(x_i)(\hat{\tau}_i - \mathbb{E}[\hat{\tau}_i]) \pi(x_i) - \Big(\hat{\eta}(x_i)^2 - \eta(x_i)^2\Big) \pi(x_i)  + \Big(\hat{c}(x_i) - c(x_i)\Big)(1 - \pi(x_i)) \right\}. 
\end{aligned} 
$$ 
We therefore bound using the triangular inequality
$$
\footnotesize 
\begin{aligned} 
 \mathbb{E}\left[\sup_{\pi \in \Pi, f \in \mathcal{F}} | \widehat{R}(f,\pi) - R_c(f,\pi)| \right] \le &  
\underbrace{\mathbb{E}\left[\sup_{\pi \in \Pi}\Big|\sum_i p_i \Big(\hat{\tau}_i^2 - \mathbb{E}[\hat{\tau}_i^2]\Big) \pi(x_i)\Big|\right]}_{T_1} + \underbrace{2 \mathbb{E}\left[\sup_{\pi \in \Pi, f \in \mathcal{F}}\Big|\sum_i p_i f(x_i)(\hat{\tau}_i - \mathbb{E}[\hat{\tau}_i]) \pi(x_i)\Big|\right]}_{T_2} \\ 
&+ \underbrace{\mathbb{E}\left[\sup_{\pi \in \Pi}\Big|\sum_i p_i \Big(\hat{\eta}_i^2 - \mathbb{E}[\hat{\eta}_i^2]\Big) \pi(x_i)\Big|\right]}_{T_3} 
 + \underbrace{\mathbb{E}\left[\sup_{\pi \in \Pi}\Big|\sum_i p_i \Big(\hat{c}_i - \mathbb{E}[\hat{c}_i]\Big) (1 - \pi(x_i))\Big|\right]}_{T_4}
\end{aligned} 
$$
\textit{Basic bound 2:} By Theorem 2.6.7 in \cite{van1996weak}\footnote{Or equivalently, Lemma \ref{lem:boundnumber} applied to a single function class $\mathcal{F}_1 = \Pi$ in the lemma with envolope $|\pi(x)| \le 1$ and $\mathcal{F}_2 = \{f(x) = 1\}$ the class of constant functions.}, combined with the VC-bound in Assumption \ref{ass:bounded_Pi}, for any distribution $Q$ on $\{x_1, \cdots, x_n\}$, we can write $\int_0^2 \sqrt{\log(\mathcal{N}(\varepsilon,\Pi, L_1(Q)))} d\varepsilon \le c_0 \sqrt{v_\Pi}$ for a universal constant $c_0 < \infty$.
\\ \textit{Basic bound 3:} The moment bound in Lemma \ref{lem:centered_square} implies that $\max\{\mathbb{E}[|\hat{\tau}_i^2 - \mathbb{E}[\hat{\tau}_i^2]|^{2 - 2u}], \mathbb{E}[|\hat{\tau}_i^2 - \mathbb{E}[\hat{\tau}_i^2]|^{3}]\} \le c B^3 (M_u + M_u^2)$ for any $u \in [0,1/2]$ for a universal constant $c < \infty$.  Also, denote $\lambda_i = N_{\mathcal{X}} p_i/\bar{p}$ and note that by Assumption \ref{ass:general_regret}(iii), $\lambda_i \le 1$. Finally, because $\max\{\mathbb{E}[|\hat{\tau}_i - \mathbb{E}[\hat{\tau}_i]|^{2 - 2u}], \mathbb{E}[|\hat{\tau}_i - \mathbb{E}[\hat{\tau}_i]|^3]\} \le \max\{\sqrt{\mathbb{E}[|\hat{\tau}_i - \mathbb{E}[\hat{\tau}_i]|^{4 - 4u}]}, \sqrt{\mathbb{E}[|\hat{\tau}_i - \mathbb{E}[\hat{\tau}_i]|^6]}\}$ Assumption \ref{ass:general_regret} directly implies that $\max\{\mathbb{E}[|\hat{\tau}_i - \mathbb{E}[\hat{\tau}_i]|^{2 - 2u}], \mathbb{E}[|\hat{\tau}_i - \mathbb{E}[\hat{\tau}_i]|^3]\} \le M_u$. 
\paragraph{Step 2: Symmetrization} Define $\sigma_i$ i.i.d. Rademacher random variables with $P(\sigma_i = 1) = P(\sigma_i = -1) = 1/2$. 
We now invoke Lemma \ref{lem:symmetrization}, where we note that each summand in $T_1, T_2, T_3, T_4$ is centered around zero. Therefore, we write 
\begin{equation} \label{eqn:bounds_start}
\footnotesize 
\begin{aligned} 
T_1 & \le 2 \bar{p} \mathbb{E}\left[\sup_{\pi \in \Pi}\Big|\frac{1}{N_{\mathcal{X}}}\sum_i \lambda_i \sigma_i (\hat{\tau}_i^2 - \mathbb{E}[\hat{\tau}_i^2]) \pi(x_i)\Big|\right], \quad 
T_2  \le 4 \bar{p} \mathbb{E}\left[\sup_{\pi \in \Pi, f \in \mathcal{F}}\Big|\frac{1}{N_{\mathcal{X}}} \sum_i \lambda_i \sigma_i f(x_i)(\hat{\tau}_i - \tau_i)  \pi(x_i)\Big|\right]  \\ 
T_3 & \le 2 \bar{p} \mathbb{E}\left[\sup_{\pi \in \Pi}\Big|\frac{1}{N_{\mathcal{X}}} \sum_i \sigma_i \lambda_i (\hat{\eta}_i^2 - \eta_i^2)  \pi(x_i)\Big|\right], \quad  
T_4  \le 2 \bar{p} \mathbb{E}\left[\sup_{\pi \in \Pi}\Big|\frac{1}{N_{\mathcal{X}}} \sum_i \sigma_i \lambda_i  (\hat{c}_i - c_i)  (1 - \pi(x_i))\Big|\right]
\end{aligned} 
\end{equation} 
\paragraph{Step 3: Bounds on $T_1,T_3,T_4$} 
By Lemma \ref{lem:kita}, where we use $\lambda_i(\hat{\tau}_i^2 - \mathbb{E}[\hat{\tau}_i^2]), \lambda_i (\hat{\eta}_i^2 - \eta_i^2)$ in lieu of $\Omega_i$ in the statement of Lemma \ref{lem:kita}, and Basic Bound 3, we have 
$
T_1\le c_0 B^{3/2} \bar{p} \sqrt{\frac{(M_u + M_u^2) v_\Pi}{N_{\mathcal{X}} u^2}},  T_3 \le  c_0 \bar{p} \sqrt{\frac{(M_u + M_u^2) v_\Pi}{N_{\mathcal{X}} u^2}}
$
for a universal constant $c_0 < \infty$. 
Similarly, by Lemma \ref{lem:vc_sum}, the VC dimension of $1 - \pi, \pi \in \Pi$ equals the VC dimension of $\Pi$. Therefore, we can invoke Lemma \ref{lem:kita} with $\lambda_i (\hat{c}_i - c_i)$ in lieu of $\Omega_i$, which implies $T_4 \le c_0' \bar{p} \sqrt{\frac{M_u v_\Pi}{N_{\mathcal{X}} u^2}}$, with a universal constant $c_0' < \infty$. 

\paragraph{Step 4: Bound for $T_2$} Recall from Assumption \ref{ass:entropy_F} $|f(x)| \le K, K \in (0,\infty)$. With an abuse of notation denote $K_+ = \max\{K,1\}$. 

For term $T_2$, we write 
$$
\footnotesize 
\begin{aligned} 
T_2 \le 4 \bar{p} K \mathbb{E}\left[\sup_{\pi \in \Pi, f \in \mathcal{F}}\Big|\frac{1}{N_{\mathcal{X}}} \sum_i \lambda_i \sigma_i \frac{f(x_i)}{K}(\hat{\tau}_i - \tau_i)  \pi(x_i)\Big|\right] = 4 \bar{p} K \mathbb{E}\left[\sup_{\pi \in \Pi, \tilde{f} \in \tilde{\mathcal{F}}}\Big|\frac{1}{N_{\mathcal{X}}} \sum_i \lambda_i \sigma_i \tilde{f}(x_i) (\hat{\tau}_i - \tau_i)  \pi(x_i)\Big|\right]
\end{aligned} 
$$ 
where we multiply and divide the expression in Equation \eqref{eqn:bounds_start} by $K$ and where $\tilde{\mathcal{F}} = \{\frac{f}{K}: f \in \mathcal{F}\}$ where in the last equality we simply divided $f$ by its envelope $K$. Define
$
\mathcal H:=\{x\mapsto \tilde{f}(x)\pi(x): \tilde{f}\in \tilde{\mathcal{F}} ,\ \pi\in\Pi\}, 
$ 
Since \(|\tilde{f}(x)|\le 1\) under Assumption \ref{ass:entropy_F} ($\tilde{f} := f/K$) and \(0\le \pi(x)\le 1\), every \(h\in\mathcal H\) satisfies
$ 
\sup_{x\in\mathcal X}|h(x)|\le 1.
$ 
By Lemma \ref{lem:cover_product}, for any probability measure \(Q\) on \(\mathcal X\),
$ 
\mathcal{N}\!\left(\varepsilon,\mathcal H,L_1(Q)\right)
\le
\mathcal{N}\!\left(\frac{\varepsilon}{2},\tilde{\mathcal F},L_1(Q)\right)
\,
\mathcal{N}\!\left(\frac{\varepsilon}{2},\Pi,L_1(Q)\right) \le \mathcal{N}\!\left(\frac{K \varepsilon}{2}, \mathcal F,L_1(Q)\right)
\,
\mathcal{N}\!\left(\frac{\varepsilon}{2},\Pi,L_1(Q)\right), 
$ 
where in the last inequality we replaced $\tilde{\mathcal{F}}$ with $\mathcal{F}$ and adjusted the radius accordingly. 
Taking logarithms and using \(\sqrt{a+b}\le \sqrt a+\sqrt b\), we obtain
$
\int_0^{2}\sqrt{\log \mathcal{N}(\varepsilon,\mathcal H,L_1(Q))}\,d\varepsilon
\le
\int_0^{2}\sqrt{\log \mathcal{N}\!\left(\frac{K \varepsilon}{2},\mathcal F,L_1(Q)\right)}\,d\varepsilon
+
\int_0^{2}\sqrt{\log \mathcal{N}\!\left(\frac{\varepsilon}{2},\Pi,L_1(Q)\right)}\,d\varepsilon.
$
By the changes of variables \(t=K \varepsilon/2\) and \(s=\varepsilon/2\),
$ 
\int_0^{2}\sqrt{\log \mathcal{N}(\varepsilon,\mathcal H,L_1(Q))}\,d\varepsilon
\le
\frac{2}{K} \int_0^{K}\sqrt{\log \mathcal{N}(t,\mathcal F,L_1(Q))}\,dt
+
2 \int_0^{1}\sqrt{\log \mathcal{N}(s,\Pi,L_1(Q))}\,ds.
$ 

The first term in the right-hand side is bounded by $\frac{2}{K} \sqrt{C(\mathcal{F})}$ by Assumption \ref{ass:entropy_F} and the second term is bounded by $2c_0' \sqrt{v_\Pi}$ by the Basic Bound 2 for a universal constant $c_0' < \infty$. Therefore we can invoke Lemma \ref{lem:kita}, where we use $(\hat{\tau}_i - \tau_i)\lambda_i$ in lieu of $\Omega_i$ whose relevant moments are bounded as described in Basic Bound 3 above, we can write (after appropriately collecting the term that depend on the envelope $K$)
$
T_2 \le c_0''  \bar{p}  \sqrt{\frac{M_u (C(\mathcal{F}) +  K^2 v_\Pi)}{u^2 N_{\mathcal{X}}}}
$
for a universal constant $c_0'' < \infty$.  

\paragraph{Step 5: Final step} Collecting the terms, we obtain for a universal constant $q_0 < \infty$
\begin{equation} \label{eqn:final_result}
\footnotesize 
\begin{aligned} 
\mathbb E\left[\sup_{\pi\in\Pi, f \in \mathcal{F}}\bigl|(\widehat R-R_c)(f,\pi)\bigr|\right] \le q_0 K_+ B^{3/2} \bar{p} \sqrt{\frac{(M_u + M_u^2) (C(\mathcal{F}) + v_\Pi)}{N_{\mathcal{X}} u^2}}
\end{aligned} 
\end{equation} 
 
The proof completes from Equation \eqref{eqn:basic_in}. 

\subsection{Proof of Theorem \ref{thm:clustering_regret}} \label{proof:thm:clustering_regret}

To prove the result, we use the main result in Equation \eqref{eqn:final_result} in the proof of Theorem \ref{thm:regret}. 
Given Definition \ref{defn:clustering}, consider $f \in \mathcal{F}\equiv \left\{x \mapsto \sum_{g  = 2}^G \mu_g 1\{\alpha(x) = g\}: \alpha \in \mathcal{G}, \mu_g \in [-K, K] \right\}, K \in (0,\infty)$ and a prediction decision rule $\pi \in \Pi$ with $v_\Pi \le v$, where $\Pi = \{\pi: \pi(x) = 1\{\alpha(x) > 1\}, \alpha \in \mathcal{G}\}$. 
In particular, following verbatim Equation \eqref{eqn:basic_in} we can write 
\begin{equation} \label{eqn:helper_g}
\footnotesize 
\begin{aligned} 
R_c(\hat f,\hat\pi)-\inf_{\alpha \in \mathcal{G}, \mu \in [-K,K]^{G-1}} R_c(f_{\alpha, \mu},\pi^\alpha)
& \le 2\sup_{\alpha \in \mathcal{G}, \mu \in [-K,K]^{G-1}}\bigl|(\widehat R-R_c)(f_{\alpha,\mu},\pi^\alpha)\bigr| \le 2\sup_{f \in \mathcal{F},\pi \in \Pi}\bigl|(\widehat R-R_c)(f,\pi)\bigr|
\end{aligned} 
\end{equation} 
where in the last step we use the fact that $\left\{(f_{\alpha,\mu}, \pi^\alpha), \alpha \in \mathcal{G}, \mu \in [-K,K]^{G-1} \right\} \subseteq \left\{(f, \pi), f \in \mathcal{F}, \pi \in \Pi \right\}$  with $\mathcal{F}, \Pi$ appropriately defined as above. 
We verify the assumptions of Theorem \ref{thm:regret} for the induced \(\Pi\) and \(\mathcal F\), and then invoke Equation \eqref{eqn:final_result}.

\paragraph{Step 1: definitions and basic observations}
First, by Definition \ref{defn:clustering}(ii), by Lemma \ref{lem:vc_sum}, the class
$
\Bigl\{x\mapsto 1\{\alpha(x)> 1\}:\alpha\in\mathcal G\Bigr\} =  \Bigl\{x\mapsto 1 - 1\{\alpha(x) = 1\}:\alpha\in\mathcal G\Bigr\}
$ 
has VC dimension at most \(v\). Next, by Definition \ref{defn:clustering}(i), each nonempty predictive group \(g\in\{2,\dots,G\}\) contains at least \(\underline\kappa |\mathcal X|\) types. Therefore, the number of nonempty predictive groups is at most
$
m_\kappa =\min\left\{G-1, \lfloor \frac{1}{\underline\kappa} \rfloor \right\}.
$  For each \(g\in\{2,\dots,G\}\), define
$
\mathcal I_g
:=
\Bigl\{x\mapsto 1\{\alpha(x)=g\}:\alpha\in\mathcal G\Bigr\}.
$ 
By Definition \ref{defn:clustering}(ii), each \(\mathcal I_g\) has VC dimension at most \(v\). Hence, by Theorem 2.6.7 in \cite{van1996weak}, there exist universal constants \(c_0,c_1<\infty\) such that, for every probability measure \(Q\) on \(\mathcal X\),
$
\log \mathcal N(\delta,\mathcal I_g,L_1(Q))
\le
c_0 (v+1)\log\!\left(\frac{c_1}{\delta}\right)$, for all $\delta\in(0,1).$

Similarly, let
$
\mathcal C:=\{x\mapsto a:\ a\in[-K,K]\}
$
denote the class of constant functions. By the definition of covering number, we can break the interval $[-K, K]$ in $1 + 2K/\varepsilon$ many intervals each of length $\varepsilon$, so that
$
\mathcal N(\varepsilon,\mathcal C,L_1(Q))\le 1+\frac{2K}{\varepsilon}.
$
Therefore, once we define 
$
\mathcal H_g:=\mathcal I_g\cdot \mathcal C
=
\Bigl\{x\mapsto \iota(x)c:\ \iota\in\mathcal I_g,\ c\in\mathcal C\Bigr\}, 
$  and using Lemma \ref{lem:cover_product}, and taking logarithms, for universal constants $c_0 > 0, c_1 > 0$
\begin{equation} \label{eqn:bb_cover}
\footnotesize 
\begin{aligned} 
\log \mathcal N(\varepsilon,\mathcal H_g,L_1(Q))
& \le\log \mathcal N(\frac{\varepsilon}{2K},\mathcal I_g,L_1(Q)) + \log \mathcal N(\frac{\varepsilon}{2},\mathcal C,L_1(Q)) \le 
\log\!\Big(1+\frac{4K}{\varepsilon}\Big)
+
c_0(v+1)\log\!\left(\frac{2c_1K}{\varepsilon}\right).
\end{aligned} 
\end{equation} 

\paragraph{Step 2: class decomposition}
Fix a subset \(S\subseteq\{2,\dots,G\}\) with \(|S|=s\le m_\kappa\), and define
$
\mathcal F_S
:=
\Big\{
x\mapsto \sum_{g\in S}\mu_g\,1\{\alpha(x)=g\}
:\ \alpha\in\mathcal G,\ \mu_g\in[-K,K]
\Big\} \subseteq 
\Big\{
\sum_{g\in S} h_g:\ h_g\in\mathcal H_g,\ g\in S
\Big\}.
$ 
Since the indicators \(1\{\alpha(x)=g\}\) are mutually exclusive across \(g\), every \(f\in\mathcal F_S\) is uniformly bounded by \(K\).
By Lemma \ref{lem:devroye_sum}, with $|S| = s \ge 1$
$ 
\mathcal N(\varepsilon,\mathcal F_S,L_1(Q))
\le
\prod_{g \in S} \mathcal N\!\left(\frac{\varepsilon}{s},\mathcal H_g,L_1(Q)\right).
$ 
Using the bound in Equation \eqref{eqn:bb_cover} for each \(\mathcal H_g\), and a universal constant $c_0 < \infty$
\begin{equation} \label{eqn:bbbb}
\footnotesize 
\begin{aligned} 
\log \mathcal N(\varepsilon,\mathcal F_S,L_1(Q))
\le
s\log\!\Big(1+\frac{4Ks}{\varepsilon}\Big)
+
s c_0 (v+1)\log\!\left(\frac{2c_1Ks}{\varepsilon}\right).
\end{aligned} 
\end{equation} 
For \(s=0\), \(\mathcal F_S\) contains only the zero function, so its covering number is one and its logarithm equals zero. 
\paragraph{Step 3: covering number of the union class}
We now pass from \(\mathcal F_S\) to the full class of predictions \(\mathcal F\). Since every \(f\in\mathcal F\) is supported on at most \(m_\kappa\) predictive groups,
$
\mathcal F
\subseteq
\bigcup_{\substack{S\subseteq\{2,\dots,G\}\\ |S|\le m_\kappa}}
\mathcal F_S.
$ 
The number of such subsets is bounded by
$
\sum_{s=0}^{m_\kappa}\binom{G-1}{s}\le G^{m_\kappa}.
$ 
By Lemma \ref{lem:cover_union},
$ 
\mathcal N(\varepsilon,\mathcal F,L_1(Q))
\le
\sum_{\substack{S\subseteq\{2,\dots,G\}\\ |S|\le m_\kappa}}
\mathcal N(\varepsilon,\mathcal F_S,L_1(Q)). 
$ 
Taking logarithms, this implies the following 
$ 
\log \mathcal N(\varepsilon,\mathcal F,L_1(Q))
\le
m_\kappa\log G
+
\sup_{\substack{S\subseteq\{2,\dots,G\}\\ |S|\le m_\kappa}}
\log \mathcal N(\varepsilon,\mathcal F_S,L_1(Q)).
$ 
Hence for universal constants $c_0 <\infty, c_1 < \infty$, by Equation \eqref{eqn:bbbb}, 
$ 
\log \mathcal N(\varepsilon,\mathcal F,L_1(Q))
\le
m_\kappa\log G
+
m_\kappa\log\!\Big(1+\frac{4Km_\kappa}{\varepsilon}\Big)
+
m_\kappa c_0 (v+1)\log\!\left(\frac{2c_1Km_\kappa}{\varepsilon}\right).
$ 
\paragraph{Step 4: Dudley's entropy integral bound}
Using \(\sqrt{a+b+c}\le \sqrt a+\sqrt b+\sqrt c\), together with the change of variables \(t=\varepsilon/(2K)\), we obtain
$ 
\int_0^{2K}\sqrt{\log \mathcal N(\varepsilon,\mathcal F,L_1(Q))}\,d\varepsilon
\le
2K\sqrt{m_\kappa\log G}
+
2K\sqrt{m_\kappa}\int_0^1 \sqrt{\log\!\Big(1+\frac{2m_\kappa}{t}\Big)}\,dt
+
2K\sqrt{m_\kappa c_0(v+1)}
\int_0^1 \sqrt{\log\!\Big(\frac{c_2m_\kappa}{t}\Big)}\,dt,
$ 
for a universal constant \(c_2<\infty\). The two integrals on the right-hand side are bounded by a universal constant times \(\sqrt{\log(m_\kappa+1)}\). Therefore, for a finite constant $c' < \infty$ 
$ 
\int_0^{2K}\sqrt{\log \mathcal N(\varepsilon,\mathcal F,L_1(Q))}\,d\varepsilon
\le
2K\sqrt{m_\kappa\log G}
+
c' K\sqrt{m_\kappa(v+1)\log(m_\kappa+1)}.
$ 
Since \(m_\kappa\le G, G \ge 2\), we have \(\log(m_\kappa+1)\le \log(2G)\), and therefore for a finite universal constant $c_0' < \infty$ 
$ 
\sup_{Q\in\mathcal Q}
\int_0^{2K}\sqrt{\log \mathcal N(\varepsilon,\mathcal F,L_1(Q))}\,d\varepsilon
\le
c_0' K\sqrt{m_\kappa(v+1)\log G}
\le
c_0' K\sqrt{\frac{(v+1)\log G}{\underline\kappa}}.
$ 
Thus, we can write 
$
C(\mathcal F)\le (c_0')^2 K^2\frac{(v+1)\log G}{\underline\kappa}.
$ 
\paragraph{Step 5: final step} Collecting our results, and using Equation \eqref{eqn:final_result}, we obtain
$$  
\small 
\begin{aligned} 
\mathbb{E}\left[\sup_{f \in \mathcal{F},\pi \in \Pi}\bigl|(\widehat R-R_c)(f,\pi)\bigr|\right] & \le C_1 \bar{p} B^{3/2} \max\{K,1\}
\sqrt{
\tilde M_u\,
\frac{C(\mathcal F)+v}{|\mathcal X|}
} \\
&\le C_1' \bar{p} B^{3/2} \max\{K,1\}
\sqrt{
\tilde M_u\,
\frac{K^2(v+1)\log G/\underline\kappa + v}{|\mathcal X|}
}, 
\end{aligned} 
$$ 
for universal constants \(C_1, C_1'<\infty\), where in the second inequality we used $
C(\mathcal F)\le (c_0')^2 K^2\frac{(v+1)\log G}{\underline\kappa}.
$  from Step 4. The proof completes from Equation \eqref{eqn:helper_g}.  

\subsection{Proof of Theorem \ref{thm:gmeans_lower_union_ignorance}}
\label{proof:lower_bound}

Write \(N:=N_{\mathcal X}, M:=\frac G2-1\). We first prove the theorem for even \(G\ge 4\), and then
indicate the modification for odd \(G\ge 5\) at the end. Throughout the proof, label
\(1\) denotes ignorance. For any rule \((\alpha,\mu)\), with \(\alpha:\mathcal X_N\to\{1,\dots,G\}\) and
\(\mu\in[-K,K]^{G-1}\), define
\begin{equation}\label{eq:proof-risk-def}
\footnotesize 
\begin{aligned} 
R_\theta(\alpha,\mu)
:=
\frac1N\sum_{x\in\mathcal X_N}
\left[
(\mu_{\alpha(x)}-\tau_\theta(x))^2\,\mathbf 1\{\alpha(x)\ge 2\}
+
K_0^2\,\mathbf 1\{\alpha(x)=1\}
\right], 
\end{aligned} 
\end{equation}
where we index $\tau_\theta$ by parameters $\theta$ defining the DGP as we discuss in Step 2 below (without loss, let $\mu_1 = 0$). We will be using lemmas  in the online Appendix \ref{sec:auxiliary_lower_bound}. 
\paragraph{Step 1.} Let
$ 
\mathcal X_N=\{x_1,\dots,x_N\},
x_i=(i,0,\ldots,0)\in\mathbb R^d,
p(x_i)=\frac1N.
$ 
Because all coordinates of $x$ are zero except for the first one, we directly work with \(x_i = i\) with its first coordinate \(i\). For either choice of \(\mathcal G\) in the theorem (Examples \ref{exmp:tree_entropy} or \ref{exmp:kmeans_entropy}), the restriction of any rule
\(\alpha\in\mathcal G\) to \(\mathcal X_N\) has the following one-dimensional structure:
the nonempty predictive groups are pairwise disjoint consecutive intervals, and their number is at most \(G-1\), with minimum size for the predictive group $\underline{\kappa} |\mathcal{X}|$ (from Definition \ref{defn:clustering}). 
The benchmark partitions in Step 3 below are realizable by both classes in Examples \ref{exmp:tree_entropy} and \ref{exmp:kmeans_entropy} in this one-dimensional setting. (This is because for the tree class we can use splits on the first coordinate so to create leaves that correspond to consecutive intervals, where each leaf node by construction can either be a single archetype or part of the basin of ignorance; for the nearest-centroid class we can choose one-dimensional centroids to generate partitions of consecutive intervals.) 

We fix integers \(a\ge 1\), \(b\ge 1\) and a real-valued \(\Delta \ge 0\) which we choose in Step 4 below, and define $L:= N-((G-1)a+Mb). 
$ We
partition \(\mathcal X_N\) into disjoint consecutive blocks
$ 
A_2,\ H_1,\ A_3,\ A_4,\ H_2,\ A_5,\ \dots,\ A_{2M},\ H_M,\ A_{2M+1},\ A_G,
$ 
with
$ 
|A_g|=a \quad (g=2,\dots,G-1),
|A_G|=a+L,
|H_m|=b \quad (m=1,\dots,M).
$ 
For every \(m=1,\dots,M\), the block \(H_m\) lies immediately
between the two blocks \(A_{2m}\) and \(A_{2m+1}\), which we refer to as anchor blocks, and both anchors adjacent to $H_m$ have
cardinality \(a\).  

\paragraph{Step 2: data-generating process.}
For each \(\theta=(\theta_1,\dots,\theta_M)\in\{-1,+1\}^M\), set
\begin{equation} \label{eqn:tau_theta}
\footnotesize 
\begin{aligned} 
\tau_\theta(x) = \sum_{g =2}^G (-1)^g K_0 1\{x \in A_g\} + \sum_{m=1}^M \theta_m \Delta 1\{x \in H_m\}
\end{aligned} 
\end{equation} 
The observations are generated by
$ 
\hat\tau(x_i)=\tau_\theta(x_i)+\xi_i,
\xi_i\overset{\mathrm{i.i.d.}}{\sim}N(0,\underline\eta^2),
\hat\eta(x_i)^2=\underline\eta^2,
\hat c(x_i)=c(x_i)=K_0^2.
$ 
For each \(\theta\in\{-1,+1\}^M\), this defines a law \(P_\theta\) over the observed estimates. We prove that this DGP satisfies Assumption \ref{ass:general_regret} at the end of the proof in Step 11. 

\paragraph{Step 3: benchmark rule.}

For each \(\theta\in\{-1,+1\}^M\), define
\(\alpha^\theta:\mathcal X_N\to\{2,\dots,G\}\) by (for $g \in \{2, \cdots, G\}, m \in \{1, \cdots, M\}$)
$$
\footnotesize 
\begin{aligned}
\alpha^\theta(x) = \sum_{g=2}^G g 1\{x \in A_g\} + \sum_{m=1}^M \Big\{2m 1\{x \in H_m\} 1\{\theta_m = 1\} + (2m +1) 1\{x \in H_m \} 1\{\theta_m = -1\}\Big\}
\end{aligned} 
$$
Thus each \(H_m\) is assigned to \(A_{2m}\) if \(\theta_m=+1\), and to
\(A_{2m+1}\) if \(\theta_m=-1\). This forms a partition of $x_i$ with consecutive predictive groups defined as follows:
$$ 
\footnotesize 
\begin{aligned} 
C_g^\theta:=\{x\in\mathcal X_N:\alpha^\theta(x)=g\}
=
\begin{cases}
A_{2m}\cup H_m, & g=2m,\ \theta_m=+1,\quad m=1,\dots,M,\\[3pt]
A_{2m}, & g=2m,\ \theta_m=-1,\quad m=1,\dots,M,\\[3pt]
A_{2m+1}, & g=2m+1,\ \theta_m=+1,\quad m=1,\dots,M,\\[3pt]
H_m\cup A_{2m+1}, & g=2m+1,\ \theta_m=-1,\quad m=1,\dots,M,\\[3pt]
A_G, & g=G.
\end{cases}
\end{aligned} 
$$ 
In Step 4 below we prove that $\alpha^\theta \in \mathcal{G}$ after appropriately choosing the constants. 

\paragraph{Step 4: choice of constants and benchmark loss.}

Fix
$
0 < \underline\eta \le \frac{K_0}{32},
b:=\Bigl\lfloor \frac{N}{128(G-1)}\Bigr\rfloor, 
\Delta:=\frac{\underline\eta}{8\sqrt b}, 
a:=\Bigl\lfloor \frac{N-Mb}{G-1}\Bigr\rfloor. 
$ We establish their properties below. 

\begin{lem}\label{lem:bound_basic_constants}
Suppose the sample-size condition in Theorem~\ref{thm:gmeans_lower_union_ignorance} holds
with \(C_0 \ge 128\). Then
$ 
a\ge \underline\kappa N,
a\ge 32b, 
0\le \Delta\le \frac{K_0}{8}, 
b\ge \max\left\{1,\frac{N}{256(G-1)}\right\},
a\ge \max\left\{\frac{N}{2(G-1)},\,4(G-1)\right\},
 $ 
and
\begin{equation}\label{eq:proof-MbDelta-vs-aK0}
\small 
\begin{aligned} 
\frac{MbK_0\Delta}{2N}\le \frac{aK_0^2}{32N}, \qquad \frac1N\frac{ab}{a+b}(K_0-\Delta)^2 \le \frac{bK_0^2}{N}
\end{aligned} 
\end{equation}
In addition, \(0\le L<G-1\), where we define $L$ a constant $L:= N-((G-1)a+Mb)$. 
\end{lem}
\begin{proof} See Appendix \ref{proof:lem:bound_basic_constant}. 
\end{proof}

Next, we show that $\alpha^\theta \in \mathcal{G}$ and construct a benchmark rule. 
Define
$$
\footnotesize  
\begin{aligned} 
\mu_g^\theta
:=
\frac1{|C_g^\theta|}\sum_{x\in C_g^\theta}\tau_\theta(x).
\end{aligned} 
$$ 
Since \(|\tau_\theta(x)|\le K_0\le K\), \(\mu_g^\theta\in[-K,K]\).  In addition, 
Because each \(C_g^\theta\) contains \(A_g\), Lemma~\ref{lem:bound_basic_constants} gives
$ 
|C_g^\theta|\ge |A_g|\ge a\ge\underline\kappa N,  g=2,\dots,G.
$ 
Hence, by Step 1, 
$ 
\alpha^\theta\in\mathcal G,
$ for the class $\mathcal{G}$ considered in Theorem \ref{thm:gmeans_lower_union_ignorance} under the minimum size requirement in Definition \ref{defn:clustering} 
and therefore
\begin{equation}\label{eq:proof-benchmark-admissible}
\footnotesize 
\begin{aligned} 
\inf_{\alpha\in\mathcal G,\,
\mu\in[-K,K]^{G-1}}
R_\theta(\alpha,\mu)
\le
R_\theta(\alpha^\theta,\mu^\theta).
\end{aligned} 
\end{equation}

We also record the the corresponding loss on $A_g$. If \(C_g^\theta=A_g\), then the benchmark
loss on \(A_g\) is zero. If \(C_g^\theta=A_g\cup H_m\), then \(|A_g|=a\) and
$ 
\mu_g^\theta=(-1)^g\frac{aK_0+b\Delta}{a+b}.
$ 
Thus, for \(x\in A_g\),
$ 
|\mu_g^\theta-\tau_\theta(x)|
=
\frac{b(K_0-\Delta)}{a+b}
\le \frac{b}{a}K_0
\le \frac{K_0}{32},
$ 
where the last inequality uses Lemma~\ref{lem:bound_basic_constants}. Hence
\begin{equation}\label{eq:proof-benchmark-anchor}
\footnotesize 
\begin{aligned} 
\sup_{g=2,\dots,G}\sup_{x\in A_g}
(\mu_g^\theta-\tau_\theta(x))^2
\le
\frac{K_0^2}{1024}.
\end{aligned} 
\end{equation}

We conclude this step with two additional definitions that will be used below. First, note that the following holds: 
\begin{equation}\label{eq:proof-two-point-ls}
\footnotesize 
\begin{aligned} 
\inf_{t\in\mathbb R}\{n(t-u)^2+p(t-v)^2\}
=
\frac{np}{n+p}(u-v)^2,
\end{aligned} 
\end{equation}
valid for \(n,p\ge0\), \(n+p>0\). Also, define, for \(k\in\{0,\dots,b\}\) and \(u\in[0,k]\),
\begin{equation} \label{eqn:Phi_k}
\footnotesize 
\begin{aligned} 
\Phi_k(u)
:=
\frac{a(b-k)}{a+b-k}(K_0-\Delta)^2
-
\frac{ab}{a+b}(K_0-\Delta)^2
+
\frac{au}{a+u}(K_0+\Delta)^2
+
(k-u)K_0^2.
\end{aligned}
\end{equation} 
\begin{lem}[Algebraic bound]\label{lem:phi-bound}
For every \(k\in\{0,\dots,b\}\) and every \(u\in[0,k]\),
$ 
\Phi_k(u)\ge kK_0\Delta, 
$ 
under the  bounds derived in Lemma \ref{lem:bound_basic_constants}. 
\end{lem}

\begin{proof} See Appendix \ref{proof:lem:phi-bound}.  
\end{proof}

\paragraph{Step 5: canonicalization.}

 The next lemma shows that, without increasing the loss, we can simplify any rule on a fixed anchor block \(A_g\). We first state the lemma formally, and then provide its interpretation below. 

\begin{lem}[Canonicalization]\label{lem:single-anchor-canonicalization}
Let \((\alpha,\mu)\) satisfy
\(\alpha:\mathcal X_N\to\{1,\dots,G\}\), \(\mu\in[-K,K]^{G-1}\), and suppose that,
for each \(h=2,\dots,G\), the set
$ 
\{x\in\mathcal X_N:\alpha(x)=h\}
$ 
is either empty or a consecutive interval in the induced order on \(\mathcal X_N\). Fix an
anchor block \(A=A_g\), on which \(\tau_\theta(x)= (-1)^g K_0\) for all \(x\in A\). Define
$ 
\mathcal U_\alpha(A)
:=
\{1:\exists x\in A \text{ such that }\alpha(x)=1\}
\cup
\{h\in\{2,\dots,G\}:\{x:\alpha(x)=h\}\cap A\neq\varnothing\},
$ 
and, for \(u\in\mathcal U_\alpha(A)\),
 $ 
\lambda_u = K_0^2 1\{u = 1\} + (\mu_u - (-1)^g K_0)^2 1\{u \in \{2, \cdots, G\}\}$. 
Let
$ 
u^\star\in\arg\min_{u\in\mathcal U_\alpha(A)}\lambda_u  
$ (with any deterministic tie-breaking rule), 
and define
$ 
\widetilde\mu:=\mu,
\widetilde\alpha(x) = \alpha(x) 1\{x \not \in A\} + u^\star 1\{x \in A\}.
$ 
Then, for each \(h=2,\dots,G\), the set
$ 
\{x\in\mathcal X_N:\widetilde\alpha(x)=h\}
$ 
is either empty or a consecutive interval in the induced order on \(\mathcal X_N\), and
\[
\footnotesize 
\begin{aligned} 
\left|\{h\in\{2,\dots,G\}:\{x:\widetilde\alpha(x)=h\}\neq\varnothing\}\right|
\le
\left|\{h\in\{2,\dots,G\}:\{x:\alpha(x)=h\}\neq\varnothing\}\right|.
\end{aligned} 
\]
Moreover, for every \(x\notin A\),
$ 
\widetilde\alpha(x)=\alpha(x)
\quad\text{and}\quad
\widetilde\mu_{\widetilde\alpha(x)}=\mu_{\alpha(x)}
\quad\text{whenever } \alpha(x)\ge2,
$ 
and, for every \(x\in A\),
$ 
\sum_{h=2}^G
(\widetilde\mu_h-(-1)^gK_0)^2\,1\{\widetilde\alpha(x)=h\}
+
K_0^2\,1\{\widetilde\alpha(x)=1\}
=
\min_{u\in\mathcal U_\alpha(A)}\lambda_u
$ 
with
$ 
\min_{u\in\mathcal U_\alpha(A)}\lambda_u
\le
\sum_{h=2}^G
(\mu_h-(-1)^gK_0)^2\,1\{\alpha(x)=h\}
+
K_0^2\,1\{\alpha(x)=1\}.
$ 
Consequently,
$ 
R_\theta(\alpha,\mu)
\ge
R_\theta(\widetilde\alpha,\widetilde\mu).
$ 
\end{lem}

The lemma states that the nonempty predictive cells of \(\widetilde\alpha\) are pairwise disjoint
consecutive intervals, their number is no larger than the number of nonempty predictive
cells of \(\alpha\), and
$ 
R_\theta(\alpha,\mu)\ge R_\theta(\widetilde\alpha,\widetilde\mu).
$ 
Moreover, outside \(A\), every point keeps the same group and the same attached predictive
mean; on \(A\), every point is predicted so to attain  
\(\min_{u\in\mathcal U_\alpha(A)}\lambda_u\).

\begin{proof} See Appendix  \ref{proof:lem:single-anchor-canonicalization}. 
\end{proof} 

We next apply Lemma \ref{lem:single-anchor-canonicalization} recursively. Specifically, fix any admissible rule \((\alpha,\mu) \in \mathcal{G} \times [-K,K]^{G-1}\). As we consider $x_i = i$, any one-dimensional partition $\alpha \in \mathcal{G}$ is such that the nonempty predictive cells of
\(\alpha\) are pairwise disjoint consecutive intervals, and their number is at most
\(G-1\).  Applying Lemma \ref{lem:single-anchor-canonicalization} successively to
\(A_2,A_3,\dots,A_G\), we obtain a rule
\((\widetilde\alpha,\widetilde\mu)\) such that
$ 
R_\theta(\alpha,\mu)\ge R_\theta(\widetilde\alpha,\widetilde\mu).
$ 
Let
$ 
D_1<\cdots<D_{J_{\widetilde\alpha}}
$ 
denote the nonempty sets
$ 
\{x\in\mathcal X_N:\widetilde\alpha(x)=h\}, h\in\{2,\dots,G\},
$ 
ordered from left to right. Then by Lemma \ref{lem:single-anchor-canonicalization} the number of non empty sets satisfies \(J_{\widetilde\alpha}\le G-1\), each \(D_j\) is a
consecutive interval, and every anchor block \(A_g\), \(g=2,\dots,G\), is either entirely
ignored or satisfies \(A_g\subseteq D_j\) for some \(j\). Now re-optimize the means on the fixed intervals. Namely denote 
\[
\footnotesize 
\begin{aligned} 
\nu_{j,\theta}^\star
:=
\frac1{|D_j|}\sum_{x\in D_j}\tau_\theta(x),
\qquad j=1,\dots,J_{\widetilde\alpha}.
\end{aligned} 
\]
By construction of $\nu^\star$,  
$ 
R_\theta(\widetilde\alpha,\widetilde\mu)
\ge
R_\theta(\widetilde\alpha,\nu_\theta^\star).
$
Therefore, we can write  from Lemma \ref{lem:single-anchor-canonicalization}, 
\begin{equation}\label{eq:proof-alpha-tilde-final}
\footnotesize 
\begin{aligned} 
R_\theta(\alpha,\mu)
\ge R_\theta(\widetilde\alpha,\widetilde\mu) \ge 
R_\theta(\widetilde\alpha,\nu_\theta^\star).
\end{aligned} 
\end{equation}
We will therefore provide a lower bound for $(\widetilde\alpha,\nu_\theta^\star)$, which will immediately imply a lower bound for $(\alpha,\mu)$ (note that \(\widetilde\alpha\) does not need to satisfy the minimum-size constraint for the \textit{lower} bound to hold).
For later use, define
\[
\footnotesize 
\begin{aligned} 
L_{\widetilde\alpha}(x) := K_0^2 1\{\widetilde\alpha(x)=1\} + (\nu_{j,\theta}^\star-\tau_\theta(x))^2 1\{x \in D_j\}
\qquad
L_{\mathrm{bench}}(x):=(\mu^\theta_{\alpha^\theta(x)}-\tau_\theta(x))^2.
\end{aligned} 
\]
Then
\begin{equation}\label{eq:proof-loss-difference-pointwise}
\footnotesize 
\begin{aligned} 
R_\theta(\widetilde\alpha,\nu_\theta^\star)-R_\theta(\alpha^\theta,\mu^\theta)
=
\frac1N\sum_{x\in\mathcal X_N}
\bigl(L_{\widetilde\alpha}(x)-L_{\mathrm{bench}}(x)\bigr).
\end{aligned} 
\end{equation}

\paragraph{Step 6: lower bound outside the good structural property.}

Using the intervals \(D_1,\dots,D_{J_{\widetilde\alpha}}\), define the ``good event set'' as 
\[
\footnotesize 
\begin{aligned} 
\widetilde\alpha\in\mathcal B_{\mathrm{good}}
\quad\Longleftrightarrow\quad
\left\{
\begin{array}{l}
\forall g\in\{2,\dots,G\},\ \exists j\in\{1,\dots,J_{\widetilde\alpha}\}
\text{ such that }A_g\subseteq D_j,\\[3pt]
\forall j\in\{1,\dots,J_{\widetilde\alpha}\},\
|\{g\in\{2,\dots,G\}:A_g\subseteq D_j\}|\le 1.
\end{array}
\right.
\end{aligned} 
\]
 We bound the error outside the good event set. 
\begin{lem}\label{lem:structural_property}
If \(\widetilde\alpha\notin\mathcal B_{\mathrm{good}}\), then
\begin{equation}\label{eq:proof-bad-structural}
\footnotesize 
\begin{aligned} 
R_\theta(\widetilde\alpha,\nu_\theta^\star)-R_\theta(\alpha^\theta,\mu^\theta)
\ge
\frac{aK_0^2}{32N}.
\end{aligned} 
\end{equation}
\end{lem}

\begin{proof} See Appendix \ref{proof:lem:structural_property} 
\end{proof}

\paragraph{Step 7: geometry under the good structural property.}

\begin{lem}[Geometry under the good structural property]\label{lem:good_geometry}
Suppose \(\widetilde\alpha\in\mathcal B_{\mathrm{good}}\). Then
\(J_{\widetilde\alpha}=G-1\), and every predictive interval \(D_j\) contains exactly one
anchor block $A_g$. For every \(m=1,\dots,M\), let \(D_m^L\) and \(D_m^R\) denote the unique
predictive intervals satisfying
$ 
A_{2m}\subseteq D_m^L,
A_{2m+1}\subseteq D_m^R.
$ 
Then \(D_m^L\neq D_m^R\), and there exist \(L_m,R_m\subseteq H_m\) such that
$ 
D_m^L=A_{2m}\cup L_m,
D_m^R=R_m\cup A_{2m+1},
L_m\cap R_m=\varnothing,
$ 
and \(\widetilde\alpha(x)=1\) for every \(x\in H_m\setminus(L_m\cup R_m)\). If
\(H_m=\{x_{q_m},\dots,x_{q_m+b-1}\}\), then for some
\(l_m,r_m,s_m\in\{0,\dots,b\}\),
$ 
l_m+r_m+s_m=b,
$, 
$ 
L_m=\{x_{q_m},\dots,x_{q_m+l_m-1}\},
R_m=\{x_{q_m+b-r_m},\dots,x_{q_m+b-1}\},
$ 
with the convention that these sets are empty when \(l_m=0\) or \(r_m=0\), and
$ 
|L_m|=l_m, |R_m|=r_m, 
|H_m\setminus(L_m\cup R_m)|=s_m.
$ 
\end{lem}

\begin{proof}  See Appendix \ref{proof:lem:good_geometry}. 
\end{proof}

\paragraph{Step 8: local lower bound under the good structural property.}

\begin{lem}[Local lower bound]\label{lem:local_good_property}
Suppose \(\widetilde\alpha\in\mathcal B_{\mathrm{good}}\). For each \(m=1,\dots,M\), let
\(D_m^L,D_m^R,L_m,R_m,l_m,r_m,s_m\) be as in Lemma~\ref{lem:good_geometry}, and define
$ 
k_m
:= (r_m + s_m) 1\{\theta_m = 1\} + (l_m + s_m) 1\{\theta_m = -1\}$. 
Then
\begin{equation}\label{eq:proof-sum-local}
\footnotesize 
\begin{aligned} 
R_\theta(\widetilde\alpha,\nu_\theta^\star)-R_\theta(\alpha^\theta,\mu^\theta)
\ge
\frac{K_0\Delta}{N}\sum_{m=1}^M k_m.
\end{aligned} 
\end{equation}
\end{lem}

\begin{proof} See Appendix \ref{proof:lem:local_good_property}.  
\end{proof}

\paragraph{Step 9: reduction to Hamming loss.}

\begin{lem}[Reduction to Hamming loss]\label{lem:hamming_reduction}
Fix any admissible \((\alpha,\mu)\) and construct
\((\widetilde\alpha,\widetilde\mu)\) and \(D_1,\dots,D_{J_{\widetilde\alpha}}\) as in
Step~5. For each \(m=1,\dots,M\), define
$$  
\footnotesize 
\begin{aligned} 
q_m(\widetilde\alpha)
:=
\begin{cases}
|H_m\cap D_j|, & \text{if there exists }j\text{ such that }A_{2m}\subseteq D_j,\\[3pt]
0, & \text{otherwise,}
\end{cases}, \qquad \widehat\theta_m
:=
\begin{cases}
+1, & q_m(\widetilde\alpha)\ge b/2,\\[3pt]
-1, & q_m(\widetilde\alpha)< b/2.
\end{cases}. 
\end{aligned} 
$$ 
where $\widehat \theta_m$ is a measurable function of $(\alpha,\mu)$ by construction in Step 5. 
Then, for \(\widehat\theta=(\widehat\theta_1,\dots,\widehat\theta_M)\),
\begin{equation}\label{eq:proof-pointwise-final}
\footnotesize 
\begin{aligned} 
R_\theta(\alpha,\mu)-R_\theta(\alpha^\theta,\mu^\theta)
\ge
\frac{bK_0\Delta}{2N}\,d_H(\widehat\theta,\theta).
\end{aligned} 
\end{equation}
\end{lem}

\begin{proof} See Appendix \ref{proof:lem:hamming_reduction} 
\end{proof}

\paragraph{Step 10: minimax reduction and Assouad bound.}

Fix any
estimator \((\widehat\alpha,\widehat\mu)\) measurable with respect to $(\hat{\tau}(x),\hat{\eta}(x)^2,\hat{c}(x))_{x\in \mathcal{X}}$, with values in
\(\mathcal G \times[-K,K]^{G-1}\).
Apply
Lemma~\ref{lem:hamming_reduction} to the estimator output. For every
\(\theta\in\{-1,+1\}^M\),
$ 
R_\theta(\widehat\alpha,\widehat\mu)-R_\theta(\alpha^\theta,\mu^\theta)
\ge
\frac{bK_0\Delta}{2N}d_H(\widehat\theta,\theta).
$ 
Taking expectations, suprema, and then infima over estimators on both sides ($\widehat{\theta}$ is a deterministic function of $\widehat \alpha, \widehat \mu$)\footnote{From Lemma \ref{lem:hamming_reduction} $\widehat \theta$ is a deterministic function of $\widehat \alpha, \widehat \mu$ which we obtain from the canonicalization Step 5 here applied to $\widehat \alpha, \widehat \mu$.}
\begin{equation}\label{eq:proof-step8-reduction}
\footnotesize 
\begin{aligned} 
\inf_{(\widehat\alpha,\widehat\mu)}
\sup_{\theta\in\{-1,+1\}^M}
\mathbb E_{P_\theta}\!\left[
R_\theta(\widehat\alpha,\widehat\mu)-R_\theta(\alpha^\theta,\mu^\theta)
\right]
\ge
\frac{bK_0\Delta}{2N}
\inf_{\widetilde\theta}
\sup_{\theta\in\{-1,+1\}^M}
\mathbb E_{P_\theta}\!\left[d_H(\widetilde\theta,\theta)\right],
\end{aligned} 
\end{equation}
where the infimum on the right is over all measurable
\(\widetilde\theta:\mathbb R^N\to\{-1,+1\}^M\).

For \(m=1,\dots,M\), let \(\theta^{(m)}\) be obtained from \(\theta\) by flipping the
\(m\)-th coordinate. The laws \(P_\theta\) and \(P_{\theta^{(m)}}\) differ only on the
\(b\) observations in \(H_m\), where the mean changes by \(2\Delta\). Hence, by denoting $\mathrm{KL}$ the KL-divergence, a direct calculation gives us
$ 
\mathrm{KL}(P_\theta,P_{\theta^{(m)}})
=
b\frac{(2\Delta)^2}{2\underline\eta^2}
=
\frac{2b\Delta^2}{\underline\eta^2}
=
\frac1{32}.
$ 
By Lemma~\ref{lem:assouad}, with \(m=M\) and \(\kappa=1/32\),
$ 
\inf_{\widetilde\theta}
\sup_{\theta\in\{-1,+1\}^M}
\mathbb E_{P_\theta}[d_H(\widetilde\theta,\theta)]
\ge
\frac M2\left(1-\sqrt{\frac{1/32}{2}}\right)
=
\frac{7M}{16}.
$ 
Substituting into \eqref{eq:proof-step8-reduction},
\begin{equation}\label{eq:proof-before-rate}
\footnotesize 
\begin{aligned} 
\inf_{(\widehat\alpha,\widehat\mu)}
\sup_{\theta\in\{-1,+1\}^M}
\mathbb E_{P_\theta}\!\left[
R_\theta(\widehat\alpha,\widehat\mu)-R_\theta(\alpha^\theta,\mu^\theta)
\right]
\ge
\frac{7}{32}\frac{MbK_0\Delta}{N}.
\end{aligned} 
\end{equation}

Since \(\Delta=\underline\eta/(8\sqrt b)\), \(M=(G-2)/2\), and
\(b\ge N/[256(G-1)]\),
$ 
\frac{MbK_0\Delta}{N}
=
\frac{K_0\underline\eta}{8}\frac{M\sqrt b}{N}
\ge
\frac{K_0\underline\eta}{256}\frac{G-2}{\sqrt{N(G-1)}}.
$ In addition, 
since \(G\ge4\),
$ 
\frac{G-2}{\sqrt{G-1}}
\ge
\frac1{\sqrt2}\sqrt{G-2}.
$ 
Therefore, for a universal \(c>0\),
\begin{equation}\label{eq:proof-factor-rate}
\footnotesize 
\begin{aligned} 
\frac{MbK_0\Delta}{N}
\ge
cK_0\underline\eta\sqrt{\frac{G-2}{N}}.
\end{aligned} 
\end{equation}
Combining \eqref{eq:proof-before-rate} and \eqref{eq:proof-factor-rate},
$ 
\inf_{(\widehat\alpha,\widehat\mu)}
\sup_{\theta\in\{-1,+1\}^M}
\mathbb E_{P_\theta}\!\left[
R_\theta(\widehat\alpha,\widehat\mu)-R_\theta(\alpha^\theta,\mu^\theta)
\right]
\ge
c_0K_0\underline\eta\sqrt{\frac{G-2}{N}}.
$ 

By \eqref{eq:proof-benchmark-admissible}, 
$ 
\inf_{(\widehat\alpha,\widehat\mu)}
\sup_{\theta\in\{-1,+1\}^M}
\mathbb E_{P_\theta}\!\left[
R_\theta(\widehat\alpha,\widehat\mu)
-
\inf_{\alpha\in\mathcal G,\,
\mu\in[-K,K]^{G-1}}
R_\theta(\alpha,\mu)
\right]
\ge
c_0K_0\underline\eta\sqrt{\frac{G-2}{N}}.
$ 
Therefore, there exists a $\theta^\star \in \{-1,1\}^M$ so that for $P = P_{\theta^\star}$ the lower bound in Equation \eqref{eqn:lower_bound} is attained, since under the chosen design and $P_\theta$ $R_c(f_{\alpha,\mu},\pi^{\alpha}) = R_\theta(\alpha,\mu)$. 
\paragraph{Step 11: \(P_\theta\) belongs to the model class in the theorem.}
We want to show that for each
\(\theta\in\{-1,+1\}^M\), the law \(P_\theta\) constructed in Step~2 belongs to the
model class in Assumption~\ref{ass:general_regret}. 
The random variables are independent across \(i\), because the noises
\(\xi_i\) are independent and the coordinates
\(\widehat\eta(x_i)^2\) and \(\widehat c(x_i)\) are deterministic. Also Assumption~\ref{ass:general_regret}(i) follows directly by construction. We next verify Assumption~\ref{ass:general_regret}(ii). Let \(Z\sim N(0,1)\). For
\(u\in[0,1/2]\),
$ 
4-4u\in[2,4],
$ 
and hence
$ 
\sqrt{
\mathbb E_{P_\theta}
\left[
\left|
\widehat\tau(x_i)
-
\mathbb E_{P_\theta}[\widehat\tau(x_i)]
\right|^{4-4u}
\right]
}
=
\sqrt{\mathbb E|Z|^{4-4u}}\,
\underline\eta^{\,2-2u}
<\infty.
$ 
Also,
$ 
\sqrt{
\mathbb E_{P_\theta}
\left[
\left|
\widehat\tau(x_i)
-
\mathbb E_{P_\theta}[\widehat\tau(x_i)]
\right|^{6}
\right]
}
=
\sqrt{\mathbb E|Z|^6}\,\underline\eta^3
=
\sqrt{15}\,\underline\eta^3
<\infty.
$ 
Furthermore, because \(\widehat\eta(x_i)^2=\eta(x_i)^2=\underline\eta^2\),
$ 
\mathbb E_{P_\theta}
\left[
\left|
\widehat\eta(x_i)^2-\eta(x_i)^2
\right|^{2-2u}
\right]
=
0$,
$\mathbb E_{P_\theta}
\left[
\left|
\widehat\eta(x_i)^2-\eta(x_i)^2
\right|^{3}
\right]
=
0.
$ 
Likewise, because \(\widehat c(x_i)=c(x_i)=K_0^2\) deterministically,
$ 
\mathbb E_{P_\theta}
\left[
\left|
\widehat c(x_i)-c(x_i)
\right|^{2-2u}
\right]
=
0,
\mathbb E_{P_\theta}
\left[
\left|
\widehat c(x_i)-c(x_i)
\right|^{3}
\right]
=
0.
$ 
Finally, by construction,
$ 
|\tau_\theta(x_i)|
\le
\max\{K_0,\Delta\}
\le
K_0
\le
K
<\infty,
$ 
where we used \(\Delta\le K_0\) (so we can take $B  = \max\{1,K\}$ in Assumption \ref{ass:general_regret}(ii)).  
Consequently, for each \(u\in[0,1/2]\), Assumption~\ref{ass:general_regret}(ii) holds
with the finite constant
$ 
M_u
:=
\max\left\{
\sqrt{\mathbb E|Z|^{4-4u}}\,
\underline\eta^{\,2-2u},
\sqrt{15}\,\underline\eta^3
\right\}, 
$ 
where $\underline{\eta} < \infty$ by construction (therefore implying that we can find a finite constant $M_u < \infty$ for any $u$). 
Finally, Assumption~\ref{ass:general_regret}(iii) follows from the choice
$ 
p(x_i)=\frac1N
=
\frac1{|\mathcal X_N|}.
$ 
Therefore Assumption~\ref{ass:general_regret}(iii) holds with \(\bar p=1\).

\paragraph{Step 12: modification for odd \(G\ge5\).}

For odd \(G\ge5\), set \(M=\lfloor(G-2)/2\rfloor=(G-3)/2\). Then
\(2M+1=G-2\) and \(2M+2=G-1\). The partition in Step~1 becomes
$ 
A_2,\ H_1,\ A_3,\ A_4,\ H_2,\ A_5,\ \dots,\ A_{2M},\ H_M,\ A_{2M+1},\ A_{2M+2},\dots,A_G,
$ 
so the leftover anchors are \(A_{G-1}\) and \(A_G\). All definitions and arguments on the
coded blocks \(B_m=A_{2m}\cup H_m\cup A_{2m+1}\) are unchanged. In
Lemma~\ref{lem:structural_property} and Lemma~\ref{lem:local_good_property}, the only change is that the final summation over leftover anchors is over
\(\{A_{2M+2},\dots,A_G\}=\{A_{G-1},A_G\}\) rather than just \(A_G\); these anchors have
zero benchmark loss when singleton under \(\alpha^\theta\) and nonnegative loss under
\((\widetilde\alpha,\nu_\theta^\star)\), so the same inequalities hold.

The Hamming cube dimension is now \(M=(G-3)/2\), so Step 10 yields the same lower
bound with \(G-3\) in place of \(G-2\). Since \(G\ge5\),
$ 
G-3\ge \frac{G-2}{3}.
$ 
Absorbing the factor \(1/\sqrt3\) into the universal constant gives the same rate
\(c_0K_0\underline\eta\sqrt{(G-2)/N}\). This completes the proof for all \(G\ge4\).

\subsection{Proof of Theorem \ref{thm:model_selection}} \label{proof:thm:model_selection} 

\paragraph{Step 1 and 2} For each \(j\in\{1,\dots,J\}\),
$
V_j:=R_c(\hat f_j,\hat\pi_j),
\quad
\widehat V_j:=\widehat R_{oos}(\hat f_j,\hat\pi_j),
\quad
O_j:=R_c(f_j^\star,\pi_j^\star), \quad 
R_{all}^\star
:=
\min_{f\in\bar{\mathcal F}}
\sum_{x\in\mathcal X} p(x)\bigl(f(x)-\tau(x)\bigr)^2.
$ 
As the first step, fix any \(j\in\{1,\dots,J\}\). Since \(\hat j\) minimizes the validation criterion,
$
\widehat V_{\hat j}\le \widehat V_j
$, which implies  
$$  
\footnotesize 
\begin{aligned} 
V_{\hat j}-R_{all}^\star
& =
(V_{\hat j}-\widehat V_{\hat j})
+
(\widehat V_{\hat j}-\widehat V_j)
+
(\widehat V_j-V_j)
+
(V_j-O_j)
+
(O_j-R_{all}^\star)
\\ & \le
(V_{\hat j}-\widehat V_{\hat j})
+
(\widehat V_j-V_j)
+
(V_j-O_j)
+
(O_j-R_{all}^\star) \le
2\max_{1\le \ell\le J} |\widehat V_\ell-V_\ell|
+
(V_j-O_j)
+
(O_j-R_{all}^\star).
\end{aligned} 
$$ 
Taking expectations and infimum over \(j\) (since the inequality holds for any $j$),
\begin{equation} \label{eqn:modelsel_decomp}
\footnotesize 
\begin{aligned} 
\mathcal R_c
\le
2\,\mathbb E\!\left[\max_{1\le \ell\le J} |\widehat V_\ell-V_\ell|\right]
+
\inf_{1\le j\le J}
\left\{
\mathbb E[V_j-O_j]+(O_j-R_{all}^\star)
\right\}.
\end{aligned} 
\end{equation}

\paragraph{Step 3: Approximation error bound} We first bound the approximation term \(O_j-R_{all}^\star\). By definition,
$
O_j
=
R_c(f_j^\star,\pi_j^\star)
=
L_{\mathcal F_j}^{\mathrm{loc}}(\pi_j^\star)+A_c(\pi_j^\star).
$ 
Also, since \(\pi_j^\star(x)\in\{0,1\}\),
$$  
\footnotesize 
\begin{aligned} 
L_{\bar{\mathcal F}}^{\mathrm{loc}}(\pi_j^\star)
=
\min_{f\in\bar{\mathcal F}}
\sum_{x\in\mathcal X} p(x)\bigl(f(x)-\tau(x)\bigr)^2\pi_j^\star(x)
\le
\min_{f\in\bar{\mathcal F}}
\sum_{x\in\mathcal X} p(x)\bigl(f(x)-\tau(x)\bigr)^2
=
R_{all}^\star.
\end{aligned} 
$$ 
Therefore
$ 
O_j-R_{all}^\star
\le
L_{\mathcal F_j}^{\mathrm{loc}}(\pi_j^\star)
-
L_{\bar{\mathcal F}}^{\mathrm{loc}}(\pi_j^\star)
+
A_c(\pi_j^\star)
=
\Delta_j^{\mathrm{loc}}(\pi_j^\star)+A_c(\pi_j^\star).
$ 

\paragraph{Step 4: Estimation error bound} Next, for the estimation error term, Theorem \ref{thm:regret} applied to the training sample with class pair \((\mathcal F_j,\Pi_j)\) yields
$ 
\mathbb E[V_j-O_j]
=
\mathbb E\!\left[
R_c(\hat f_j,\hat\pi_j)-\min_{\pi\in\Pi_j,\ f\in\mathcal F_j}R_c(f,\pi)
\right]
\le
c_0\sqrt{\tilde M_u\,\frac{C(\mathcal F_j)+v_{\Pi_j}}{|\mathcal X|}},
$ 
with $c_0$ as bounded in Theorem \ref{thm:regret}. 

\paragraph{Step 5: model selection bound} It remains to control the model-selection term
$ 
\mathbb E\!\left[\max_{1\le \ell\le J} |\widehat V_\ell-V_\ell|\right].
$ 
To do so, we use the law of iterated expectations where we first condition on the training sample. 
Conditional on the training sample, each candidate \((\hat f_\ell,\hat\pi_\ell)\) is fixed, and the randomness in \(\widehat V_\ell\) comes only from the validation sample. Writing
$
\widehat V_\ell-V_\ell
=
\sum_{i=1}^{|\mathcal X|} p(x_i) \xi_{\ell,i},
$ 
where
$$  
\footnotesize 
\begin{aligned} 
\xi_{\ell,i}
:=
\Big(
\big[(\hat f_\ell(x_i)-\hat\tau_{oos}(x_i))^2-\hat\eta_{oos}(x_i)^2-(\hat f_\ell(x_i)-\tau(x_i))^2\big]\hat\pi_\ell(x_i)
+
(\hat c_{oos}(x_i)-c(x_i))(1-\hat\pi_\ell(x_i))
\Big),
\end{aligned} 
$$ 
we note that \(\mathbb E[\xi_{\ell,i}\mid (\hat f_j,\hat\pi_j)_{j=1}^J ]=0\) by Assumption \ref{ass:general_regret}(i) (which also applies to the out-of-sample estimates), and independence of the in-sample and out-of-sample estimates. By Lemma \ref{lem:centered_square}, together with Assumption \ref{ass:general_regret}(ii) (both applied to the out-of-sample estimates), combined with the uniform bounded $K$ on $f \in \mathcal{F}$, there exists a universal constant \(c_1<\infty\) such that\footnote{
To see this, condition on \((\hat f_j,\hat\pi_j)_{j=1}^J\). Then
\(\hat f_\ell(x_i)\) and \(\hat\pi_\ell(x_i)\) are fixed, with
\(|\hat f_\ell(x_i)|\le K\) and \(\hat\pi_\ell(x_i)\in\{0,1\}\). Expanding the square gives
$ 
(\hat f_\ell(x_i)-\hat\tau_{oos}(x_i))^2
-\hat\eta_{oos}(x_i)^2
-(\hat f_\ell(x_i)-\tau(x_i))^2
=
\big(\hat\tau_{oos}(x_i)^2-\mathbb E[\hat\tau_{oos}(x_i)^2]\big)
-2\hat f_\ell(x_i)\big(\hat\tau_{oos}(x_i)-\tau(x_i)\big)
-\big(\hat\eta_{oos}(x_i)^2-\eta(x_i)^2\big),
$ 
since
\(\mathbb E[\hat\tau_{oos}(x_i)^2]=\tau(x_i)^2+\eta(x_i)^2\).
Using \(\hat\pi_\ell(x_i)\in\{0,1\}\), \((1-\hat\pi_\ell(x_i))\in\{0,1\}\), and
\((a_1+\cdots+a_4)^2\le 4(a_1^2+\cdots+a_4^2)\), we obtain
\[
|\xi_{\ell,i}|^2
\le c_0\Big(
\big|\hat\tau_{oos}(x_i)^2-\mathbb E[\hat\tau_{oos}(x_i)^2]\big|^2
+
K^2|\hat\tau_{oos}(x_i)-\tau(x_i)|^2
+
|\hat\eta_{oos}(x_i)^2-\eta(x_i)^2|^2
+
|\hat c_{oos}(x_i)-c(x_i)|^2
\Big)
\]
for a universal constant \(c_0<\infty\). By Lemma \ref{lem:centered_square} with
\(u=0\), the first term has expectation bounded by
\(cB^3(M_0+M_0^2)\). By Assumption \ref{ass:general_regret}(ii), the remaining terms have expectations bounded by constants times
\(K^2M_0\), \(M_0\), and \(M_0\), respectively. Since \(M_0\le M_0+M_0^2\), this yields
$ 
\mathbb E\big[|\xi_{\ell,i}|^2\mid(\hat f_j,\hat\pi_j)_{j=1}^J\big]
\le
c_1(1+B^3+K^2)(M_0+M_0^2),
$ 
for a universal constant \(c_1<\infty\). 
}, 
$ 
\mathbb E\big[|\xi_{\ell,i}|^2 \mid (\hat f_j,\hat\pi_j)_{j=1}^J \big]
\le
c_1 (1 + B^3 +  K^2)  (M_0 + M_0^2).
$ 
Hence,  using independence of the validation sample across $i$ 
$
\mathbb E\big[|\widehat V_\ell-V_\ell|^2 \mid (\hat f_j,\hat\pi_j)_{j=1}^J  \big] = \sum_{i=1}^{N_{\mathcal{X}}} p(x_i)^2 \mathbb E\big[|\xi_{\ell,i}|^2 \mid (\hat f_j,\hat\pi_j)_{j=1}^J \big]
\le
c_1(1 + B^3 + K^2)  \frac{(M_0 + M_0^2) \bar{p} }{|\mathcal X|} 
$
Using the basic inequality 
$ 
\max_{1\le \ell\le J} |a_\ell|
\le
\left(\sum_{\ell=1}^J |a_\ell|^2\right)^{1/2},
$  and
Jensen's inequality  
$
\mathbb E\!\left[\max_{1\le \ell\le J} |\widehat V_\ell-V_\ell| \mid (\hat f_j,\hat\pi_j)_{j=1}^J  \right] \le
\left(
\sum_{\ell=1}^J
\mathbb E\big[|\widehat V_\ell-V_\ell|^2 \mid (\hat f_j,\hat\pi_j)_{j=1}^J \big]
\right)^{1/2}$ \\ 
$\le
\sqrt{\bar{p} c_1 (1 + B^3 + K^2)} J^{1/2} \sqrt{\frac{M_0 + M_0^2}{N_{\mathcal{X}}}} \  
$ 
for a universal constant $c_1 < \infty$.

\subsection{Proof of Theorem \ref{thm:breakevena}} \label{proof:thm:breakevena}

Before deriving the theorem, we introduce some notation. Let 
$ 
\xi_x^{(m)}
=
(\frac{\hat{\eta}_{oos}^2(x)}{m+1} - \omega_m(x)) (1 - \hat\pi^{(m)}(x))
$ 
and similarly denote 
$ 
q_m^2 := |\mathcal X|\sum_{x\in\mathcal X} p(x)^2
\mathbb{V}\!\left(\xi_x^{(m)}\mid \mathcal{T}\right)  
$ 
Note that because $p(x)^2$ is of order $1/N_{\mathcal{X}}^2$, we expect $q_m^2 = O(1)$. 
We can write 
$ 
\widehat{\Delta}_\eta(m) = \sum_x p(x) \xi_x^{(m)}, \Delta_\eta(m)
:= \sum_x p(x) \Big\{\frac{\eta(x)^2}{m+1} - \omega_m(x)\Big\}(1 - \hat\pi^{(m)}(x))
.
$ 
Note that because $\mathbb{E}[\hat{\Delta}_\eta(m) | \mathcal{T}] = \Delta_\eta(m)$. To establish asymptotic normality, we first derive the following lemma. 

\begin{lem}[Marginal asymptotic normality of the break-even statistic] \label{prop:clt_break_evena} Let Assumptions \ref{ass:general_regret}, \ref{ass:entropy_F} hold.
Take any $m \in \mathcal{M}$ and uniformly bounded $\omega_m(x)$ which are both $\mathcal{T}$-measurable. Let $q_m^2 > \underline{l}$ for a constant $\underline{l} > 0$ almost surely.  Then
$ 
\frac{\sqrt{|\mathcal X|} \Big(\widehat\Delta_\eta(m)-\mathbb{E}[\widehat{\Delta}_\eta(m) | \mathcal{T}]\Big)}{q_m} 
\overset{d}{\longrightarrow}N(0,1).
$ 
\end{lem}

\begin{proof}[Proof of Lemma \ref{prop:clt_break_even}]
We will write $\xi_x := \xi_x^{(m)}$ whenever clear from the context. 
By Assumption \ref{ass:general_regret} (independence), and independence of the training and out-of-sample estimates, the random variables 
$ 
p(x)\Big(\xi_x-\mathbb E[\xi_x\mid \mathcal{T}]\Big),
$ 
are conditionally independent and centered given $\mathcal{T}$. Therefore, we establish the central limit theorem using Lyapunov central limit theorem conditional on $\mathcal{T}$. 
We first verify a uniform third-moment bound. Since $\omega_m(x)$ is uniformly bounded by assumption, by Assumption \ref{ass:general_regret}(ii), there exists a finite constant $c_0 < \infty$, so that 
$ 
\mathbb E\Big[\big|\xi_x-\mathbb E[\xi_x\mid \mathcal{T}]\big|^3
\mid \mathcal{T}\Big]
\le
c_1 < \infty. 
$ 
By Assumption \ref{ass:general_regret}(iii),
$
\sum_{x\in\mathcal X} p(x)^3
\le
\max_{x\in\mathcal X}p(x)\sum_{x\in\mathcal X}p(x)^2
\le
\frac{\bar p}{|\mathcal X|}\cdot \frac{\bar p}{|\mathcal X|}
=
\frac{\bar p^2}{|\mathcal X|^2}.
$
Therefore
$ 
\sum_{x\in\mathcal X}
\mathbb E\Big[
\big|p(x)\big(\xi_x-\mathbb E[\xi_x\mid \mathcal{T}]\big)\big|^3
\mid \mathcal{T}
\Big]
\le
c_1
\sum_{x\in\mathcal X} p(x)^3
\le
c_1\bar p^2 |\mathcal X|^{-2}.
$ 
Since \(q_m^2 > \underline{l}\) for a positive constant $\underline{l} > 0$ by assumption, it follows that
$
\frac{
|\mathcal{X}|^{3/2} \sum_{x\in\mathcal X}
\mathbb E\Big[
\big|p(x)\big(\xi_x-\mathbb E[\xi_x\mid \mathcal{T}]\big)\big|^3
\mid \mathcal{T}
\Big]
}{
q_m^{3}
}
\longrightarrow 0.
$ 
Thus the Lyapunov condition holds conditionally on $\mathcal{T}$. Therefore, it follows that as we define $Z_{|\mathcal{X}|} = \frac{\sqrt{|\mathcal X|} \Big(\widehat\Delta_\eta(m)-\mathbb{E}[\widehat{\Delta}_\eta(m) | \mathcal{T}]\Big)}{q_m}$, we have $P(Z_{|\mathcal{X}|} \le t|\mathcal{T}) \rightarrow \Phi(t)$ where $\Phi(t)$ is the Gaussian CDF. Integrating over \(\mathcal T\) over both sides and applying the dominated convergence theorem gives the
unconditional convergence stated in the lemma.
\end{proof}

\begin{lem} \label{prop:breakeven_fixed_ma}
Let Assumptions \ref{ass:general_regret}, \ref{ass:entropy_F} hold, and let \(\alpha\in(0,1/2)\). Let \(Q_m^2\) be any (possibly random) bound such that
$ 
\lim_{N_{\mathcal{X}} \rightarrow \infty} P\!\left(q_m^2 \le Q_m^2\right) = 1,
$ and $q_m^2 > \underline{l} > 0$
for any $m \in \mathcal{M}$, with $|\mathcal{M}| < \infty$.  Let $\omega_m(x)$ be uniformly bounded and $\mathcal{T}$-measurable. 
Then for any $m \in \mathcal{M}$ which is $\mathcal{T}$-measurable, 
$ 
\liminf_{|\mathcal X|\to\infty}
\Pr\!\left(
\Delta_\eta(m)
\le
\widehat\Delta_\eta(m)+\frac{Q_m\,\Phi^{-1}(1-\alpha)}{\sqrt{|\mathcal X|}}
\right)
\ge
1-\alpha.
$ 
\end{lem}

\begin{proof}
Write
$
z_{1-\alpha}:=\Phi^{-1}(1-\alpha)>0.
$ 
First note that $\mathbb{E}[\widehat{\Delta}_\eta(m) | \mathcal{T}] = \Delta_\eta(m)$ by sample splitting, with $\widehat{\Delta}_\eta(m)$ defined in Equation \eqref{eqn:delta_hat_f}. By Lemma \ref{prop:clt_break_even},
$ 
\frac{\sqrt{|\mathcal X|}\bigl(\Delta_\eta(m)-\widehat\Delta_\eta(m)\bigr)}{q_m}
\overset{d}{\longrightarrow} N(0,1),
$ 
for any $m$ which is $\mathcal{T}$-measurable. 
Let
$ 
A_{N_{\mathcal{X}}}:=\{Q_m \ge q_m\}, 
$ 
and $A_{N_{\mathcal{X}}}^c$ its complement. Similarly, denote \\ $B_{N_{\mathcal{X}}} := \left\{\Delta_\eta(m) \le \widehat\Delta_\eta(m)+\frac{q_mz_{1-\alpha}}{\sqrt{|\mathcal X|}} \right\}$ and $B_{N_{\mathcal{X}}}^c$ its complement. 
The assumption
$ 
P\!\left(Q_{m'}^2 \ge q_{m'}^2\right)\to 1
$ 
for any $m' \in \mathcal{M}$, and the fact that $|\mathcal{M}|$ implies that 
$P\!\left(\cup_{m' \in \mathcal{M}}\{Q_{m'}^2 < q_{m'}^2\}\right) \le \sum_{m' \in \mathcal{M}} P\!\left(\{Q_{m'}^2 < q_{m'}^2\}\right)\to 0$, which 
in turn implies 
$ 
P(A_{N_{\mathcal{X}}}) \to 1.
$ 
Define 
$
E_{N_{\mathcal{X}}}
:=
\left\{
\Delta_\eta(m)
\le
\widehat\Delta_\eta(m)+\frac{Q_m z_{1-\alpha}}{\sqrt{|\mathcal X|}}
\right\}.
$ 
It follows 
$ 
E_{N_{\mathcal{X}}}
\supseteq
B_{N_{\mathcal{X}}}
\cap
A_{N_{\mathcal{X}}}.
$ 
Therefore, because $P(E) \ge P(B \cap A) = P(B) - P(B \cap A^c) \ge P(B) - P(A^c)$ and 
taking \(\liminf\) on both sides gives
$ 
\liminf_{|\mathcal X|\to\infty}\Pr(E_{N_{\mathcal{X}}})
\ge
\liminf_{|\mathcal X|\to\infty}
\Pr\!\left(
B_{N_{\mathcal{X}}}\right) - \limsup_{N_{\mathcal{X}}} P(A_{N_{\mathcal{X}}}^c) =
1-\alpha,
$ 
where the first term converges to \(1-\alpha\) by asymptotic normality, and the second term converges to \(0\) by assumption.  
\end{proof}

\paragraph{Proof of Theorem \ref{thm:breakevena}}  Write the ordered grid as
$
\mathcal M=\{m_1<\cdots<m_{|\mathcal M|}\},
$ 
and define
$ 
T_j:=\widehat\Delta_\eta(m_j)+\Phi^{-1}(1-\alpha)\frac{Q_{m_j}}{\sqrt{|\mathcal X|}},
U_j:=\max_{l\ge j} T_l.
$ 
By construction, \(U_j\) is weakly decreasing in \(j\), and
$ 
\hat j^\star=\min\{j:U_j<0\},
\hat m^\star=m_{\hat j^\star},
$ 
with the convention \(\hat j^\star=\infty, m_{\infty} = \infty\) if the set is empty. Let
$ 
\mathcal J_0:=\{j\in\{1,\dots,|\mathcal M|\}:\Delta(m_j)\ge 0\}.
$ 
If \(\mathcal J_0=\varnothing\), then \(\Delta_\eta(m_j)<0\) for every \(j\), and the result is immediate. Suppose therefore that \(\mathcal J_0\neq\varnothing\), and define
$ 
j_0:=\max \mathcal J_0.
$ 
Thus \(j_0\) is the last grid point for which the population loss difference is nonnegative. Note that $\mathcal{J}_0$ and so $j_0$ is $\mathcal{T}$-measurable. Consider the event
$ 
\{\hat m^\star<\infty,\ \Delta_\eta(\hat m^\star)\ge 0\}.
$ 
If this event occurs, then necessarily \(\hat j^\star\in \mathcal J_0\), and hence \(\hat j^\star\le j_0\). Since \(U_{\hat j^\star}<0\) by definition of \(\hat j^\star\), we have
$ 
T_l<0 \text{ for all } l\ge \hat j^\star.
$ 
In particular, because \(j_0\ge \hat j^\star\),
$ 
T_{j_0}<0 \Rightarrow \widehat\Delta(m_{j_0})+\Phi^{-1}(1-\alpha)\frac{Q_{m_{j_0}}}{\sqrt{|\mathcal X|}}<0. 
$
Since \(\Delta_\eta(m_{j_0})\ge 0\), it follows that
$ 
\{\hat m^\star<\infty,\ \Delta_\eta(\hat m^\star)\ge 0\}
\subseteq
\left\{
\Delta_\eta(m_{j_0})
>
\widehat\Delta_\eta(m_{j_0})+\Phi^{-1}(1-\alpha)\frac{Q_{m_{j_0}}}{\sqrt{|\mathcal X|}}
\right\}.
$ 
Therefore, 
$ 
\Pr\!\left(
\hat m^\star<\infty,\ \Delta_\eta(\hat m^\star)\ge 0
\right)
\le
\Pr\!\left(
\Delta_\eta(m_{j_0})
>
\widehat\Delta_\eta(m_{j_0})+\Phi^{-1}(1-\alpha)\frac{Q_{m_{j_0}}}{\sqrt{|\mathcal X|}}
\right).
$ 
Lemma \ref{prop:breakeven_fixed_ma} applied at \(m=m_{j_0}\) (where $m_{j_0}$ is measurable with respect to $\mathcal{T}$ since $j_0$ is $\mathcal{T}$-measurable), yields 
$ 
\limsup_{|\mathcal X|\to\infty}
\Pr\!\left(
\Delta_\eta(m_{j_0})
>
\widehat\Delta_\eta(m_{j_0})+\Phi^{-1}(1-\alpha)\frac{Q_{m_{j_0}}}{\sqrt{|\mathcal X|}}
\right)
\le
\alpha.
$

\subsection{Proof of Proposition \ref{prop:complexity}} \label{proof:prop:complexity}

\paragraph{Exactness.}
We prove the claim by induction on the tree depth \(L\).

For the base case \(L=1\), the tree consists of a single split and two terminal leaves.
Algorithm~\ref{alg:alg3} enumerates all admissible coordinate-threshold splits and all
admissible assignments of the two leaves to prediction or to the basin of ignorance.
For each such configuration, it evaluates the empirical criterion
\(\widehat R(f,\pi)\) exactly. Hence, when \(L=1\), the algorithm minimizes \(\widehat R(f,\pi)\) exactly over
all admissible trees of depth \(1\).

Now suppose the claim holds for trees of depth \(\ell-1\), and consider a tree of depth
\(\ell\). Any such tree is obtained by choosing a root split and then attaching a left
subtree and a right subtree, each of depth \(\ell-1\). Once the root split is
fixed, the empirical criterion decomposes as the sum of the contributions from the
observations sent to the left and right child nodes. Therefore, conditional on the root
split, minimizing \(\widehat R(f,\pi)\) over trees of depth \(\ell\) reduces to two
independent minimization problems, one on the left node and one on the right node. Because the problem on the left and right node consist of solving the optimal tree of depth $\ell - 1$, 
by the induction hypothesis, the recursive calls in Algorithm~\ref{alg:alg3} solve
these two subproblems exactly. It follows that, for every admissible root split, the
algorithm computes the exact minimizer among all trees of depth \(\ell\) having that
root split. Since the algorithm enumerates all admissible root splits, it returns a global minimizer.

\paragraph{Complexity.}  
We follow the same recursive argument. For \(L=1\), the
algorithm considers all admissible one-split trees. There are \(d\) features and at most
\(|\mathcal X|\) relevant thresholds for each feature, so the number of candidate splits
is of order \(d|\mathcal X|\). Evaluating all candidate splits at the root has total cost
\(O(d|\mathcal X|^2)\), since computing the corresponding left and right leaf losses over
all candidate splits can be done by summing over at most \(|\mathcal X|\) observations for
each split. Hence the total complexity when \(L=1\) is
$ 
O(d|\mathcal X|^2).
$ Now suppose that, for depth \(\ell-1\), the algorithm has complexity
$ 
O\!\left(2^{\ell-2}d^{\ell-1}|\mathcal X|^{\ell}\right).
$ 
Consider depth \(\ell\). The algorithm first enumerates all admissible root splits. The
number of such choices is of order \(d|\mathcal X|\), while the total local cost of
evaluating the splits at the root is \(O(d|\mathcal X|^2)\). For each root split, the
algorithm makes one recursive call on the left child node and one recursive call on the
right child node. Each such recursive call has depth at most \(\ell-1\), so by the
induction hypothesis its cost is bounded by
$ 
O\!\left(2^{\ell-2}d^{\ell-1}|\mathcal X|^{\ell}\right).
$ 
Therefore the total complexity at depth \(\ell\) is bounded by
$ 
O(d|\mathcal X|^2)
+
O(d|\mathcal X|)\cdot
2\,
O\!\left(2^{\ell-2}d^{\ell-1}|\mathcal X|^{\ell}\right).
$ 
Solving this recursion gives a complexity of order 
$ 
O\!\left(2^{\ell-1}d^\ell|\mathcal X|^{\ell+1}\right).
$ 
Applying this with \(\ell=L\) gives the complexity bound
$ 
O\!\left(2^{L-1}d^L|\mathcal X|^{L+1}\right) \le O\!\left(2^{L+1} d^{L+1}|\mathcal X|^{L+1}\right) .
$ 
In particular, for fixed \(L\) the algorithm is polynomial in \((d,|\mathcal X|)\),
and since \(L=\log_2(G-1)\), it is quasi-polynomial in \((d,|\mathcal X|,G)\).

\newpage

\section{Additional tables, figures and comparison}

\subsection{Out-of-sample comparisons with standard benchmarks} \label{sec:oos_comparisons}

We compare our method with standard approaches for estimating treatment-effect heterogeneity, including generalized random forest, CART, and Bayes procedures (with CART and linear regression models) without abstention option in our application. 
Let $\hat\tau_{ih}$ denote the estimated effect for type $i$ and outcome $h$, and let $\hat\eta_{ih}^2$ denote its estimated sampling variance. First, for generalized random forests, we fit a separate default regression forest for each outcome $h$, using the covariate vector $x_i$ to predict $\hat\tau_{ih}$. Second, for CART, we estimate a pooled-site regression tree separately for each outcome using the original covariates augmented with site dummies, with precision weights $1/\hat\eta_{ih}^2$. Third, we consider a CART-centered empirical-Bayes estimator that assumes
$ 
\hat\tau_{ih}\mid \theta_{ih}\sim N(\theta_{ih},\hat\eta_{ih}^2),  
\theta_{ih}\mid x_i,s(i)\sim N\!\bigl(\hat t_h(x_i),\sigma_{s(i)h}^2\bigr),
$ 
where $s(i)$ indexes the site of type $i$. For each training fold and outcome $h$, we estimate the site-specific prior variance $\sigma_{sh}^2$ by maximizing the Gaussian marginal likelihood and $\hat t$ is the CART mean estimate.
The resulting posterior mean is
$ 
\hat\theta^{EB}_{ih}
=
\hat t_h(x_i)+w_{ih}\bigl(\hat\tau_{ih}-\hat t_h(x_i)\bigr),
w_{ih}
=
\frac{\hat\sigma_{s(i)h}^2}{\hat\sigma_{s(i)h}^2+\hat\eta_{ih}^2}, 
$ whose expected loss can be readily computed out-of-sample as a  function of the out-of-sample loss of $\hat t$, variance of $\hat{\tau}_{ih}$ and shrinkage factor $w_{ih}$. Finally, we include an empirical-Bayes ridge benchmark with covariates that includes the site dummies, individual level characteristics and their interactions. We estimate the ridge penalty on the coefficient fold by fold by empirical Bayes under a Gaussian prior and likelihood.\footnote{Specifically, we assume an heteroskedastic Gaussian sampling model
$ 
y_i \mid \alpha,\beta \sim N(\alpha+z_i^\top\beta,\,\hat{\eta}_i^2),
\beta \sim N(0,\tau^2 I),
$ 
with the intercept left unpenalized. Writing the ridge penalty as $1/\tau^{2}$, we choose it by minimizing the negative log marginal likelihood. 
We then fit a separate ridge predictor for each outcome $h$, predict out of sample outcome by outcome, and average the resulting variance-corrected losses across outcomes.} Figure~\ref{fig:comparison_grf} reports the benchmark loss minus the loss of the ignorance-aware tree on the prediction set, separately by baseline poverty-score quartile. Positive values therefore indicate that the ignorance-aware tree has lower held-out loss than the benchmark. The estimated loss differences are positive and significant across all quartiles (with only one exception of higest poverty score for ridge, where effects are positive but noisy). Intuitively, GRF, CART, and empirical Bayes are useful methods for estimating heterogeneous effects, but they are designed to force predictions for all types and use a different objective criterion than the one we introduce for the archetype discovery problem with abstention.

Figure \ref{fig:comparison_k_means} report the comparison with $k$-means using ten and twenty clusters estimated via dynamic programming. The figure illustrate large MSE of clustering with many clusters due to potential overfitting of a large number of clusters, with MSE about 35 times larger than our method. In addition, standard clustering without abstention may misclassify units in the basin of ignorance, as reported in the right-side panel; the panel counts the share of pairs assigned to the same cluster where only one of the two element is not in the basin oignorance. It shows that several units are assigned to archetypes even if they should be classified as part of the basin of ignorance.

\begin{figure}[!ht]
\centering
\includegraphics[scale = 0.35]{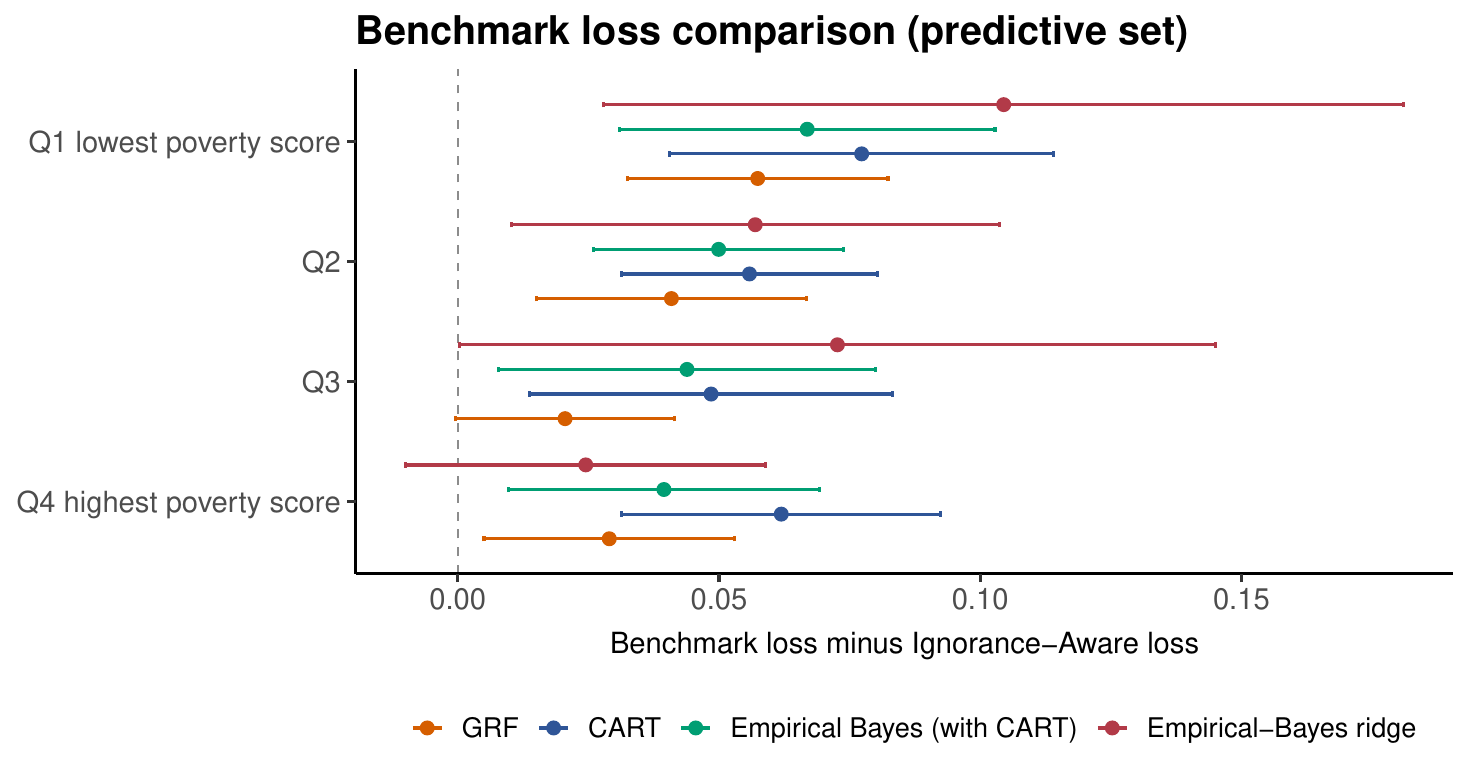}
\caption{\textbf{Out-of-sample comparisons for the depth-4 ignorance-aware tree at \(m=3\), shown separately by baseline poverty-score quartile.} Each point reports the mean held-out loss difference on the prediction set between a benchmark method and the ignorance-aware tree. The benchmarks are \(\text{GRF},\text{CART},\text{Empirical Bayes with CART}, \text{Empirical Bayes Ridge}\). Positive values indicate that the ignorance-aware tree has lower predictive-set loss than the benchmark. Horizontal bars are \(90\%\) confidence intervals.}
\label{fig:comparison_grf}
\end{figure}

 \subsection{Additional figures and tables}

\begin{figure}[!ht]
\centering 
\includegraphics[scale = 0.5]{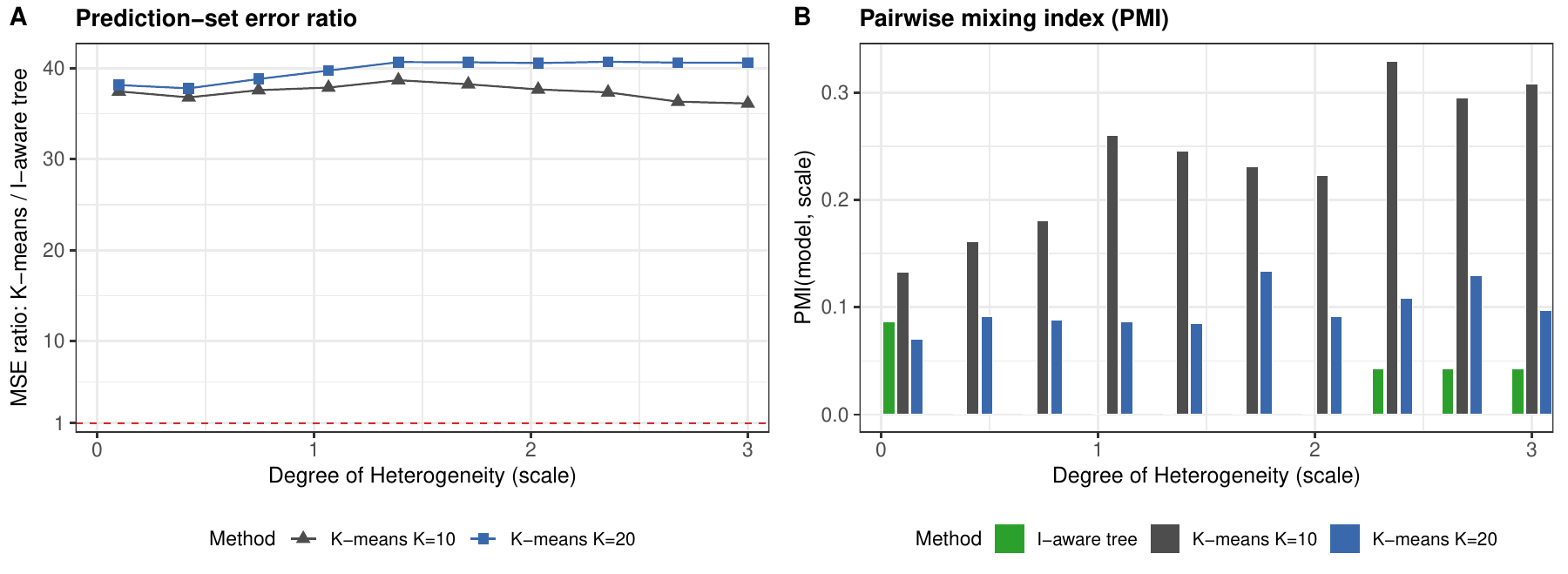}
\caption{\textbf{Comparison to k-means clustering without abstention with ten and twenty clusters}. Comparison with $k$-means clustering on the estimated CATE estimated exactly via dynamic programming with ten or twenty clusters in calibrated simulations. The pairwise mixing index counts  the number of pairs across clusters which are assigned to the same clusters of which only one of the units in the pair is (under the DGP) in the true basin of ignorance. It is computed formally as $\sum_c (\#\{\text{ignorance unit in } c\}\#\{\text{non-ignorance unit in } c\}/I_{tot}N_{tot}$ where $I_{tot}$ is the overall number of units in the basin of ignorance and $N_{tot}$ is the overall number of units outside the basin.} \label{fig:comparison_k_means}
\end{figure} 

\begin{figure}[!ht]
\centering
\includegraphics[scale = 0.4]{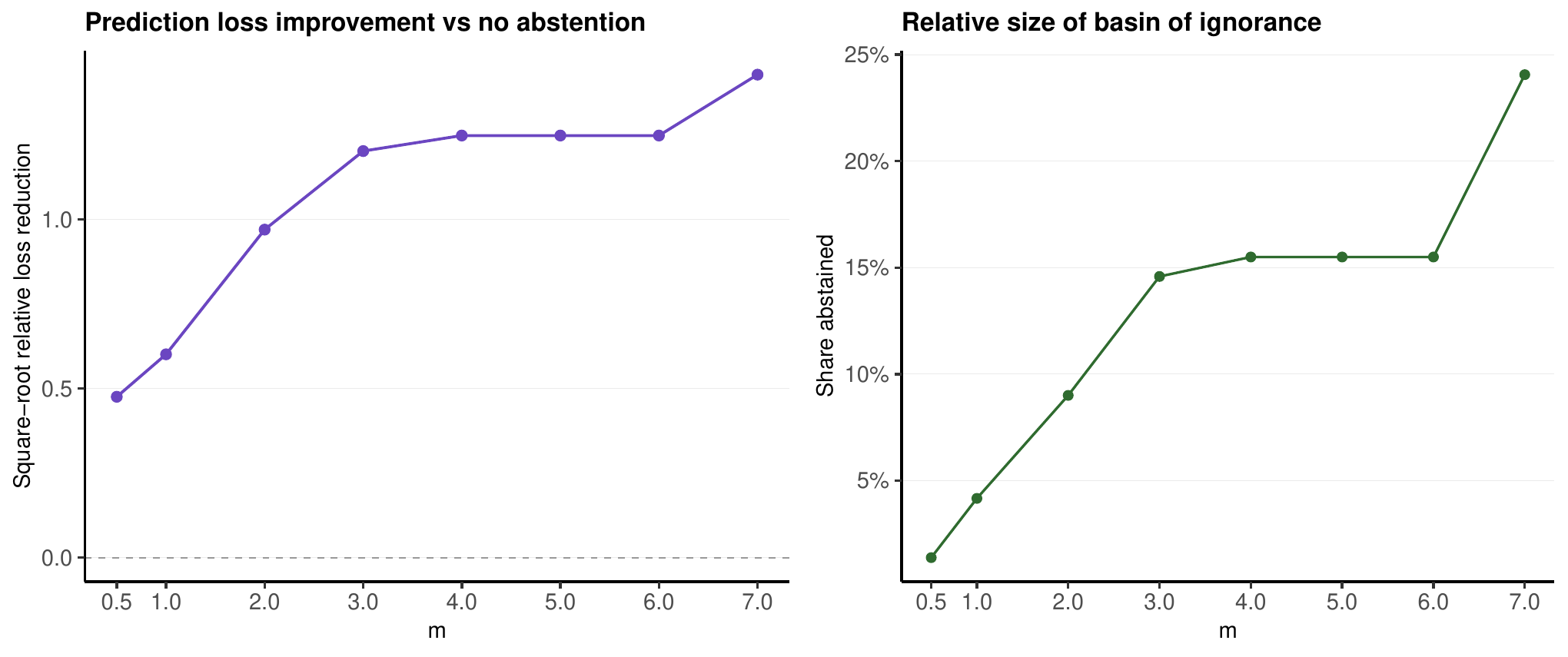}
\caption{\textbf{In-sample prediction-loss improvement relative to no abstention and size of the basin of ignorance.}
The left panel reports the prediction-loss improvement of the depth-four abstaining tree relative to the global no-abstention tree, separately for each value of the abstention-cost parameter \(m\). The plotted statistic is  
$ 
\sqrt{\frac{\widehat L_{\mathrm{no\ abstain,full}}-\widehat L_m}
{\widehat L_{\mathrm{no\ abstain,full}}}},
$ 
where \(\widehat L_m\) denotes the prediction loss of the abstaining tree at value \(m\), and \(\widehat L_{\mathrm{no\ abstain,full}}\) denotes the prediction loss of the no-abstention benchmark. Positive values therefore indicate lower prediction loss for the abstaining tree. The right panel reports the relative size of the basin of ignorance at each \(m\), weighting each estimated type by the number of matched units used to construct the corresponding type-level estimate \(\hat\tau(x)\).}
\label{fig:estimated_tree}
\end{figure}


\begin{figure}[!ht]
\centering 
\includegraphics[scale = 0.4]{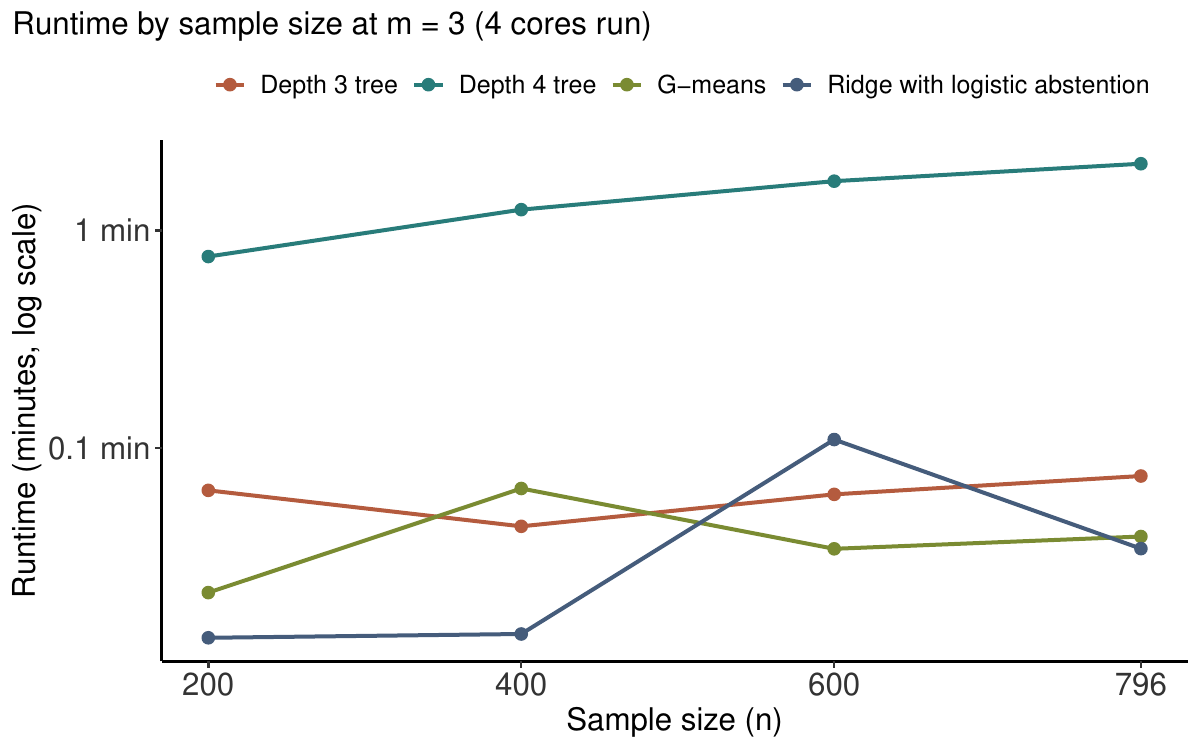}
\caption{Runtime results for $m=3$ were generated on a Mac Studio (Mac14,14) running macOS 26.3.1, with an Apple M2 Ultra processor, 24 CPU cores (16 performance and 8 efficiency), 192 GB unified memory, and arm64 architecture. The method only uses 4 of the 24 CPU cores available. 
Depth 3 and Depth 4 denote the exact tree search computed under Algorithm \ref{alg:alg3}, with depths 3 and 4 corresponding to 8 and 16 terminal groups, respectively. G-means clustering denotes the approximate G-means procedure with 16 groups and a maximum of 200 gradient iterations as in Algorithm \ref{alg:gmeans_abstention}. Ridge with logistic abstention denotes the differentiable abstention-ridge estimator with site-by-covariate interaction features, and a maximum of 200 gradient iterations.}
\label{fig:runtime} 
\end{figure}

\begin{figure}[!ht]
\centering
\scalebox{0.7}{\begin{tikzpicture}[
  node distance=7mm and 10mm,
  >=Latex,
  every node/.style={font=\small},
  block/.style={draw, rounded corners, align=center, minimum width=40mm, minimum height=7mm},
  decision/.style={draw, diamond, aspect=2, align=center, inner sep=1pt},
  arrow/.style={->, thick}
]

\node[block] (input) {Node \(S \subseteq \mathcal{X}\), depth \(\ell\)};
\node[decision, below=of input] (stop) {\(\ell=L\)?};

\node[block, below left=12mm and 22mm of stop] (finalsplit) {Final layer: enumerate candidate splits};
\node[below=1mm of finalsplit, align=center, font=\scriptsize] (finalsplitc) {At most \(O(d|S|)\) candidates};
\node[block, below=of finalsplitc] (finaleval) {For each: evaluate left/right leaves,\\ assign archetype or ignorance, sum losses};
\node[below=1mm of finaleval, align=center, font=\scriptsize] (finalevalc) {Local cost: \(O(d|S|^2)\)};
\node[block, below=of finalevalc] (finalret) {Return minimizing split\\ \& leaf assignments};

\node[block, below right=12mm and 22mm of stop] (split) {Enumerate candidate splits};
\node[below=1mm of split, align=center, font=\scriptsize] (splitc) {At most \(O(d|S|)\) candidates};
\node[block, below left=10mm and 12mm of splitc] (leftcall) {Recursive call on\\ left child \(S_L\)\\[1mm]
{\scriptsize subtree cost: \(O((d|S_L|^2)^{L-\ell})\)}};
\node[block, below right=10mm and 12mm of splitc] (rightcall) {Recursive call on\\ right child \(S_R\)\\[1mm]
{\scriptsize subtree cost: \(O((d|S_R|^2)^{L-\ell})\)}};
\node[block, below=16mm of splitc] (combine) {For each:\\ total loss = left loss + right loss};
\node[below=1mm of combine, align=center, font=\scriptsize] (combinec) {Local cost: \(O(d|S|^2)\)};
\node[block, below=of combinec] (best) {Return minimizing split\\ \& descendant assignments};

\draw[arrow] (input) -- (stop);

\draw[arrow] (stop) -- node[above left] {yes} (finalsplit);
\draw[arrow] (finalsplitc) -- (finaleval);
\draw[arrow] (finalevalc) -- (finalret);

\draw[arrow] (stop) -- node[above right] {no} (split);
\draw[arrow] (splitc) -- (leftcall);
\draw[arrow] (splitc) -- (rightcall);
\draw[arrow] (combinec) -- (best);

\end{tikzpicture}
}
\caption{Recursive structure of the exact tree-based optimization. At a non-terminal node, the algorithm enumerates candidate splits, recursively solves the left and right children for each split, sums the two continuation losses, and returns the split with minimum total loss together with the corresponding descendant assignments. At the final depth \(L\), the algorithm still enumerates candidate splits, but then evaluates the two resulting leaves directly, assigning each leaf to either an archetype or the basin of ignorance. If candidate thresholds are restricted to observed split points, there are at most \(O(d|S|)\) candidate splits at node \(S\), and evaluating them all costs \(O(d|S|^2)\) locally.}
\label{fig:tree_recursion}
\end{figure}

\begin{table}[!ht]
\centering
\footnotesize 
\begin{threeparttable}
\caption{\textbf{Calibrated simulation design.}}
\label{tab:calibrated_simulation_design}
\begin{tabular}{p{0.28\textwidth}p{0.62\textwidth}}
\toprule
\textbf{Component} & \textbf{Specification} \\
\midrule

Outcome &
Index outcome, defined as the average of standardized consumption, food-security, and asset outcomes. \\

\specialrule{0.2pt}{3pt}{3pt}

Calibration tree &
Depth-two ignorance-aware tree estimated on individual-level IPW pseudo-outcomes, with \(G=4\), minimum leaf size \(20\), five candidate split points per variable, and constant cost \(c_{\mathrm{cal}}=2.5\). Active covariates are Peru, baseline log consumption, and baseline asset index. \\

\specialrule{0.2pt}{3pt}{3pt}

Stable effects &
In predictive leaves, \(\tau(x)\) is set equal to the calibration-tree prediction. The larger non-Peru ignorance region is also pooled with the stable component. \\

\specialrule{0.2pt}{3pt}{3pt}

Contaminated region &
The small Peru ignorance region, approximately \(4\%\) of the sample. In this region, \(\tau(x)\sim \mathrm{Cauchy}(0,s)\). \\

\specialrule{0.2pt}{3pt}{3pt}

Heterogeneity parameter &
Cauchy scale \(s\in[0.1,3]\), which controls the amount of arbitrary heterogeneity. \\

\specialrule{0.2pt}{3pt}{3pt}

Simulated data &
Conditional on \(\tau(x)\), outcomes are homoskedastic Gaussian with variance one and treatment is assigned by a Bernoulli design. The simulated IPW signal is unbiased for \(\tau(x)\). \\

\specialrule{0.2pt}{3pt}{3pt}

Estimators &
In each replication, re-estimate the ignorance-aware tree using simulated \(\hat\tau(x)\) and \(\hat\eta(x)^2\), for \(c\in\{0.1,0.5,1,1.5,2\}\). Benchmarks are a no-abstention tree, CART, generalized random forests, and empirical Bayes with CART shrinkage. \\

\specialrule{0.2pt}{3pt}{3pt}

Evaluation &
Prediction error is computed on the stable prediction set. We also report the share of contaminated observations correctly assigned to ignorance and the total share assigned to ignorance. \\

\specialrule{0.2pt}{3pt}{3pt}

Replications &
One hundred Monte Carlo replications for each value of \(s\). \\

\bottomrule
\end{tabular}
\end{threeparttable}
\end{table}

\section{Auxiliary lemmas} \label{sec:auxiliary_lemmas}


 
\begin{lem} \label{lem:kita} For any $i \in \{1, \cdots, n\}$, let $X_i \in \mathcal{X}$ be an arbitrary random variable and $\mathcal{F}$ a class of uniformly bounded functions with envelope $\bar{F}$. 
Let $\Omega_i | X_1, \cdots, X_n$ be random variables independently but not necessarily identically distributed, where $\Omega_i$ is a scalar. Let for some arbitrary $u > 0, u' \in (0,1]$, 
$
\max\{\mathbb{E}[|\Omega_i|^{2 -2u'}|X],\mathbb{E}[|\Omega_i|^{2 + u} | X]\} \le B_{u,u'}, \quad \forall i \in \{1, \cdots, n\}.
$ In addition, assume that for any distribution $Q_n$ on points $x_1^n \in \mathcal{X}^n$, for some $V_n \ge 0$, for all $n \ge 1$,  
$\int_0^{2\bar{F}} \sqrt{\log\Big(\mathcal{N}\Big(\eta, \mathcal{F}, L_1(Q_n)\Big)\Big)} d\eta \le \sqrt{V_n}. 
$ Let $\sigma_i$ be $i.i.d$ Rademacher random variables independent of $(\Omega_i)_{i=1}^n,(X_i)_{i=1}^n$. 
Then for a constant $0 < C_{\bar{F}}< \infty$ that depends only on $\bar{F}$ and $u$, for all $n \ge 1$, and for $\Omega_i \ge 0$
\begin{equation} \label{eqn:part1}
\small
\int_0^{\infty} \mathbb{E}\Big[\sup_{f \in \mathcal{F}}\Big|\frac{1}{n} \sum_{i=1}^n  \sigma_i f(X_i)1\{\Omega_i > \omega\}\Big| | X_1, \cdots, X_n \Big]d\omega \le  \frac{C_{\bar{F}}}{u'}  \sqrt{\frac{B_{u, u'} V_n}{n}}. 
\end{equation} 
In addition, for $\Omega_i \in \mathbb{R}$
\begin{equation} \label{eqn:part2}
\small
\mathbb{E}\Big[\sup_{f \in \mathcal{F}}\Big|\frac{1}{n} \sum_{i=1}^n  \sigma_i f(X_i)\Omega_i \Big| | X_1, \cdots, X_n \Big] \le  \frac{C_{\bar{F}}}{u'}  \sqrt{\frac{B_{u, u'} V_n}{n}}. 
\end{equation}
\end{lem}

\begin{proof}[Proof of Lemma \ref{lem:kita}] We denote $X = (X_1, \cdots, X_n)$
For Equation \eqref{eqn:part1}, versions of this lemma can be found in Lemma A.5 in \cite{kitagawa2021equality} and \cite{viviano2024policy} (Lemma D.4), whose complete proof is available on the additional supplementary material available online at \url{https://dviviano.github.io/projects/note\_preliminary\_lemmas.pdf} (Appendix E, proof of Lemma E.9). We introduce a small modification to the above two references. 
Instead of defining $B$ to be some upper bound on the second plus $u$ moment of $\Omega_i$ (e.g., greater than one), we define it using an exact equality, taking into account also the moment $\mathbb{E}[\Omega_i^{2 - 2u'}|X]$ and then divide by $u'$.  For example, for $u = 1, u' = 1$, then $B$ defines the maximum between the third moment of $\Omega_i | X$ and one.
Following verbatim the proof of Lemma E.9 in \url{https://dviviano.github.io/projects/note\_preliminary\_lemmas.pdf}, we can write from the paragraph ``Integral Bound"
\begin{equation} 
\footnotesize 
\begin{aligned} 
\int_0^{\infty} \mathbb{E}\Big[\sup_{f \in \mathcal{F}}\Big|\frac{1}{n} \sum_{i=1}^n  \sigma_i f(X_i)1\{|\Omega_i| > \omega\}\Big| | X_1, \cdots, X_n \Big]d\omega & \le  \underbrace{\int_0^1 C_{\bar{F}} \sqrt{\frac{V_n}{n}} \sqrt{\frac{\sum_{i=1}^n P(|\Omega_i| > \omega|X)}{n}} d\omega}_{(I)} \\ &+ \underbrace{\int_1^{\infty} C_{\bar{F}} \sqrt{\frac{V_n}{n}} \sqrt{\frac{\sum_{i=1}^n P(|\Omega_i| > \omega|X)}{n}} d\omega}_{(II)}.
\end{aligned} 
\end{equation}
Here we bound $(II)$ as in \cite{viviano2024policy}, and therefore write $(II) \le C_{\bar{F}'}\sqrt{V_n \max_i \mathbb{E}[|\Omega_i|^{2 + u}|X]/n} \le C_{\bar{F}'} \sqrt{V_n B_{u, u'}/n}$. For $(I)$, instead of bounding $P(|\Omega_i| > \omega|X) \le 1$ as in \cite{viviano2024policy}, we use $P(|\Omega_i| > \omega|X) \le \mathbb{E}[|\Omega_i|^{2 - 2u'} | X]/(\omega^{2 - 2u'})$, which, after integrating out, give us 
$(I) \le \frac{1}{u'} C_{\bar{F}'} \sqrt{V_n B/n}$. To prove the second claim 
$$
\footnotesize 
\begin{aligned} 
\mathbb{E}\Big[\sup_{f \in \mathcal{F}}\Big|\frac{1}{n} \sum_{i=1}^n  \sigma_i f(X_i)\Omega_i \Big| | X_1, \cdots, X_n \Big] & = \mathbb{E}\Big[\sup_{f \in \mathcal{F}}\Big|\frac{1}{n} \sum_{i=1}^n  \sigma_i f(X_i)|\Omega_i| \mathrm{sign}(\Omega_i)  \Big| | X_1, \cdots, X_n \Big] \\ 
& =  \mathbb{E}\Big[\sup_{f \in \mathcal{F}}\Big|\frac{1}{n} \sum_{i=1}^n  \tilde{\sigma}_i f(X_i)|\Omega_i|   \Big| | X_1, \cdots, X_n \Big]
\end{aligned} 
$$ 
where $\tilde{\sigma}_i = \mathrm{sign}(\Omega_i) \sigma_i$. We can then write 
$$
\footnotesize 
\begin{aligned} 
\mathbb{E}\Big[\sup_{f \in \mathcal{F}}\Big|\frac{1}{n} \sum_{i=1}^n  \tilde{\sigma}_i f(X_i)|\Omega_i|   \Big| | X_1, \cdots, X_n \Big] & = \mathbb{E}\Big[\sup_{f \in \mathcal{F}}\Big|\frac{1}{n} \sum_{i=1}^n  \tilde{\sigma}_i f(X_i)\int 1\{|\Omega_i| \ge \omega\} d\omega   \Big| | X_1, \cdots, X_n \Big] \\ 
&\le \mathbb{E}\Big[\sup_{f \in \mathcal{F}}\int \Big|\frac{1}{n} \sum_{i=1}^n  \tilde{\sigma}_i f(X_i) 1\{|\Omega_i| \ge \omega\}    \Big| d\omega | X_1, \cdots, X_n \Big] \\
&\le \int \mathbb{E}\Big[\sup_{f \in \mathcal{F}} \Big|\frac{1}{n} \sum_{i=1}^n  \tilde{\sigma}_i f(X_i) 1\{|\Omega_i| \ge \omega\}   \Big| | X_1, \cdots, X_n \Big] d\omega. 
\end{aligned} 
$$
Finally, note that  $\mathbb{P}(\tilde{\sigma}_i = 1 | \Omega, X) =\mathbb{P}(\sigma_i \mathrm{sign}(\Omega_i) = 1 | \Omega, X) = 1/2$ which implies that $\tilde{\sigma}_i$ are Rademacher random variables independent of $\Omega_i, X$. We can then invoke Equation \eqref{eqn:part1} to complete the proof. 
\end{proof}

\begin{lem}(\cite{vershynin2018high}, Lemma 6.4.2) \label{lem:symmetrization}
Let $\sigma_1, ..., \sigma_n$ be Rademacher sequence independent of $X_1, ..., X_n$. Suppose that $X_1, \cdots, X_n$ are independent. Then 
$$
\small
\mathbb{E}\Big[\sup_{f \in \mathcal{F}} \Big| \sum_{i=1}^n f(X_i) - \mathbb{E}[f(X_i)]\Big|\Big] \le 2\mathbb{E}\Big[\sup_{f \in \mathcal{F}}\Big| \sum_{i=1}^n \sigma_i f(X_i)\Big|\Big].
$$	
\end{lem}

\begin{lem}[\cite{devroye2013probabilistic}, Theorem 29.7] \label{lem:cover_product}
Let \(\mathcal F\) and \(\mathcal G\) be classes of measurable functions on \(\mathcal X\), and suppose that there exist finite constants \(B_{\mathcal F},B_{\mathcal G}\) such that
$
\sup_{f\in\mathcal F,\ x\in\mathcal X}|f(x)|\le B_{\mathcal F}, 
\sup_{g\in\mathcal G,\ x\in\mathcal X}|g(x)|\le B_{\mathcal G}.
$
Define the product class
$ 
\mathcal F\cdot \mathcal G
:=
\{x\mapsto f(x)g(x): f\in\mathcal F,\ g\in\mathcal G\}, 
$ 
and for $n$ points $x_1, \cdots, x_n$, with distribution $Q_n$, the corresponding $\mathcal{N}(\varepsilon, (F\cdot \mathcal G), L_1(Q_n))$ covering number. 
Then, for every \(\varepsilon>0\), $Q_n, n\ge 1$
\[
\mathcal{N}\!\left(\varepsilon,(\mathcal F\cdot \mathcal G), L_1(Q_n))\right)
\le
\mathcal{N}\!\left(\frac{\varepsilon}{2B_{\mathcal G}},\mathcal F, L_1(Q_n)\right)
\,
\mathcal{N}\!\left(\frac{\varepsilon}{2B_{\mathcal F}},\mathcal G, L_1(Q_n)\right).
\]
\end{lem}

\begin{lem}[\cite{devroye2013probabilistic}, Theorem 29.6] \label{lem:devroye_sum} Let $\mathcal{F}_1,\cdots,\mathcal{F}_k$,$k\ge 1$ be classes of real
functions on $x$. Define $\mathcal{F}= \{f_1 + \cdots +
f_k; f_j \in \mathcal{F}_j, j = 1,···,k\}$. Then for every $\eta  > 0$, $x_1^n$ with distribution $Q_n$, any $Q_n$, $n  \ge 1, k \ge 1$, $\mathcal{N}(\eta, \mathcal{F}, L_1(Q_n)) \le \prod_{j=1}^k \mathcal{N}(\eta/k, \mathcal{F}_j, L_1(Q_n))$
\end{lem} 
\begin{lem}[\cite{viviano2024policy}, Lemma D.5]\label{lem:boundnumber} Let $\mathcal{F}_1, \cdots, \mathcal{F}_{k}
$ be classes of bounded functions with VC dimension $v$ and envelope $\bar{F} < \infty$. Let 
$$
\small
\begin{aligned} 
\mathcal{J} = \Big\{f_1(f_2 + ... + f_{k}), \quad f_j \in \mathcal{F}_j, \quad j = 1, \cdots, k\Big\}.
\end{aligned} 
$$
 For arbitrary points $x_1^n \in \mathcal{X}^n$ with distribution $Q_n$, any $Q_n$, for \textit{any} $n \ge 1, k \ge 2,  v\ge 1$, 
$
\int_0^{2\bar{F}} \sqrt{\log\Big(\mathcal{N}\Big(\eta, \mathcal{J}, L_1(Q_n)\Big)\Big)} d\eta < c_{\bar{F}} \sqrt{k \log(k+1) v}
$
for a constant $c_{\bar{F}} < \infty$ that only depends on $\bar{F}$. 
\end{lem} 

\begin{lem}[\cite{devroye2013probabilistic}, Theorem 13.5(ii)] \label{lem:vc_sum} Let $g : \mathbb{R}^d \mapsto \mathbb{R}$ be an arbitrary function and
consider the class of functions $\mathcal{G}= \{g + f,f  \in \mathcal{F}\}$. Then the VC-dimension of $\mathcal{G}$ equals the VC-dimension of $\mathcal{F}$. 
\end{lem}

\begin{lem}[Moments of the centered square] \label{lem:centered_square}
Fix \(u\in[0,1/2]\).  Let Assumption \ref{ass:general_regret} hold. 
Then there exists a universal constant \(c<\infty\), such that for each $x \in \mathcal{X}$
\[
\max\Big\{
\mathbb E\big[|\hat\tau(x)^2-\mathbb E[\hat\tau(x)^2]|^{2-2u}\big],
\ \mathbb E\big[|\hat\tau(x)^2-\mathbb E[\hat\tau(x)^2]|^{3}\big]
\Big\}
\le
c B^{3}\,(M_u+M_u^2).
\]
\end{lem}

\begin{proof}
Write \(\varepsilon=\hat\tau(x)-\tau(x)\) omitting $x$ for brevity. Since \(\mathbb E[\varepsilon]=0\) and \(\eta(x)^2=\mathbb E[\varepsilon^2]\),
$
\hat\tau(x)^2-\mathbb E[\hat\tau(x)^2]
=
2\tau(x)\varepsilon+\bigl(\varepsilon^2-\eta(x)^2\bigr).
$ 
Using \(|a-b|^q\le 2^{q-1}(|a|^q+|b|^q)\), for any \(q\in\{2-2u,3\}\), and $|\tau(x)| \le B$, 
\[
\mathbb E\big[|\hat\tau(x)^2-\mathbb E[\hat\tau(x)^2]|^q\big]
\le
2^{q-1}(2B)^q \mathbb E[|\varepsilon|^q]
+
2^{q-1}\mathbb E[|\varepsilon^2-\eta(x)^2|^q], 
\]
and 
$ 
\mathbb E[|\varepsilon^2-\eta(x)^2|^q]
\le
2^{q-1}\Big(\mathbb E[|\varepsilon|^{2q}] + \eta(x)^{2q}\Big).
$ 
Since \(q\ge 1\), Jensen's inequality gives
$
\eta(x)^{2q}=\big(\mathbb E[\varepsilon^2]\big)^q\le \mathbb E[|\varepsilon|^{2q}],
$ 
so that 
$ 
\mathbb E[|\varepsilon^2-\eta(x)^2|^q]\le 2^q\,\mathbb E[|\varepsilon|^{2q}].
$ 
Also, by Hölder's inequality,
$ 
\mathbb E[|\varepsilon|^q]\le \big(\mathbb E[|\varepsilon|^{2q}]\big)^{1/2}.
$ 
This implies that 
$
\mathbb E[|\varepsilon|^q]\le M_u,
\mathbb E[|\varepsilon|^{2q}]\le M_u^2.
$ 
Therefore,
\[
\mathbb E\big[|\hat\tau(x)^2-\mathbb E[\hat\tau(x)^2]|^q\big]
\le
2^{q-1}(2B)^q M_u + 2^{2q-1}M_u^2.
\]
Since \(q\le 3\), both coefficients are bounded by \(32(1+B^3)\). Using the fact that $B \ge 1$, the result follows.   
\end{proof}

\begin{lem}[Covering number of a union] \label{lem:cover_union}
Let \(\mathcal F_1,\dots,\mathcal F_m\) be classes of measurable functions on \(\mathcal X\), and let
$ 
\mathcal F:=\bigcup_{j=1}^m \mathcal F_j.
$ 
Then, for any probability measure \(Q\) on \(\mathcal X\) and any \(\varepsilon>0\),
$ 
\mathcal{N}(\varepsilon,\mathcal F,L_1(Q))
\le
\sum_{j=1}^m \mathcal{N}(\varepsilon,\mathcal F_j,L_1(Q)).
$ 
\end{lem}

\begin{proof}
For each \(j\in\{1,\dots,m\}\), let \(\mathcal C_j\) be an \(\varepsilon\)-cover of \(\mathcal F_j\) under \(L_1(Q)\) with cardinality \(\mathcal{N}(\varepsilon,\mathcal F_j,L_1(Q))\). Then
$ 
\mathcal C:=\bigcup_{j=1}^m \mathcal C_j
$ 
is an \(\varepsilon\)-cover of \(\mathcal F=\bigcup_{j=1}^m \mathcal F_j\). Hence
$ 
\mathcal{N}(\varepsilon,\mathcal F,L_1(Q))
\le
|\mathcal C|
\le
\sum_{j=1}^m \mathcal{N}(\varepsilon,\mathcal F_j,L_1(Q)),
$ 
which proves the claim.
\end{proof}

\begin{lem}[Assouad's lemma for the Hamming loss, Theorem 2.12 in \cite{tsybakov2009introduction}] \label{lem:assouad}
Let \(\{P_\theta:\theta\in\{-1,+1\}^m\}\) be a family of probability measures, and for \(j\in\{1,\dots,m\}\) let \(\theta^{(j)}\) denote the vector obtained from \(\theta\) by changing only its \(j\)-th coordinate. Suppose that
$ 
\mathrm{KL}\!\left(P_\theta,P_{\theta^{(j)}}\right)\le \kappa
\text{ for all }\theta\in\{-1,+1\}^m,\ j\in\{1,\dots,m\}, 
$ 
where KL denotes the KL-divergence. 
Then, for every estimator \(\widehat\theta=(\widehat\theta_1,\dots,\widehat\theta_m)\),
\[
\sup_{\theta\in\{-1,+1\}^m}
\mathbb E_\theta\!\left[
\sum_{j=1}^m 1\{\widehat\theta_j\neq \theta_j\}
\right]
\ge
\frac{m}{2}\Bigl(1-\sqrt{\kappa/2}\Bigr).
\]
\end{lem}

\begin{proof}
The result follows directly from \cite{tsybakov2009introduction}, Theorem~2.12.
\end{proof}

\begin{lem}[Entropy bound for bounded linear classes]
\label{lem:linear_entropy}
Let \(q:\mathcal X\to\mathbb R^d\) satisfy
$
\sup_{x\in\mathcal X}\|q(x)\|_2\le U
$ 
for some \(U \in (0,\infty)\), and consider the function class
$
\mathcal F_{\mathrm{lin}}
:=
\bigl\{f_\theta(x)=q(x)^\top\theta:\ \|\theta\|_2\le R\bigr\}
$ 
for some \(R \in (0,\infty)\). Let \(\mathcal Q\) denote the collection of finitely supported probability measures on \(\mathcal X\). Then \(\sup_{f\in\mathcal F_{\mathrm{lin}},\,x\in\mathcal X}|f(x)|\le RU\), and there exists a universal constant \(c_0<\infty\) such that
$ 
\sup_{Q\in\mathcal Q}
\int_0^{2RU}
\sqrt{\log \mathcal N\!\left(\varepsilon,\mathcal F_{\mathrm{lin}},L_1(Q)\right)}
\,d\varepsilon
\le
c_0\,RU\,\sqrt d.
$ 
\end{lem}

\begin{proof} The proof is standard \citep{devroye2013probabilistic}. 
For any \(Q\in\mathcal Q\) and any \(\theta,\theta'\in\mathbb R^d\),
\[
\|f_\theta-f_{\theta'}\|_{Q,1}
=
\int_{\mathcal X}\bigl|q(x)^\top(\theta-\theta')\bigr|\,dQ(x)
\le
\int_{\mathcal X}\|q(x)\|_2\,\|\theta-\theta'\|_2\,dQ(x)
\le
U\|\theta-\theta'\|_2.
\]
Hence an \((\varepsilon/U)\)-cover of the Euclidean ball
$ 
B_2^d(R):=\{\theta\in\mathbb R^d:\|\theta\|_2\le R\}
$
induces an \(\varepsilon\)-cover of \(\mathcal F_{\mathrm{lin}}\) under \(L_1(Q)\), so
$
\mathcal N\!\left(\varepsilon,\mathcal F_{\mathrm{lin}},L_1(Q)\right)
\le
\mathcal N\!\left(\varepsilon/U,B_2^d(R),\|\cdot\|_2\right).
$ 
By the standard covering bound for the Euclidean ball,
$ 
\mathcal N\!\left(\delta,B_2^d(R),\|\cdot\|_2\right)
\le
\left(1+\frac{2R}{\delta}\right)^d,
$ 
and therefore
$ 
\log \mathcal N\!\left(\varepsilon,\mathcal F_{\mathrm{lin}},L_1(Q)\right)
\le
d\log\!\left(1+\frac{2RU}{\varepsilon}\right).
$ 
It follows that
$ 
\int_0^{2RU}
\sqrt{\log \mathcal N\!\left(\varepsilon,\mathcal F_{\mathrm{lin}},L_1(Q)\right)}
\,d\varepsilon
\le
\sqrt d
\int_0^{2RU}
\sqrt{\log\!\left(1+\frac{2RU}{\varepsilon}\right)}
\,d\varepsilon.
$ 
With the change of variables \(t=\varepsilon/(2RU)\), the right-hand side becomes
$ 
2RU\sqrt d
\int_0^1 \sqrt{\log(1+t^{-1})}\,dt,
$ 
and the last integral is a finite universal constant. This proves the claim.
\end{proof}

\subsection{Proofs of auxiliary lemmas for Theorem \ref{thm:gmeans_lower_union_ignorance}} \label{sec:auxiliary_lower_bound}

Throughout, we will use the notation introduced in the proof of Theorem \ref{thm:gmeans_lower_union_ignorance}, Appendix \ref{proof:lower_bound}.

\subsubsection{Proof of Lemma \ref{lem:bound_basic_constants}} \label{proof:lem:bound_basic_constant}
Since \(N\ge 128(G-1)^2\) and \(G\ge4\), \(N/[128(G-1)]\ge G-1\ge3\).
Thus \(b \ge \max\{1, N/[256(G-1)]\}\). Since 
$ 
Mb\le \frac{N}{256},
$ 
$ 
a\ge \frac{N-Mb}{G-1}-1
\ge
\frac{255}{256}\frac{N}{G-1}-1.
$ 
Therefore it follows 
\(a\ge N/[2(G-1)]\). Since \(\underline\kappa\le 1/[4(G-1)]\) by assumption in Theorem \ref{thm:gmeans_lower_union_ignorance},  
\(a\ge \underline\kappa N\). Moreover,
$ 
a\ge \frac{N}{2(G-1)} \ge 
4(G-1).
$ 
Also, 
$ 
a\ge \frac{N}{2(G-1)}
=
64\frac{N}{128(G-1)}
\ge 32b.
$ 
Finally, since \(b\ge1\) and \(\underline\eta\le K_0/32\),
$ 
\Delta=\frac{\underline\eta}{8\sqrt b}\le \frac{\underline\eta}{8}
\le \frac{K_0}{8}.
$ 

Next, we prove the first inequality in \eqref{eq:proof-MbDelta-vs-aK0}. This inequality is equivalent to
$ 
16Mb\Delta\le aK_0.
$ 
Since \(\Delta=\underline\eta/(8\sqrt b)\), it is enough to prove
$ 
2M\underline\eta\sqrt b\le aK_0.
$ 
Using \(2M\le G-1\), \(b\le N/[128(G-1)]\), and
\(a\ge N/[2(G-1)]\), we have
$ 
2M\underline\eta\sqrt b
\le
\underline\eta\sqrt{\frac{N(G-1)}{128}}, 
aK_0\ge \frac{NK_0}{2(G-1)}.
$ 
Hence \(2M\underline\eta\sqrt b\le aK_0\) follows if
$ 
\underline\eta\sqrt{\frac{N(G-1)}{128}}
\le
\frac{NK_0}{2(G-1)}.
$ 
After squaring and simplifying, this condition is
$ 
N\ge
\frac{1}{32}\frac{\underline\eta^2}{K_0^2}(G-1)^3.
$ 
This is implied by the sample-size condition in Theorem \ref{thm:gmeans_lower_union_ignorance}. 

Next, we prove the second inequality in \eqref{eq:proof-MbDelta-vs-aK0}. Note that because $0 \le \Delta \le K_0/8$, $(K_0- \Delta)^2 \le K_0^2$; in addition, it follows immediately that $ab/(a+b) \le b$ for $a \ge 0$. 
The last condition on $L$ follows since $0 \le L = N - Mb - (G-1)\lfloor (N - Mb)/(G-1)\rfloor < G-1$.

\subsubsection{Proof of Lemma \ref{lem:phi-bound}} \label{proof:lem:phi-bound}
For fixed \(k\), \(u\mapsto\Phi_k(u)\) is concave on \([0,k]\), since
\(u\mapsto au/(a+u)\) is concave and the other terms are affine in \(u\). Thus
$ 
\Phi_k(u)\ge \min\{\Phi_k(0),\Phi_k(k)\}.
$ 
At \(u=0\),
$ 
\Phi_k(0)
=
kK_0^2
-
\left(
\frac{ab}{a+b}
-
\frac{a(b-k)}{a+b-k}
\right)(K_0-\Delta)^2,
$ 
and
$ 
\frac{ab}{a+b}-\frac{a(b-k)}{a+b-k}
=
\frac{a^2k}{(a+b)(a+b-k)}
\le k.
$ 
Hence \(\Phi_k(0)\ge k(2K_0\Delta-\Delta^2)\ge kK_0\Delta\), using
\(\Delta\le K_0/8\).
Next note that 
\(k \mapsto \Phi_k(k)\), is concave on \([0,b]\), and \(\Phi_{0}(0)=0\). Hence
$ 
\Phi_k(k)\ge \frac{k}{b}\Phi_{b}(b).
$ 
A direct calculation gives \(\Phi_b(b)=4abK_0\Delta/(a+b)\). Since
\(a\ge 32b\) by Lemma~\ref{lem:bound_basic_constants},
$ 
\frac{4ab}{a+b}\ge  3b.
$ 
Thus \(\Phi_k(k)\ge3kK_0\Delta\). Combining the endpoint bounds proves the claim.

\subsubsection{Proof of Lemma \ref{lem:single-anchor-canonicalization}} \label{proof:lem:single-anchor-canonicalization}
Fix \(h\in\{2,\dots,G\}\) (therefore $h \neq 1$ does not include the abstention group). We first show that
$ 
\{x\in\mathcal X_N:\widetilde\alpha(x)=h\}
$ 
is either empty or a consecutive interval.

If \(u^\star=1\), then by construction 
$ 
\{x:\widetilde\alpha(x)=h\}
=
\{x:\alpha(x)=h\}\setminus A.
$ 
The set \(\{x:\alpha(x)=h\}\) is a consecutive interval by assumption. Removing the
consecutive block \(A\) from it can produce a non-consecutive set only if
$ 
A\subseteq \{x:\alpha(x)=h\}.
$ 
But in that case every point in \(A\) has label \(h\), so
$ 
\mathcal U_\alpha(A)=\{h\}.
$ 
This contradicts \(u^\star=1\), since $h > 1$. Hence
\(\{x:\widetilde\alpha(x)=h\}\) is either empty or consecutive.

Now suppose \(u^\star=h^\star\in\{2,\dots,G\}\). If \(h=h^\star\), then
$ 
\{x:\widetilde\alpha(x)=h^\star\}
=
\{x:\alpha(x)=h^\star\}\cup A.
$ 
Since \(h^\star\in\mathcal U_\alpha(A)\), we have
$ 
\{x:\alpha(x)=h^\star\}\cap A\neq\varnothing.
$ 
Thus \(\{x:\alpha(x)=h^\star\}\) and \(A\) are two consecutive intervals with nonempty
intersection, so their union is consecutive.

If instead \(h\neq h^\star\), then
$ 
\{x:\widetilde\alpha(x)=h\}
=
\{x:\alpha(x)=h\}\setminus A.
$ 
As above, this set can fail to be consecutive only if
$ 
A\subseteq \{x:\alpha(x)=h\}.
$ 
But then every point in \(A\) has label \(h\), so
$ 
\mathcal U_\alpha(A)=\{h\},
$ 
which would force \(u^\star=h\), contradicting \(h\neq h^\star\). Therefore
\(\{x:\widetilde\alpha(x)=h\}\) is either empty or consecutive for every
\(h\in\{2,\dots,G\}\).

The sets \(\{x:\widetilde\alpha(x)=h\}\), \(h=2,\dots,G\), are pairwise disjoint by construction  of
\(\widetilde\alpha\). Moreover, the number of nonempty sets among them
cannot increase. If \(u^\star=1\), each set is obtained from the corresponding original
set by removing points in \(A\). If \(u^\star=h^\star\ge2\), then the label \(h^\star\)
was already used on \(A\), because
$ 
\{x:\alpha(x)=h^\star\}\cap A\neq\varnothing,
$ 
and all other labels only lose points in \(A\). Hence no new nonempty predictive label is
created. It remains to prove the loss comparison. For every \(x\notin A\),
 the contribution of \(x\) to \(R_\theta(\alpha,\mu)\) and to
\(R_\theta(\widetilde\alpha,\widetilde\mu)\) is the same because $\widetilde{\alpha}(x) = \alpha(x)$ and $\widetilde{\mu} = \mu$.

For every \(x\in A\), we have \(\tau_\theta(x)= (-1)^g K_0\). Also,
\(\alpha(x)\in\mathcal U_\alpha(A)\). Therefore the original contribution of \(x\) under $\alpha,\mu$ to the loss is
\[
\sum_{h=2}^G
(\mu_h- (-1)^g K_0)^2\,1\{\alpha(x)=h\}
+
K_0^2\,1\{\alpha(x)=1\}
=
\lambda_{\alpha(x)}.
\]
The new contribution of \(x\) under $\widetilde{\alpha}, \widetilde{\mu}$ is
$ 
\sum_{h=2}^G
(\widetilde\mu_h-(-1)^gK_0)^2\,1\{\widetilde\alpha(x)=h\}
+
K_0^2\,1\{\widetilde\alpha(x)=1\}
=
\lambda_{u^\star}.
$ 
By definition of \(u^\star\),
$ 
\lambda_{u^\star}
=
\min_{u\in\mathcal U_\alpha(A)}\lambda_u
\le
\lambda_{\alpha(x)}.
$ 
Thus the pointwise contribution weakly decreases for every \(x\in A\), and is unchanged
for every \(x\notin A\). Summing over \(x\in\mathcal X_N\) gives
$ 
R_\theta(\widetilde\alpha,\widetilde\mu)
\le
R_\theta(\alpha,\mu).
$ 
This proves the claim.

\subsubsection{Proof of Lemma \ref{lem:structural_property} } \label{proof:lem:structural_property}
\emph{Part 1: defect count and construction of selected anchors.}
The goal of this part is to quantify how \(\widetilde\alpha\) fails to satisfy
\(\mathcal B_{\mathrm{good}}\). We first need to introduce some notation. 

For each \(j\in\{1,\dots,J_{\widetilde\alpha}\}\), define
$ 
t_j:=\left|\{g\in\{2,\dots,G\}:A_g\subseteq D_j\}\right|.
$ 
Let
$ 
I
:=
\left|
\left\{
g\in\{2,\dots,G\}:
\widetilde\alpha(x)=1\text{ for all }x\in A_g
\right\}
\right|,
$ 
and define
$ 
J_0:=|\{j:t_j=0\}|,
J_{\mathrm{anch}}:=|\{j:t_j\ge 1\}|.
$ 
Since every anchor block is either entirely ignored or entirely contained in exactly one
predictive interval, we have
$ 
G-1
=
I+\sum_{j=1}^{J_{\widetilde\alpha}}t_j
=
I+J_{\mathrm{anch}}
+
\sum_{j=1}^{J_{\widetilde\alpha}}(t_j-1)_+.
$ 
Define
$ 
\delta
:=
I+\sum_{j=1}^{J_{\widetilde\alpha}}(t_j-1)_+.
$ 
Then
$ 
\delta=G-1-J_{\mathrm{anch}}.
$ 
Since \(J_{\widetilde\alpha}=J_{\mathrm{anch}}+J_0\le G-1\), it follows that
$ 
J_0\le \delta.
$ 
Moreover, by the definition of \(\mathcal B_{\mathrm{good}}\),
$ 
\widetilde\alpha\notin\mathcal B_{\mathrm{good}}
\Longleftrightarrow
\delta\ge 1.
$ 
Thus, under the maintained assumption \(\widetilde\alpha\notin\mathcal B_{\mathrm{good}}\),
we have
$ 
\delta\ge 1.
$ 

Define
$ 
\mathcal I_{\mathrm{ign}}
:=
\left\{
g\in\{2,\dots,G\}:
\widetilde\alpha(x)=1\text{ for all }x\in A_g
\right\}.
$ 
For each \(j\in\{1,\dots,J_{\widetilde\alpha}\}\) such that \(t_j\ge 2\), choose a subset
$ 
\mathcal S_j
\subseteq
\{g\in\{2,\dots,G\}:A_g\subseteq D_j\}
$ 
with
$ 
|\mathcal S_j|=t_j-1.
$ 
Such a choice is possible because
$ 
\left|\{g\in\{2,\dots,G\}:A_g\subseteq D_j\}\right|=t_j.
$ 
Now define
$ 
\mathcal S
:=
\{A_g:g\in\mathcal I_{\mathrm{ign}}\}
\cup
\bigcup_{\{j:t_j\ge2\}}\{A_g:g\in\mathcal S_j\}.
$ 
Thus \(\mathcal S\subseteq\{A_2,\dots,A_G\}\), every ignored anchor block is included in
\(\mathcal S\), and for every interval \(D_j\) with \(t_j\ge2\), exactly \(t_j-1\) of the
anchor blocks contained in \(D_j\) are included in \(\mathcal S\). Hence
$ 
|\mathcal S|
=
I+\sum_{j=1}^{J_{\widetilde\alpha}}(t_j-1)_+
=
\delta.
$ 

\medskip
\noindent
\emph{Part 2: lower bound on the excess loss of each selected anchor.}
The goal of this part is to prove the following inequality: 
\begin{equation}\label{eq:proof-selected-anchor}
\frac1N\sum_{x\in A_g}
\bigl(L_{\widetilde\alpha}(x)-L_{\mathrm{bench}}(x)\bigr)
\ge
\frac{aK_0^2}{8N},
\qquad A_g\in\mathcal S.
\end{equation}

Fix \(A_g\in\mathcal S\). There are two cases. First suppose that \(A_g\) is ignored by \(\widetilde\alpha\). Then
\(L_{\widetilde\alpha}(x)=K_0^2\) for every \(x\in A_g\). By
\eqref{eq:proof-benchmark-anchor},
$ 
L_{\mathrm{bench}}(x)\le \frac{K_0^2}{1024},
x\in A_g.
$ 
Since \(|A_g|\ge a\), we get
$ 
\frac1N\sum_{x\in A_g}
\bigl(L_{\widetilde\alpha}(x)-L_{\mathrm{bench}}(x)\bigr)
\ge
\frac{aK_0^2}{N}\left(1-\frac1{1024}\right)
>
\frac{aK_0^2}{8N}.
$ 
Now suppose that \(A_g\subseteq D_j\), where \(D_j\) contains at least two anchor blocks.
Write the anchor blocks contained in \(D_j\) as
$ 
A_h,A_{h+1},\dots,A_v,
h<v.
$ 
Since \(D_j\) is a consecutive interval and each anchor block is either entirely ignored or entirely
contained in one predictive interval, the anchor indices contained in \(D_j\) are
consecutive. Therefore
$ 
t_j=v-h+1\ge2.
$ 
By definition,
$ 
\nu_{j,\theta}^\star
=
\frac1{|D_j|}\sum_{z\in D_j}\tau_\theta(z).
$ 
On the anchor blocks contained in \(D_j\), because the signs of $\tau_\theta(x)$ alternate,
$ 
\left|
\sum_{z\in D_j\cap\bigcup_{r=h}^v A_r}\tau_\theta(z)
\right|
\le
(a+L)K_0
\le
(a+G-2)K_0,
$ 
where the last inequality uses \(L\le G-2\) from Lemma \ref{lem:bound_basic_constants}.

Next, define 
\[
\mathcal H(D_j):=\{m\in\{1,\dots,M\}:H_m\cap D_j\neq\varnothing\}.
\]
We claim that \(|\mathcal H(D_j)|\le t_j\). To show this note that if \(H_m\cap D_j\neq\varnothing\),
then, because \(D_j\) is a consecutive interval, \(H_m\) can only be immediately to the left of
\(A_h\), between two anchors contained in \(D_j\), or immediately to the right of \(A_v\).
Formally, one of the following three conditions holds:
$ 
2m+1=h,
h\le 2m<2m+1\le v,
2m=v.
$ 
Define
$ 
r:\mathcal H(D_j)\to\{h,\dots,v\}
$ 
by
$$ 
\small 
\begin{aligned} 
r(m) := h 1\{2 m + 1 = h\} + (2m + 1) 1\{h \le 2m < v\} + v 1\{2 m = v\}
\end{aligned} 
$$ 
This map is injective, and hence
$ 
|\mathcal H(D_j)|\le |\{h,\dots,v\}|=v-h+1=t_j.
$ 
Since \(|\tau_\theta(z)|=\Delta\) for every \(z\in H_m\), it follows that
$ 
\left|
\sum_{z\in D_j\setminus\bigcup_{r=h}^v A_r}\tau_\theta(z)
\right|
\le
t_jb\Delta.
$ 
Also,
$ 
|D_j|
\ge
\left|\bigcup_{r=h}^v A_r\right|
=
\sum_{r=h}^v |A_r|
\ge
(v-h+1)a
=
t_ja.
$ 
Combining the preceding displays and using \(t_j\ge2\), we obtain
$ 
|\nu_{j,\theta}^\star|
\le
\frac{(a+G-2)K_0+t_jb\Delta}{t_ja}
\le
\frac{a+G-2}{2a}K_0+\frac{b\Delta}{a}.
$ 
Using Lemma~\ref{lem:bound_basic_constants}, it follows that
$ 
|\nu_{j,\theta}^\star|\le \frac{161}{256}K_0.
$ 
Thus, for every \(x\in A_g\), since \(\tau_\theta(x)\in\{-K_0,K_0\}\),
$ 
L_{\widetilde\alpha}(x)
=
(\nu_{j,\theta}^\star-\tau_\theta(x))^2
\ge
\left(K_0-|\nu_{j,\theta}^\star|\right)^2
\ge
\left(\frac{95}{256}K_0\right)^2.
$ 
Combining this with \eqref{eq:proof-benchmark-anchor} gives
$ 
L_{\widetilde\alpha}(x)-L_{\mathrm{bench}}(x)
\ge
\left(\frac{95^2}{256^2}-\frac1{1024}\right)K_0^2
>
\frac18 K_0^2.
$ 
Since \(|A_g|\ge a\), this proves \eqref{eq:proof-selected-anchor}.

\medskip
\noindent
\emph{Part 3: anchorless intervals and the set \(\mathcal M_0\).}
The goal of this part is to control the number of blocks $H_m$ intersected by predictive
intervals $D_j$ that contain no anchor block $A_g$.
Specifically, define 
\[
\mathcal M_0
:=
\left\{
m\in\{1,\dots,M\}:
\exists j\in\{1,\dots,J_{\widetilde\alpha}\}
\text{ such that }D_j\cap H_m\neq\varnothing
\text{ and }t_j=0
\right\}.
\]
We first show that every predictive interval with \(t_j=0\) can intersect at most one
coded block \(H_m\), and if it intersects \(H_m\), then it is contained in \(H_m\). Suppose
\(t_j=0\) and \(D_j\cap H_m\neq\varnothing\). If \(D_j\not\subseteq H_m\), then, since
\(D_j\) is consecutive and \(H_m\) lies immediately between \(A_{2m}\) and
\(A_{2m+1}\), the interval \(D_j\) must intersect either \(A_{2m}\) or \(A_{2m+1}\). Let
this anchor be \(A_g\), with \(g\in\{2m,2m+1\}\). By 
Lemma~\ref{lem:single-anchor-canonicalization}, \(A_g\) is either entirely ignored or
entirely contained in one predictive interval. Since \(D_j\cap A_g\neq\varnothing\), the
block \(A_g\) is not entirely ignored; hence it must be entirely contained in a predictive
interval. Because predictive intervals are disjoint and \(D_j\) intersects \(A_g\), this
predictive interval must be \(D_j\). Therefore \(A_g\subseteq D_j\), contradicting
\(t_j=0\). Hence \(D_j\subseteq H_m\).

For each \(m\in\mathcal M_0\), choose one index
\(j(m)\in\{1,\dots,J_{\widetilde\alpha}\}\) such that
$ 
D_{j(m)}\cap H_m\neq\varnothing,
t_{j(m)}=0.
$ 
By the previous paragraph,
$ 
D_{j(m)}\subseteq H_m.
$ 
Since the blocks \(H_1,\dots,H_M\) are pairwise disjoint, if \(m\neq m'\), then
\(D_{j(m)}\neq D_{j(m')}\). Therefore \(m\mapsto j(m)\) is an injection from
\(\mathcal M_0\) into the set \(\{j:t_j=0\}\), whose cardinality is \(J_0\). Thus
\begin{equation}\label{eq:proof-M0}
|\mathcal M_0|\le J_0\le\delta, 
\end{equation}
where the second inequality follows from Part 1 of this proof.

\medskip
\noindent
\emph{Part 4: nonnegative excess on clean coded blocks.}
The goal of this part is to show that if \(m\notin\mathcal M_0\), then the local excess on
\(B_m:=A_{2m}\cup H_m\cup A_{2m+1}\) is nonnegative. For \(m=1,\dots,M\), define
$ 
\mathcal E_m
:=
\frac1N\sum_{x\in B_m}
\bigl(L_{\widetilde\alpha}(x)-L_{\mathrm{bench}}(x)\bigr)
$. 
We claim that
\begin{equation}\label{eq:proof-clean-block}
m\notin\mathcal M_0
\quad\Longrightarrow\quad
\mathcal E_m\ge0.
\end{equation}

Fix \(m\notin\mathcal M_0\). First suppose that \(A_{2m}\) and \(A_{2m+1}\) lie in the
same predictive interval. Since that interval is consecutive, it contains all $x$ in \(B_m\).
Using \eqref{eq:proof-two-point-ls}\footnote{Using
$ 
a(t-K_0)^2+a(t+K_0)^2=2aK_0^2+2at^2,
$ 
we can reduce the three-term least-squares problem to the two-point identity
\eqref{eq:proof-two-point-ls}. Namely,
$ 
 \frac1N\inf_{t\in\mathbb R}
\left\{
a(t-K_0)^2+a(t+K_0)^2+b(t-\theta_m\Delta)^2
\right\}  =
\frac1N\left[
2aK_0^2+
\inf_{t\in\mathbb R}
\left\{
2a(t-0)^2+b(t-\theta_m\Delta)^2
\right\}
\right] =
\frac1N\left(
2aK_0^2+\frac{2ab}{2a+b}\Delta^2
\right)
\ge
\frac{2aK_0^2}{N}. 
$},
$$  
\small 
\begin{aligned} 
\frac1N\inf_{t\in\mathbb R}
\left\{
a(t-K_0)^2+a(t+K_0)^2+b(t-\theta_m\Delta)^2
\right\}
=
\frac1N\left(2aK_0^2+\frac{2ab}{2a+b}\Delta^2\right)
\ge
\frac{2aK_0^2}{N}.
\end{aligned} 
$$ 
Combining this with \eqref{eq:proof-MbDelta-vs-aK0} and \(a\ge32b\) (Lemma \ref{lem:bound_basic_constants}) gives
$ 
\mathcal E_m
\ge
\frac{2aK_0^2}{N}-\frac{bK_0^2}{N}
\ge0.
$ 

Now suppose that \(A_{2m}\) and \(A_{2m+1}\) do not lie in the same predictive interval.
Since \(m\notin\mathcal M_0\), every predictive point of \(H_m\) belongs to a predictive
interval containing either \(A_{2m}\) or \(A_{2m+1}\); all remaining points of \(H_m\) are
ignored. Let \(\bar s_m^\theta = 2m 1\{\theta_m = -1\}  + (2m + 1) 1\{\theta_m = 1\}\). Define \(k_m\) as the number of points in
\(H_m\) either ignored or assigned to a predictive interval containing
\(A_{\bar s_m^\theta}\), and define \(u_m\) as the number of points in \(H_m\) assigned to
a predictive interval containing \(A_{\bar s_m^\theta}\). Then
$ 
0\le u_m\le k_m\le b.
$ 
From Equation \eqref{eq:proof-two-point-ls}, it follows that
$ 
\mathcal E_m\ge \frac1N\Phi_{k_m}(u_m), 
$ 
with $\Phi$ denoted in Equation \eqref{eqn:Phi_k}. 
In addition, by Lemma~\ref{lem:phi-bound},
$ 
\Phi_{k_m}(u_m)\ge0.
$ 
Therefore \(\mathcal E_m\ge0\), proving \eqref{eq:proof-clean-block}.

\medskip
\noindent

\emph{Part 5: global summation.}
The goal of this final part is to combine the bounds obtained above. Since \(G\) is even and \(M=G/2-1\), the sets
$ 
B_1,\ldots,B_M,A_G
$ 
are pairwise disjoint and their union is \(\mathcal X_N\), where
$ 
B_m:=A_{2m}\cup H_m\cup A_{2m+1}.
$ 
For each \(m\in\{1,\ldots,M\}\), define
$ 
\mathcal S_m:=\mathcal S\cap\{A_{2m},A_{2m+1}\},
\zeta_m:=|\mathcal S_m|,
$ 
and
$ 
\mathcal M_1:=\{m\in\{1,\ldots,M\}:\zeta_m\ge1\}.
$ 
Also define
\[
\mathcal E_m
:=
\frac1N\sum_{x\in B_m}
\bigl(L_{\widetilde\alpha}(x)-L_{\mathrm{bench}}(x)\bigr),
\qquad
\mathcal E_G
:=
\frac1N\sum_{x\in A_G}
\bigl(L_{\widetilde\alpha}(x)-L_{\mathrm{bench}}(x)\bigr).
\]

We first record two lower bounds that will be used below. First, for any
anchor block \(A_g\subseteq\{A_2,\ldots,A_{G-1}\}\) that is not selected in
\(\mathcal S\), we have
\begin{equation}\label{eq:proof-unselected-anchor-bound}
\frac1N\sum_{x\in A_g}
\bigl(L_{\widetilde\alpha}(x)-L_{\mathrm{bench}}(x)\bigr)
\ge
-\frac{aK_0^2}{1024N}.
\end{equation}
Indeed, \(L_{\widetilde\alpha}(x)\ge0\) for all \(x\), while
\eqref{eq:proof-benchmark-anchor} gives
$ 
L_{\mathrm{bench}}(x)\le \frac{K_0^2}{1024},
\qquad x\in A_g.
$ 
Since \(|A_g|=a\) for \(g=2,\ldots,G-1\), 
\eqref{eq:proof-unselected-anchor-bound} follows.

Second, for every \textit{bridge} block \(H_m\), we have
\begin{equation}\label{eq:proof-bridge-bound}
\frac1N\sum_{x\in H_m}
\bigl(L_{\widetilde\alpha}(x)-L_{\mathrm{bench}}(x)\bigr)
\ge
-\frac{bK_0^2}{N}.
\end{equation}
To see this, note again that \(L_{\widetilde\alpha}(x)\ge0\). Under the benchmark
partition \(\alpha^\theta\), the block \(H_m\) is pooled with \(A_{2m}\) when
\(\theta_m=1\), and with \(A_{2m+1}\) when \(\theta_m=-1\). The total benchmark
least-squares loss on the corresponding pooled anchor-and-bridge block is
$ 
\frac{ab}{a+b}(K_0-\Delta)^2.
$ 
Therefore, the benchmark contribution from \(H_m\) alone is no larger than this total
quantity. By \eqref{eq:proof-MbDelta-vs-aK0},
$ 
\frac1N\frac{ab}{a+b}(K_0-\Delta)^2
\le
\frac{bK_0^2}{N},
$ 
which proves \eqref{eq:proof-bridge-bound}.

We now lower-bound \(\mathcal E_m\) block by block. If \(m\in\mathcal M_1\), then
\(\zeta_m\in\{1,2\}\). For each selected anchor in \(\mathcal S_m\),
\eqref{eq:proof-selected-anchor} gives a contribution at least
\(aK_0^2/(8N)\). For each non-selected anchor in
\(\{A_{2m},A_{2m+1}\}\setminus\mathcal S_m\), we use
\eqref{eq:proof-unselected-anchor-bound}. Finally, for \(H_m\), we use
\eqref{eq:proof-bridge-bound}. Hence, for every \(m\in\mathcal M_1\),
\begin{equation}\label{eq:proof-Em-M1-fixed}
\mathcal E_m
\ge
\frac{\zeta_m aK_0^2}{8N}
-
\frac{(2-\zeta_m)aK_0^2}{1024N}
-
\frac{bK_0^2}{N}.
\end{equation}

If \(m\notin\mathcal M_1\), then \(\zeta_m=0\). When also \(m\notin\mathcal M_0\),
\eqref{eq:proof-clean-block} gives \(\mathcal E_m\ge0\). When instead
\(m\in\mathcal M_0\), we use the bounds
\eqref{eq:proof-unselected-anchor-bound} for both anchors and
\eqref{eq:proof-bridge-bound} for \(H_m\). Thus, for every \(m\notin\mathcal M_1\),
\begin{equation}\label{eq:proof-Em-not-M1-fixed}
\mathcal E_m
\ge
-
\left(
\frac{2a}{1024}+b
\right)
\frac{K_0^2}{N}\,
\mathbf 1\{m\in\mathcal M_0\}.
\end{equation}

For the final anchor block \(A_G\), note that the benchmark assigns \(A_G\) to its own
group and \(\tau_\theta(x)\) is constant on \(A_G\). Hence
\(L_{\mathrm{bench}}(x)=0\) for all \(x\in A_G\). Therefore \(\mathcal E_G\ge0\) if
\(A_G\notin\mathcal S\), while \eqref{eq:proof-selected-anchor} gives a lower bound
of \(aK_0^2/(8N)\) if \(A_G\in\mathcal S\). Consequently,
\begin{equation}\label{eq:proof-EG-fixed}
\mathcal E_G
\ge
\frac{aK_0^2}{8N}\mathbf 1\{A_G\in\mathcal S\}.
\end{equation}

Combining \eqref{eq:proof-loss-difference-pointwise},
\eqref{eq:proof-Em-M1-fixed}, \eqref{eq:proof-Em-not-M1-fixed}, and
\eqref{eq:proof-EG-fixed}, we obtain
\[
\small 
\begin{aligned}
&
R_\theta(\widetilde\alpha,\nu_\theta^\star)
-
R_\theta(\alpha^\theta,\mu^\theta)
=
\sum_{m=1}^M\mathcal E_m+\mathcal E_G
\\
&\ge
\frac{aK_0^2}{8N}
\left(
\sum_{m\in\mathcal M_1}\zeta_m
+
\mathbf 1\{A_G\in\mathcal S\}
\right)
-
\frac{aK_0^2}{1024N}
\sum_{m\in\mathcal M_1}(2-\zeta_m)
-
\frac{bK_0^2}{N}|\mathcal M_1|
-
\left(\frac{2a}{1024}+b\right)
\frac{K_0^2}{N}
|\mathcal M_0\cap\mathcal M_1^c|.
\end{aligned}
\]
By Part 1,
$ 
|\mathcal S|
=
\sum_{m\in\mathcal M_1}\zeta_m
+
\mathbf 1\{A_G\in\mathcal S\}
=
\delta.
$ 
Moreover,
$ 
|\mathcal M_1|
\le
\sum_{m\in\mathcal M_1}\zeta_m
\le
\delta,
|\mathcal M_0\cap\mathcal M_1^c|
\le
|\mathcal M_0|
\le
\delta
$ 
by \eqref{eq:proof-M0}. Also, since \(\zeta_m\in\{1,2\}\) for
\(m\in\mathcal M_1\),
$ 
\sum_{m\in\mathcal M_1}(2-\zeta_m)
\le
|\mathcal M_1|
\le
\delta.
$ 
Therefore,
\[
\begin{aligned}
R_\theta(\widetilde\alpha,\nu_\theta^\star)
-
R_\theta(\alpha^\theta,\mu^\theta)
&\ge
\delta
\left(
\frac a8
-
\frac{3a}{1024}
-
2b
\right)
\frac{K_0^2}{N}.
\end{aligned}
\]
Using \(a\ge32b\) from Lemma~\ref{lem:bound_basic_constants}, we have
$ 
2b\le \frac a{16}, 
$ which implies 
$ 
\frac a8-\frac{3a}{1024}-2b
\ge
\frac a{32}.
$ 
It follows that
\[
R_\theta(\widetilde\alpha,\nu_\theta^\star)
-
R_\theta(\alpha^\theta,\mu^\theta)
\ge
\delta\,\frac{aK_0^2}{32N}.
\]
Finally, since \(\delta\ge1\) from Part 1 under the maintained assumption
\(\widetilde\alpha\notin\mathcal B_{\mathrm{good}}\), we conclude that
$ 
R_\theta(\widetilde\alpha,\nu_\theta^\star)
-
R_\theta(\alpha^\theta,\mu^\theta)
\ge
\frac{aK_0^2}{32N},
$ 
which proves \eqref{eq:proof-bad-structural}.

\subsubsection{Proof of Lemma \ref{lem:good_geometry}} \label{proof:lem:good_geometry} 
Since \(\widetilde\alpha\in\mathcal B_{\mathrm{good}}\), every anchor \(A_g\) is contained
in some \(D_j\), and no \(D_j\) contains more than one anchor. Thus the \(G-1\) anchor
blocks \(A_2,\dots,A_G\) lie in \(G-1\) distinct predictive intervals. Since
\(J_{\widetilde\alpha}\le G-1\), \(J_{\widetilde\alpha}=G-1\), and every \(D_j\) contains
exactly one anchor.

Fix \(m\). The intervals \(D_m^L\) and \(D_m^R\) are distinct because they contain distinct
anchors. We show \(D_m^L\subseteq A_{2m}\cup H_m\). If \(m=1\), \(A_2\) is the leftmost
block. If \(m\ge2\), a point of \(D_m^L\) strictly to the left of \(A_{2m}\), together
with consecutiveness and \(A_{2m}\subseteq D_m^L\), would force
\(D_m^L\cap A_{2m-1}\neq\varnothing\), contradicting that \(D_m^L\) contains exactly one
anchor. Similarly, a point of \(D_m^L\) strictly to the right of \(H_m\) would force
\(D_m^L\cap A_{2m+1}\neq\varnothing\). Hence \(D_m^L\subseteq A_{2m}\cup H_m\). Therefore
\(L_m:=D_m^L\cap H_m\) is an initial segment of \(H_m\), and
\(D_m^L=A_{2m}\cup L_m\).

The argument for \(D_m^R\) is symmetric. Since \(m\le M\) and \(M=G/2-1\),
\(2m+2\le G\). A point of \(D_m^R\) strictly to the left of \(H_m\) would force
intersection with \(A_{2m}\), while a point strictly to the right of \(A_{2m+1}\) would
force intersection with \(A_{2m+2}\). Both contradict the one-anchor property. Hence
\(D_m^R\subseteq H_m\cup A_{2m+1}\), so \(R_m:=D_m^R\cap H_m\) is a terminal segment of
\(H_m\), and \(D_m^R=R_m\cup A_{2m+1}\).

No third predictive interval can meet \(B_m:=A_{2m}\cup H_m\cup A_{2m+1}\). Indeed, if
\(D_j\cap B_m\neq\varnothing\) and \(D_j\neq D_m^L,D_m^R\), then
\(D_j\cap A_{2m}=D_j\cap A_{2m+1}=\varnothing\), so \(D_j\cap B_m\subseteq H_m\). If
\(D_j\) contained any anchor to the left of \(A_{2m}\), consecutiveness would force
intersection with \(A_{2m}\); if it contained any anchor to the right of \(A_{2m+1}\),
consecutiveness would force intersection with \(A_{2m+1}\). Thus \(D_j\) would contain no
anchor, contradicting the first paragraph.

Therefore the only predictive intervals meeting \(B_m\) are \(D_m^L\) and \(D_m^R\).
Since they are disjoint, \(L_m\cap R_m=\varnothing\), and every point of
\(H_m\setminus(L_m\cup R_m)\) is ignored. The representation of \(L_m\) and \(R_m\) as an
initial and terminal segment of \(H_m=\{x_{q_m},\dots,x_{q_m+b-1}\}\), and the identity
\(l_m+r_m+s_m=b\), follow directly from the definitions.
 
\subsubsection{Proof of Lemma \ref{lem:local_good_property}} \label{proof:lem:local_good_property} 
For \(m=1,\dots,M\), set \(B_m:=A_{2m}\cup H_m\cup A_{2m+1}\) and denote
$ 
\ell_m(\widetilde\alpha,\nu_\theta^\star)
:=
\frac1N\sum_{x\in B_m}L_{\widetilde\alpha}(x),
\ell_m(\alpha^\theta,\mu^\theta)
:=
\frac1N\sum_{x\in B_m}L_{\mathrm{bench}}(x).
$ 
Define
$ 
u_m:= r_m 1\{\theta_m = 1\} + l_m 1\{\theta_m = -1\}$ 
Then \(0\le u_m\le k_m\le b\), \(k_m-u_m=s_m\). Also, by construction,
$ 
b-k_m= l_m 1\{\theta_m = 1\} + r_m 1\{\theta_m = -1\}$. 
Using Lemma~\ref{lem:good_geometry} and \eqref{eq:proof-two-point-ls}, we can write
\begin{equation}\label{eq:proof-local-exact}
\ell_m(\widetilde\alpha,\nu_\theta^\star)
=
\frac1N\left[
\frac{a(b-k_m)}{a+b-k_m}(K_0-\Delta)^2
+
\frac{au_m}{a+u_m}(K_0+\Delta)^2
+
(k_m-u_m)K_0^2
\right].
\end{equation}
Equation \eqref{eq:proof-local-exact} follows from the fact that if \(\theta_m=+1\), the first term corresponds to \(A_{2m}\cup L_m\) and the
second to \(R_m\cup A_{2m+1}\); if \(\theta_m=-1\), the roles of \(l_m\) and \(r_m\) are
reversed. The ignored portion has size \(s_m=k_m-u_m\).

Under the benchmark rule, all \(x\in H_m\) have label
\(2m\) if \(\theta_m=+1\), and \(2m+1\) if \(\theta_m=-1\). The
other adjacent anchor to $H_m$ which is not pooled with $H_m$ ($A_{2m}$ if $\theta_m = -1$ and $A_{2m + 1}$ if $\theta_m = 1)$ contributes zero loss by construction. Hence
\begin{equation}\label{eq:proof-local-benchmark}
\ell_m(\alpha^\theta,\mu^\theta)
=
\frac1N\frac{ab}{a+b}(K_0-\Delta)^2.
\end{equation}
Subtracting \eqref{eq:proof-local-benchmark} from \eqref{eq:proof-local-exact} gives
$ 
\ell_m(\widetilde\alpha,\nu_\theta^\star)-\ell_m(\alpha^\theta,\mu^\theta)
=
\frac1N\Phi_{k_m}(u_m).
$ 
By Lemma~\ref{lem:phi-bound},
$ 
\ell_m(\widetilde\alpha,\nu_\theta^\star)-\ell_m(\alpha^\theta,\mu^\theta)
\ge
\frac{K_0\Delta}{N}k_m.
$ 
Summing over \(m\), and using that \(B_1,\dots,B_M,A_G\) form a disjoint partition of
\(\mathcal X_N\), the contribution on \(A_G\) is nonnegative because the benchmark loss on
\(A_G\) is zero and \(L_{\widetilde\alpha}\ge0\). This proves \eqref{eq:proof-sum-local}.

\subsubsection{Proof of Lemma \ref{lem:hamming_reduction}} \label{proof:lem:hamming_reduction}
If \(\widetilde\alpha\in\mathcal B_{\mathrm{good}}\), then Lemma~\ref{lem:good_geometry}
gives \(q_m(\widetilde\alpha)=|L_m|=l_m\). If \(\theta_m=+1\) and
\(\widehat\theta_m\neq\theta_m\), then \(l_m<b/2\), so
\(k_m=r_m+s_m=b-l_m>b/2\). If \(\theta_m=-1\) and
\(\widehat\theta_m\neq\theta_m\), then \(l_m\ge b/2\), so
\(k_m=l_m+s_m\ge b/2\). Hence, for \(\widetilde\alpha\in\mathcal B_{\mathrm{good}}\),
$ 
k_m\ge \frac b2\,\mathbf 1\{\widehat\theta_m\neq\theta_m\}.
$ 
Summing and applying Lemma~\ref{lem:local_good_property},
$ 
R_\theta(\widetilde\alpha,\nu_\theta^\star)-R_\theta(\alpha^\theta,\mu^\theta)
\ge
\frac{bK_0\Delta}{2N}d_H(\widehat\theta,\theta)$  if $\widetilde\alpha\in\mathcal B_{\mathrm{good}}$.
If \(\widetilde\alpha\notin\mathcal B_{\mathrm{good}}\), Lemma~\ref{lem:structural_property}
gives
$ 
R_\theta(\widetilde\alpha,\nu_\theta^\star)-R_\theta(\alpha^\theta,\mu^\theta)
\ge
\frac{aK_0^2}{32N}.
$ 
Combining the two expressions we have 
$$  
\small  
\begin{aligned}
R_\theta(\widetilde\alpha,\nu_\theta^\star)-R_\theta(\alpha^\theta,\mu^\theta)
&\ge
\frac{aK_0^2}{32N}\mathbf 1\{\widetilde\alpha\notin\mathcal B_{\mathrm{good}}\} + \frac{bK_0\Delta}{2N}d_H(\widehat\theta,\theta)
\mathbf 1\{\widetilde\alpha\in\mathcal B_{\mathrm{good}}\}.
\end{aligned}
$$ 
Since \(d_H(\widehat\theta,\theta)\le M\),
$ 
d_H(\widehat\theta,\theta)\mathbf 1\{\widetilde\alpha\in\mathcal B_{\mathrm{good}}\}
\ge
d_H(\widehat\theta,\theta)
-
M\mathbf 1\{\widetilde\alpha\notin\mathcal B_{\mathrm{good}}\}.
$ 
Using \eqref{eq:proof-MbDelta-vs-aK0}, $a K_0^2/(32N) -M b K_0\Delta/(2N) \ge 0$, so that 
$ 
R_\theta(\widetilde\alpha,\nu_\theta^\star)-R_\theta(\alpha^\theta,\mu^\theta)
\ge
\frac{bK_0\Delta}{2N}d_H(\widehat\theta,\theta).
$ 
Finally, \eqref{eq:proof-alpha-tilde-final} implies
\(R_\theta(\alpha,\mu)\ge R_\theta(\widetilde\alpha,\nu_\theta^\star)\), which proves
\eqref{eq:proof-pointwise-final}.

\section{Additional extensions} 

\subsection{Breakeven analysis on prediction error} \label{app:more_more_breakeven}

In this section we present a complementary break-even analysis based on the prediction accuracy of the procedure. Let $\hat{R}^{(m)}, \hat{\pi}^{(m)}, \hat{f}^{(m)}$ be as defined in Section \ref{app:breakeven_other}. 
 We will adopt the convention that $R^{(0)}(f)$ denotes $R(f, \pi = 1)$ where all units are in the prediction set and $f$ is the prediction function and $\widehat{R}^{(0)}(f)$ its in-sample estimator with $\hat{f}^{(0)} \in \mathrm{arg} \min_{f \in \mathcal{F}} \widehat{R}^{(0)}(f)$. Via sample splitting, take an independent validation sample as in \ref{sec:model_selection} (Remark \ref{sec:sample_splitting}), with $\hat{\tau}_{oos}(x)$ denoting the estimated $\tau(x)$ out-of-sample.  
For each \(m\in\mathcal M\), define 
$
\Delta(m)  = \sum_x p(x) \left[(\hat{f}^{(m)}(x) - \tau(x))^2 - (\hat{f}^{(0)}(x) - \tau(x))^2 \right] \hat{\pi}^{(m)}(x),  \widehat \Delta(m)  = \sum_x p(x) \left[(\hat{f}^{(m)}(x) - \hat{\tau}_{oos}(x))^2 - (\hat{f}^{(0)}(x) - \hat{\tau}_{oos}(x))^2 \right] \hat{\pi}^{(m)}(x)
$ 
as the difference in loss with respect to the population $\tau(x)$ on the prediction set under the prediction rule $\hat{\pi}^{(m)}$ and its estimator out-of-sample. 
Negative values of \(\Delta(m)\) indicate that, under policy $\pi^{(m)}$ corresponding to abstention cost \(c_m\), the estimated ignorance-aware rule is better than the estimated prediction that predicts everywhere on the prediction set. 

Denote, for any $m \in \mathcal{M}$ (with $N_{\mathcal{X}}:= |\mathcal{X}|$) 
$
\sigma_m^2 = \mathbb{V}\left(\sqrt{N_{\mathcal{X}} } \widehat \Delta(m) | \mathcal{T} \right), \mathcal{T} = \left(\hat{f}^{(0)}, \Big(\hat{f}^{(m)}, \hat{\pi}^{(m)}\Big)_{m \in \mathcal{M}}\right)
$ 
the variance of $\sqrt{N_{\mathcal{X}}}(\widehat{R}_{oos}^{(m)}(\hat{f}^{(m)}, \hat{\pi}^{(m)}) - \widehat{R}_{oos}^{(0)}(\hat{f}^{(0)}))$, conditional on the training sample.  
We take a variance estimate $S_m$, so that $P(S_m^2 \ge \sigma_m^2) \rightarrow 1$ for any $m \in \mathcal{M}$. We provide formal expressions for $S_m^2$ below.  For any $j \le |\mathcal{M}|$, we then construct a statistic 
\begin{equation} \label{eqn:S_l}
\small 
\begin{aligned} 
U_j = \max_{l \ge j} \left\{\widehat\Delta(m_l)+ \Phi^{-1}(1 - \alpha) \cdot \frac{S_{m_l}}{\sqrt{N_{\mathcal{X}}}}\right\}
\end{aligned} 
\end{equation}
where $\Phi(\cdot)$ is the CDF of a Standard Normal random variable, and $\alpha \in (0,1/2)$. As in Section \ref{app:breakeven_other}, we choose 
$
\hat{j}^\star = \min\{j : U_j < 0\}, \hat{m}^\star = m_{\hat{j}^\star}
$
with the convention that if $U_j \ge 0$ for all $j$, $\hat{m}^\star = \infty$.

\begin{thm}[Inference for break-even analysis] \label{thm:breakeven} Let Assumptions \ref{ass:general_regret}, \ref{ass:entropy_F} hold for the in-sample and out-of-sample estimates and assume the two are independent. Let $\sigma_m^2 > \underline{l}$ for a positive constant $\underline{l} > 0$ almost surely for all $m \in \mathcal{M}$, with $|\mathcal{M}|$ finite. Take any $S_m^2$ so that $\lim_{N_{\mathcal{X}} \rightarrow \infty} P(S_m^2 \ge \sigma_m^2) = 1$ for all $m \in \mathcal{M}$. Then
$
\limsup_{N_{\mathcal{X}} \rightarrow \infty} P(\hat{m}^\star < \infty, \Delta(\hat{m}^\star) \ge 0) \le \alpha. 
$
\end{thm}

See Appendix \ref{proof:thm:breakeven} for the proof.  We are left to derive an expression for a uniform bound $S_m^2$. Take $\xi_x^{(m)}
=
\Big\{(\hat f^{(m)}(x)-\hat\tau_{oos}(x))^2 - (\hat{f}^{(0)}(x)-\hat\tau_{oos}(x))^2\Big\}\hat\pi^{(m)}(x)$.  
A natural approach is to take $S_m^2 = N_{\mathcal{X}}\sum_x p(x)^2 \mathbb{E}[(\xi_x^{(m)} - \mathbb{E}[\bar{\xi}^m| \mathcal{T}])^2| \mathcal{T}]$ where  $\bar{\xi}^m = \sum_x p(x) \xi_x^{(m)}$. 

By replacing the expectations with the empirical average, this approach would provide us with a consistent estimate of a valid upper bound $S_m^2$ (which averages across types), where $\mathbb{E}[\bar{\xi}^m| \mathcal{T}]$ can be consistently estimated by the law of large numbers (since $\xi_x$ are independent conditional on $\mathcal{T}$ and have bounded second moment).

\subsection{Approximate clustering with gradient descent}

Algorithm \ref{alg:gmeans_abstention} describes in details gradient descent for clustering methods. In practice, because the algorithm may not converge to optimum we recommend comparing its solution against the exact solution with hard stopping time to compare the two. 

\begin{algorithm}[!ht]
\footnotesize
\caption{Approximate multi-dimensional clustering with abstention}
\label{alg:gmeans_abstention}
\begin{algorithmic}[1]
\Require Types \(x_1,\dots,x_{|\mathcal X|}\in\mathbb R^d\), estimates \(\hat\tau(x_i)\), variance estimates \(\hat\eta(x_i)^2\), abstention costs \(\hat c(x_i)\), weights \(p(x_i)\), number of predictive groups \(G-1\), temperature \(\rho>0\), smooth-min parameter \(\lambda>0\), initial centroids \(z^{(0)}=(z_2^{(0)},\dots,z_G^{(0)})\), minimum cluster size \(\underline\kappa |\mathcal X|\), step sizes \((\gamma_t)_{t\ge 0}\)

\For{\(t=0,1,2,\dots\)}
    \State \textbf{Soft cluster weights:} for each \(i\) and \(g\in\{2,\dots,G\}\), set
    $ 
    w_{ig}^{(t)}
    :=
    \frac{\exp\!\left(-\|x_i-z_g^{(t)}\|^2/\rho\right)}
    {\sum_{h=2}^G \exp\!\left(-\|x_i-z_h^{(t)}\|^2/\rho\right)}.
    $

    \State \textbf{Profiled weighted cluster means:} for each \(g\in\{2,\dots,G\}\), define, for any candidate centroids \(z\),
    $ 
    \hat\mu_g(z)
    :=
    \frac{\sum_{i=1}^{|\mathcal X|} p(x_i)\,w_{ig}(z,\rho)\,\hat\tau(x_i)}
    {\sum_{i=1}^{|\mathcal X|} p(x_i)\,w_{ig}(z,\rho)}.
    $

    \State \textbf{Profiled soft prediction and ignorance costs:} for each \(g\in\{2,\dots,G\}\), define
    $ 
    \widehat C_{g,\rho}^{\mathrm{pred}}(z)
    :=
    \sum_{i=1}^{|\mathcal X|}
    p(x_i)\,w_{ig}(z,\rho)
    \Big\{
    \big(\hat\mu_g(z)-\hat\tau(x_i)\big)^2-\hat\eta(x_i)^2
    \Big\},
    $ 
    and
    $ 
    \widehat C_{g,\rho}^{\mathrm{ign}}(z)
    :=
    \sum_{i=1}^{|\mathcal X|}
    p(x_i)\,w_{ig}(z,\rho)\,\hat c(x_i).
    $ 
    \State \textbf{Smoothed profiled empirical loss:} define
    $ 
    \widehat Q_{\rho,\lambda}(z)
    :=
    \sum_{g=2}^G
    \operatorname{smin}_{\lambda}
    \left\{
    \widehat C_{g,\rho}^{\mathrm{pred}}(z),
    \widehat C_{g,\rho}^{\mathrm{ign}}(z)
    \right\},
    $ 
    where
    $ 
    \operatorname{smin}_{\lambda}(a,b)
    :=
    -\lambda
    \log\left(
    \exp(-a/\lambda)+\exp(-b/\lambda)
    \right).
    $ 
    As \(\lambda\downarrow 0\), \(\operatorname{smin}_{\lambda}(a,b)\) approaches \(\min\{a,b\}\).

    \State \textbf{Gradient step:} update the centroids by
    $ 
    z^{(t+1)}
    =
    z^{(t)}
    -
    \gamma_t
    \nabla_z \widehat Q_{\rho,\lambda}(z^{(t)}).
    $
    \State \textbf{Hard assignment:} for each \(i\), set
    \[
    \alpha^{(t+1)}(x_i)
    :=
    \arg\min_{g\in\{2,\dots,G\}} \|x_i-z_g^{(t+1)}\|^2.
    \]

    \For{each cluster \(g\in\{2,\dots,G\}\)}
        \State Define the profiled hard-cluster mean
        $ 
        \hat\mu_g^{(t+1)}
        :=
        \frac{
        \sum_{i:\alpha^{(t+1)}(x_i)=g}
        p(x_i)\hat\tau(x_i)
        }{
        \sum_{i:\alpha^{(t+1)}(x_i)=g}
        p(x_i)
        }.
        $

        \State Define the empirical loss from keeping cluster \(g\) predictive and from assigning it to ignorance as
        \[
        \widehat L_g^{\mathrm{pred}}
        :=
        \sum_{i:\alpha^{(t+1)}(x_i)=g}
        p(x_i)
        \Big\{
        \big(\hat\mu_g^{(t+1)}-\hat\tau(x_i)\big)^2-\hat\eta(x_i)^2
        \Big\},
        \qquad
        \widehat L_g^{\mathrm{ign}}
        :=
        \sum_{i:\alpha^{(t+1)}(x_i)=g}
        p(x_i)\hat c(x_i).
        \]

        \If{\(\#\{i:\alpha^{(t+1)}(x_i)=g\}<\underline\kappa |\mathcal X|\) or \(\widehat L_g^{\mathrm{pred}}>\widehat L_g^{\mathrm{ign}}\)}
            \State Set \(\pi^{(t+1)}(x_i)=0\) for all \(i\) with \(\alpha^{(t+1)}(x_i)=g\)
        \Else
            \State Set \(\pi^{(t+1)}(x_i)=1\) for all \(i\) with \(\alpha^{(t+1)}(x_i)=g\)
        \EndIf
    \EndFor

    \If{\(z^{(t+1)}\) and \(\pi^{(t+1)}\) have converged}
        \State \textbf{break}
    \EndIf
\EndFor

\State \Return final centroids \(z^{(t+1)}\), partition \(\alpha^{(t+1)}\), and abstention rule \(\pi^{(t+1)}\)
\end{algorithmic}
\end{algorithm}

\section{Proofs for extensions} \label{app:more_proofs_extensions}

\subsection{Proof of Lemma \ref{lem:estimated_p_unbiased}} 
\label{proof:lem:estimated_p_unbiased}
For each \(x\in\mathcal X\), define
$ 
\widehat \rho_x(f,\pi)
:=
\left\{\bigl(f(x)-\hat\tau(x)\bigr)^2-\hat\eta(x)^2\right\}\pi(x)
+
\hat c(x)\bigl(1-\pi(x)\bigr).
$ 
By the law of iterated expectations,
$ 
\mathbb E\!\left[S_x \widehat \rho_x(f,\pi)\right]
=
\mathbb E\!\left[
\mathbb E\!\left[S_x \widehat \rho_x(f,\pi)\mid S_x\right]
\right]
=
\mathbb E\!\left[
S_x\,\mathbb E\!\left[\widehat \rho_x(f,\pi)\mid S_x\right]
\right].
$ 
Next, using Assumption \ref{ass:sampling2}(ii),
$$ 
\small 
\begin{aligned} 
\mathbb E\!\left[\bigl(f(x)-\hat\tau(x)\bigr)^2-\hat\eta(x)^2 \mid S_x\right]
=
\bigl(f(x)-\tau(x)\bigr)^2
+
\mathbb V(\hat\tau(x)\mid S_x)
-
\mathbb E[\hat\eta(x)^2\mid S_x]
=
\bigl(f(x)-\tau(x)\bigr)^2,
\end{aligned}
$$
and
$ 
\mathbb E[\hat c(x)\mid S_x]=c(x).
$ 
Therefore
$ 
\mathbb E\!\left[\widehat \rho_x(f,\pi)\mid S_x\right]
=
\bigl(f(x)-\tau(x)\bigr)^2\pi(x)+c(x)\bigl(1-\pi(x)\bigr).
$ 
Substituting back,
$ 
\mathbb E\!\left[S_x \widehat \rho_x(f,\pi)\right]
=
\mathbb E[S_x]\,
\Big(
\bigl(f(x)-\tau(x)\bigr)^2\pi(x)+c(x)\bigl(1-\pi(x)\bigr)
\Big).
$ 
By Assumption \ref{ass:sampling2}(i), \(\mathbb E[S_x]=N_Sp(x)\), so
$ 
\mathbb E\!\left[S_x \widehat \rho_x(f,\pi)\right]
=
N_Sp(x)\,
\Big(
\bigl(f(x)-\tau(x)\bigr)^2\pi(x)+c(x)\bigl(1-\pi(x)\bigr)
\Big).
$ 
Summing over \(x\in\mathcal X\) and dividing by \(N_S \) yields
$ 
\mathbb E[\widehat R_S(f,\pi)]
=
\sum_{x\in\mathcal X} p(x)
\Big(
\bigl(f(x)-\tau(x)\bigr)^2\pi(x)+c(x)\bigl(1-\pi(x)\bigr)
\Big)
=
R_c(f,\pi). 
$

\subsection{Proof of Theorem \ref{thm:regret_sampling}}  \label{proof:thm:regret_sampling}

First, we introduce some notation. We denote
$ 
\tau_i:=\tau(x_i), \hat\tau_i:=\hat\tau(x_i), 
\hat\eta_i^2:=\hat\eta(x_i)^2, 
\hat c_i:=\hat c(x_i), 
c_i:=c(x_i), 
S_i:=S_{x_i},
$ 
and let
$
\lambda_i:=\frac{S_i}{\bar p}.
$ 
By Assumption \ref{ass:sampling2}(i),
$ 
0\le \lambda_i\le 1
\text{ almost surely,}
$ 
and
$ 
\mathbb E[n_S]
=
\sum_{i=1}^{|\mathcal X|}\mathbb E[S_i]
=
N_S.
$ 
Define
\[
\widehat R_S(f,\pi)
:=
\frac{1}{N_S}
\sum_{i=1}^{|\mathcal X|}
S_i
\left[
\left\{\bigl(f(x_i)-\hat\tau_i\bigr)^2-\hat\eta_i^2\right\}\pi(x_i)
+
\hat c_i\bigl(1-\pi(x_i)\bigr)
\right].
\]

Let
$ 
(\pi^\star,f^\star)\in\arg\min_{\pi\in\Pi,\ f\in\mathcal F} R_c(f,\pi).
$ 
Since \((\hat f_S,\hat\pi_S)\) minimizes \(\widehat R_S(f,\pi)\) over \((\mathcal F,\Pi)\),
$ 
\widehat R_S(\hat f_S,\hat\pi_S)\le \widehat R_S(f^\star,\pi^\star).
$ 
Therefore,
\begin{equation}\label{eq:basic_in_sampling}
R_c(\hat f_S,\hat\pi_S)-R_c(f^\star,\pi^\star)
\le
\bigl(R_c-\widehat R_S\bigr)(\hat f_S,\hat\pi_S)
+
\bigl(\widehat R_S-R_c\bigr)(f^\star,\pi^\star)
\le
2\sup_{\pi\in\Pi,\ f\in\mathcal F}\bigl|(\widehat R_S-R_c)(f,\pi)\bigr|.
\end{equation}
Hence it suffices to bound
$ 
\mathbb E\left[\sup_{\pi\in\Pi,\ f\in\mathcal F}\bigl|(\widehat R_S-R_c)(f,\pi)\bigr|\right].
$

\paragraph{Step 1: Symmetrization of the centered empirical process.}
For every \((f,\pi)\), define
\[
\mathcal Z_i(f,\pi)
:=
S_i
\left[
\left\{\bigl(f(x_i)-\hat\tau_i\bigr)^2-\hat\eta_i^2\right\}\pi(x_i)
+
\hat c_i\bigl(1-\pi(x_i)\bigr)
\right].
\]
Then by Lemma \ref{lem:estimated_p_unbiased},
$  
\mathbb E[\widehat R_S(f,\pi)] = R_c(f,\pi)$, which implies we can write 
$ 
(\widehat R_S-R_c)(f,\pi)
=
\frac{1}{N_S}\sum_{i=1}^{|\mathcal X|}\Big(\mathcal Z_i(f,\pi)-\mathbb E[\mathcal Z_i(f,\pi)]\Big).
$ 
Let \(\sigma_i\) be i.i.d. Rademacher random variables independent of the data. By Lemma \ref{lem:symmetrization},
$$ 
\small 
\begin{aligned} 
\mathbb E\left[\sup_{\pi\in\Pi,\ f\in\mathcal F}\bigl|(\widehat R_S-R_c)(f,\pi)\bigr|\right]
\le
2\,
\mathbb E\left[
\sup_{\pi\in\Pi,\ f\in\mathcal F}
\left|
\frac{1}{N_S}\sum_{i=1}^{|\mathcal X|}\sigma_i \mathcal Z_i(f,\pi)
\right|
\right].
\end{aligned} 
$$ 
Using
$ 
\mathbb E[\hat\tau_i^2\mid S_i]
=
\tau_i^2+\mathbb E[\hat\eta_i^2\mid S_i],
$ 
which follows from Assumption \ref{ass:sampling2}(ii), we can rewrite (adding and subtracting the relevant terms)
\[
\small
\begin{aligned}
\mathcal Z_i(f,\pi)
=
S_i\Big[
\bigl(\hat\tau_i^2-\mathbb E[\hat\tau_i^2\mid S_i]\bigr)\pi(x_i)
-2f(x_i)(\hat\tau_i-\tau_i)\pi(x_i)
-\bigl(\hat\eta_i^2-\mathbb E[\hat\eta_i^2\mid S_i]\bigr)\pi(x_i) \\
\qquad\qquad
+(\hat c_i-c_i)\bigl(1-\pi(x_i)\bigr)
+\bigl(f(x_i)-\tau_i\bigr)^2\pi(x_i)
+c_i\bigl(1-\pi(x_i)\bigr)
\Big].
\end{aligned}
\]
Therefore, by the triangle inequality,
\begin{equation} \label{eqn:equation_all}
\small
\begin{aligned}
&\mathbb E\left[\sup_{\pi\in\Pi,\ f\in\mathcal F}\bigl|(\widehat R_S-R_c)(f,\pi)\bigr|\right]
 \le
\frac{2 \bar{p}}{N_S}\mathbb{E}[T_1+T_2+T_3+T_4+T_5+T_6],
\end{aligned}
\end{equation} 
where we have 
$$
\footnotesize 
\begin{aligned} 
T_1
&:=&
\mathbb 
\sup_{\pi\in\Pi}
\left|
\sum_i \sigma_i \frac{S_i}{\bar{p}}
\bigl(\hat\tau_i^2-\mathbb E[\hat\tau_i^2\mid S_i]\bigr)\pi(x_i)
\right|, \quad 
T_2
&:=&
2 K
\sup_{\pi\in\Pi,\ f\in\mathcal F}
\left|
\sum_i \sigma_i \frac{S_i}{\bar{p}}
\frac{f(x_i)}{K}(\hat\tau_i-\tau_i)\pi(x_i)
\right|
\\ 
T_3
& :=&
\mathbb \sup_{\pi\in\Pi}
\left|
\sum_i \sigma_i \frac{S_i}{\bar{p}}
\bigl(\hat\eta_i^2-\mathbb E[\hat\eta_i^2\mid S_i]\bigr)\pi(x_i)
\right|,
T_4
&:=&
\sup_{\pi\in\Pi}
\left|
\sum_i \sigma_i \frac{S_i}{\bar{p}}
(\hat c_i-c_i)(1-\pi(x_i))
\right|, \\ 
T_5
&:=& 
\sup_{\pi\in\Pi,\ f\in\mathcal F}
(K + B)^2 \left|
\sum_i \sigma_i \frac{S_i}{\bar{p}}
\frac{\bigl(f(x_i)-\tau_i\bigr)^2}{ (K+B)^2}\pi(x_i)
\right|,
T_6
&:=&
\sup_{\pi\in\Pi}
\left|
\sum_i \sigma_i \frac{S_i}{\bar{p}}
c_i(1-\pi(x_i))
\right|, 
\end{aligned} 
$$
where for $T_2$ and $T_5$ we divided and multiply by the upper envelope $K$ and $(K + B)^2$ respectively. 
Next, we bound each of these expectations conditional on the vector $(S)_{i=1}^n$ and then apply the law of iterated expectations. 
\paragraph{Step 2: complexity bounds.} The following complexity bounds hold. 
\begin{itemize} 
\item[(i)]
By the same argument as in the proof of Theorem \ref{thm:regret}, Assumption \ref{ass:bounded_Pi} implies that, for any probability measure \(Q\) on \(\mathcal X\),
$
\int_0^2\sqrt{\log \mathcal N(\varepsilon,\Pi,L_1(Q))}\,d\varepsilon
\le c_0\sqrt{v_\Pi}
$ 
for a universal constant \(c_0<\infty\). The same bound also holds for the class \(\{1-\pi:\pi\in\Pi\}\) by Lemma \ref{lem:vc_sum}. 
\item[(ii)] Similarly, following Step 4 of the proof of Theorem \ref{thm:regret}, We can bound the covering number of the function class $\mathcal{H};= \{x \mapsto \frac{f(x)}{K}\pi(x): f\in \mathcal{F}, \pi \in \Pi\}$, so that $\int_0^{2} \sqrt{\log(\mathcal{N}(\varepsilon, \mathcal{H}, L_1(Q))} d\varepsilon \le c_0  \sqrt{C(\mathcal{F}) + v_\Pi}$ for a universal constant $c_0 < \infty$. 
\item[(iii)]Define
$$
\small 
\begin{aligned} 
\mathcal G_1
:=
\left\{
x\mapsto \frac{\bigl(f(x)-\tau(x)\bigr)^2}{(K + B)^2} \pi(x):
f\in\mathcal F,\ \pi\in\Pi
\right\}, \qquad \mathcal F_\tau^2
:=
\left\{
x\mapsto \frac{\bigl(f(x)-\tau(x)\bigr)^2}{(K + B)^2}:
f\in\mathcal F
\right\}. 
\end{aligned} 
$$
Since \(|f(x)|\le K\) and \(|\tau(x)|\le B\),
$ 
\sup_{g\in\mathcal F_\tau^2,\ x\in\mathcal X}|g(x)|\le 1.
$ 
Moreover, for any \(f,g\in\mathcal F\),
$ 
\bigl|(f(x)-\tau(x))^2-(g(x)-\tau(x))^2\bigr|
=
|f(x)-g(x)|\,|f(x)+g(x)-2\tau(x)|
\le
2(K+B)|f(x)-g(x)|.
$ 
Therefore, since $\mathcal{F}_\tau^2$ divides by $(K + B)^2$
$ 
\mathcal N(\varepsilon ,\mathcal F_\tau^2,L_1(Q))
\le
\mathcal N\!\left(\frac{\varepsilon (K+B)}{4},\mathcal F,L_1(Q)\right).
$ 
By Lemma \ref{lem:cover_product},
$ 
\mathcal N(\varepsilon,\mathcal G_1,L_1(Q))
\le
\mathcal N\!\left(\frac{\varepsilon}{2},\mathcal F_\tau^2,L_1(Q)\right)
\mathcal N\!\left(\frac{\varepsilon}{2},\Pi,L_1(Q)\right) \le \mathcal N\!\left(\frac{\varepsilon (K + B)}{4},\mathcal F,L_1(Q)\right)
\mathcal N\!\left(\frac{\varepsilon}{2},\Pi,L_1(Q)\right) .
$ 
Hence, from a simple change of variable 
$ 
\sup_Q
\int_0^{2}\sqrt{\log \mathcal N(\varepsilon,\mathcal G_1,L_1(Q))}\,d\varepsilon
\le
c_0'\Big(\frac{1}{(K + B)} \sqrt{C(\mathcal F)}+\sqrt{v_\Pi}\Big)
$ 
for a universal constant \(c_0'<\infty\).

\end{itemize} 

\paragraph{Step 3: Moments bound.} Using the complexity bounds in Step 2, we next invoke Lemma \ref{lem:kita} for each term $T_1, T_2, T_3, T_4, T_5, T_6$. To do so, it suffices to check that for each random variable in the summand the moment bound is sufficiently attained. 
By Assumption \ref{ass:sampling2}(iv), for all \(i\),
\[
\mathbb E\big[|\hat\tau_i-\tau_i|^3\mid S_i\big], 
\mathbb E\big[|\hat\tau_i^2-\mathbb E[\hat\tau_i^2\mid S_i]|^3\mid S_i\big], 
\mathbb E\big[|\hat\eta_i^2-\mathbb E[\hat\eta_i^2\mid S_i]|^3\mid S_i\big], 
\mathbb E\big[|\hat c_i-c_i|^3\mid S_i\big]\le M
\]
and in addition $|c_i|^3 \le M$ by assumption, and $|S_i|/\bar{p} \le 1$ almost surely. 
Therefore, conditional on $S$, we can bound for any $j \in \{1, \cdots, 6\}$ using Lemma \ref{lem:kita}, with $(K + B)_+^2 = \max\{1, (K + B)^2\}$
$$
\small 
\begin{aligned} 
\mathbb{E}[\mathbb{E}[T_j|S]]  \le q_0  (K + B)_+^2 \mathbb{E}\left[\mathbb{E}\Big[\sqrt{M \sum_i 1\{S_i > 0\}(C(\mathcal{F}) + v_\Pi)} | S\Big] \right] \le q_0  (K + B)_+^2  \sqrt{M N_{\mathcal{X}} (C(\mathcal{F}) + v_\Pi)}
\end{aligned} 
$$
for a universal constant $q_0$ (where we use $\sum_i 1\{S_i > 0\} \le N_{\mathcal{X}}$).  

\paragraph{Step 4: law of iterated expectations} Using the expression in Equation \eqref{eqn:equation_all}, and collecting the terms, 
$$
\small 
\begin{aligned} 
\mathbb E\left[\sup_{\pi\in\Pi,\ f\in\mathcal F}\bigl|(\widehat R_S-R_c)(f,\pi)\bigr|\right] \le \frac{2}{N_S} q_0 (K + B)_+^2 \bar{p}  \sqrt{M N_{\mathcal{X}}(C(\mathcal{F}) + v_\Pi)}. 
\end{aligned} 
$$ 
Finally note that because $\mathbb{E}[S_x] > \gamma$, $N_S = \sum_x \mathbb{E}[S_x] \ge \gamma N_{\mathcal{X}}$. Therefore we have $ \frac{2}{N_S} q_0 (K + B)_+^2 \bar{p}  \sqrt{M N_{\mathcal{X}}(C(\mathcal{F}) + v_\Pi)} \le \frac{2}{\gamma N_{\mathcal{X}}} q_0 (K + B)_+^2 \bar{p}  \sqrt{M N_{\mathcal{X}}(C(\mathcal{F}) + v_\Pi)}$ which completes the claim after rearrangement.

\subsection{Proof of Theorem \ref{thm:l0_sparse}} \label{proof:thm:l0_sparse}

We verify Assumption \ref{ass:entropy_F} for \(\mathcal F_{\ell_0}\) and then invoke Theorem \ref{thm:regret}. The argument is standard and follows from \cite{devroye2013probabilistic, vershynin2018high} among others. 

\paragraph{Step 1: $l_2$-cover} 
Since \(\|q(x)\|_2\le U\) and \(\|\theta\|_2\le R\), every \(f_\theta\in\mathcal F_{\ell_0}\) satisfies
$
|f_\theta(x)|\le \|q(x)\|_2\|\theta\|_2\le RU,
\forall x\in\mathcal X.
$
Thus Assumption \ref{ass:entropy_F} holds with envelope \(K=RU\), provided we bound the corresponding entropy integral. We now characterize $\mathcal{C}(\mathcal F_{\ell_0})$. This follows standard arguments in empirical process theory \citep{van1996weak}. Fix a probability measure \(Q\) on \(\mathcal X\). For \(\theta,\theta'\in\Theta_0\),
\begin{equation} \label{eqn:l1_bound}
\|f_\theta-f_{\theta'}\|_{Q,1}
=
\int_{\mathcal X}\big|q(x)'(\theta-\theta')\big|\,dQ(x)
\le
\int_{\mathcal X}\|q(x)\|_2\,\|\theta-\theta'\|_2\,dQ(x)
\le
U\|\theta-\theta'\|_2.
\end{equation}
Therefore, an \(\ell_2\)-cover of \(\Theta_0\) at radius \(\varepsilon/U\) induces an \(L_1(Q)\)-cover of \(\mathcal F_{\ell_0}\) at radius \(\varepsilon\).

\paragraph{Step 2: slicing} Now decompose \(\Theta_0\) according to supports. For each subset \(S\subseteq\{1,\dots,d\}\) with \(|S|\le s\), let
$
\Theta_S:=\{\theta\in\mathbb R^d: \theta_j = 0 \forall j \not \in  S,\ \|\theta\|_2\le R\}.
$
Then
$
\Theta_0=\bigcup_{|S|\le s}\Theta_S.
$ 
For fixed \(S\), \(\Theta_S\) is an Euclidean ball of radius \(R\) in dimension \(|S|\le s\), so its \(\ell_2\)-covering number satisfies for any $Q$ on $x_1, \cdots, x_{N_{\mathcal{X}}}$, 
$ 
\mathcal{N}\!\left(\frac{\varepsilon}{U},\Theta_S, ||\cdot||_2\right)
\le
\left(1+\frac{2RU}{\varepsilon}\right)^s.
$ 
Moreover, the number of supports of size at most \(s\) is bounded by
$
\sum_{j=0}^s \binom{d}{j}\le \left(\frac{ed}{s}\right)^s.
$ 
By Lemma \ref{lem:cover_union}, and using the bound on the covering number under the $l_1$-norm that follows from Equation \eqref{eqn:l1_bound}, this implies, 
$ 
\mathcal{N}(\varepsilon,\mathcal F_{\ell_0},L_1(Q))
\le
\left(\frac{ed}{s}\right)^s
\left(1+\frac{2RU}{\varepsilon}\right)^s.
$ 
\paragraph{Step 3: Dudley's entropy integral bound} Hence
$ 
\int_0^{2RU}\sqrt{\log \mathcal{N}(\varepsilon,\mathcal F_{\ell_0},L_1(Q))}\,d\varepsilon
\le
\int_0^{2RU}
\sqrt{s\log\!\left(\frac{ed}{s}\right)
+
s\log\!\left(1+\frac{2RU}{\varepsilon}\right)}\,d\varepsilon.
$ 
Using \(\sqrt{a+b}\le \sqrt a+\sqrt b\), we can write \\ 
$ 
\int_0^{2RU}\sqrt{\log \mathcal{N}(\varepsilon,\mathcal F_{\ell_0},L_1(Q))}\,d\varepsilon \le 2RU\sqrt{s\log\!\left(\frac{ed}{s}\right)}
+
\sqrt s\int_0^{2RU}\sqrt{\log\!\left(1+\frac{2RU}{\varepsilon}\right)}\,d\varepsilon.
$ 
With the change of variables \(t=\varepsilon/(2RU)\), the second term becomes
$ 
2RU\sqrt s\int_0^1 \sqrt{\log(1+t^{-1})}\,dt
\le
c_1 RU\sqrt s
$ 
for a universal constant \(c_1<\infty\). Therefore, for a universal constant $c_2 < \infty$, 
$ 
\sup_{Q\in\mathcal Q}
\int_0^{2RU}\sqrt{\log \mathcal{N}(\varepsilon,\mathcal F_{\ell_0},L_1(Q))}\,d\varepsilon
\le
c_2 RU\sqrt{s\log\!\left(\frac{ed}{s}\right)}.
$ 
Thus Assumption \ref{ass:entropy_F} holds with
$
C(\mathcal F_{\ell_0})\le c_2^2 R^2U^2\, s\log\!\left(\frac{ed}{s}\right).
$ 
The result follows directly from Theorem \ref{thm:regret}.

\subsection{Proof of Proposition \ref{prop:conservative_cost_model_selection}} 
\label{proof: prop:conservative_cost_model_selection}
We apply Theorem \ref{thm:model_selection} with \(\bar c\) in place of \(c\). This yields
$$ 
\small 
 \begin{aligned} 
& \mathbb E\!\left[
R_{\bar c}(\hat f^{\,\bar c},\hat\pi^{\,\bar c})
-
\min_{f\in\bar{\mathcal F}}
\sum_{x\in\mathcal X} p(x)\bigl(f(x)-\tau(x)\bigr)^2
\right] 
\\ &\le
c_0 J^{1/2}\sqrt{\frac{M_0 + M_0^2}{|\mathcal X|}}
+
\inf_{1\le j \le J}
\left\{
c_0
\sqrt{\tilde M_u\,\frac{C(\mathcal F_j)+v_{\Pi_j}}{|\mathcal X|}}
+
\Delta_{j}^{\mathrm{loc}}(\pi_{j,\bar c}^\star)
+
A_{\bar c}(\pi_{j,\bar c}^\star)
\right\}.
\end{aligned} 
$$ 
Since \(\bar c(x)\ge c(x)\) for all \(x\in\mathcal X\), we have
$
R_c(f,\pi)\le R_{\bar c}(f,\pi)$ 
for all $(f,\pi)$. On the other hand, 
the no-abstention benchmark
$ 
\min_{f\in\bar{\mathcal F}}
\sum_{x\in\mathcal X} p(x)\bigl(f(x)-\tau(x)\bigr)^2
$ 
does not depend on the abstention cost. Therefore
$$
\begin{aligned}
\mathbb E\!\left[
R_c(\hat f^{\,\bar c},\hat\pi^{\,\bar c})
-
\min_{f\in\bar{\mathcal F}}
\sum_{x\in\mathcal X} p(x)\bigl(f(x)-\tau(x)\bigr)^2
\right]
\le
\mathbb E\!\left[
R_{\bar c}(\hat f^{\,\bar c},\hat\pi^{\,\bar c})
-
\min_{f\in\bar{\mathcal F}}
\sum_{x\in\mathcal X} p(x)\bigl(f(x)-\tau(x)\bigr)^2
\right],
\end{aligned} 
$$ 
which proves the claim.

\subsection{Proof of Theorem \ref{thm:regret_local}} \label{proof:thm:regret_local}

Following the argument of Equation \eqref{eqn:basic_in} it suffices to bound $ 
\mathbb E\left[\sup_{\pi\in\Pi, f \in \mathcal{F}}\bigl|(\widehat R-R_c)(f,\pi)\bigr|\right].
$ 
By Brook's theorem (following a similar argument as in \cite{viviano2024policy}), we can partition units $x \in \mathcal{X}$ into $L + 1$ disjoint groups $\{\mathcal{X}_l\}_{l=1}^{L+1}, \cup_{l=1}^{L+1}\mathcal{X}_l = \mathcal{X}$, with independent observations within each group. Therefore, by the triangular inequality, we can write 
$ 
\mathbb E\left[\sup_{\pi\in\Pi, f \in \mathcal{F}}\bigl|(\widehat R-R_c)(f,\pi)\bigr|\right] \le \sum_{l = 1}^{L+1}\mathbb{E}\left[\sup_{\pi\in\Pi, f \in \mathcal{F}}\bigl|(\widehat{R}^l-R_c^l)(f,\pi)\bigr| \right].
$ 
where we denote $R_c^L = \sum_{x \in \mathcal{X}_l} p(x) \left\{ (f(x) - \tau(x))^2 \pi(x) + c(x)(1 - \pi(x))\right\}$ and $\widehat{R}^l$ the empirical analog. We follow then verbatim the proof of Theorem \ref{thm:regret} to bound each component $\mathbb{E}\left[\sup_{\pi\in\Pi, f \in \mathcal{F}}\bigl|(\widehat{R}^l-R_c^l)(f,\pi)\bigr| \right]$. 

\subsection{Proof of Proposition \ref{cor:error}} \label{proof:cor:error}

Denote by \(\mathcal P_L^{\mathrm{un}}\) the unrestricted depth-\(L\) tree class solved by Algorithm \ref{alg:alg3}, and by \(\mathcal P_{L,G}\subseteq \mathcal P_L^{\mathrm{un}}\) the constrained subclass with at most \(G-1\) predictive leaves. Let
$ 
E_L^\star
:=
\inf_{(f,\pi)\in\mathcal P_L^{\mathrm{un}}}\widehat R(f,\pi)
$ 
be the empirical optimum returned by Algorithm \ref{alg:alg3}. Algorithm \ref{alg:alg4} returns a pruned rule \((\hat f^t,\hat\pi^t)\) and reports
$ 
\varepsilon
:=
\widehat R(\hat f^t,\hat\pi^t)-E_L^\star .
$ 
Since \(\mathcal P_{L,G}\subseteq\mathcal P_L^{\mathrm{un}}\), we have
$ 
E_L^\star
\le
\inf_{(f,\pi)\in\mathcal P_{L,G}}\widehat R(f,\pi).
$ 
Therefore,
$ 
\widehat R(\hat f^t,\hat\pi^t)
\le
\inf_{(f,\pi)\in\mathcal P_{L,G}}\widehat R(f,\pi)
+\varepsilon.
$ 
Thus \(\varepsilon\) is an observable upper bound on the empirical optimization error of the pruned rule relative to the constrained class \(\mathcal P_{L,G}\).
 
Repeating the basic inequality used in the proof of Theorem \ref{thm:regret} (Equation \eqref{eqn:basic_in}), 
$ 
R_c(\hat f^t,\hat\pi^t)
-
\inf_{(f,\pi)\in\mathcal P_{L,G}}
R_c\!\left(f,\pi\right)
\le
2\sup_{(f,\pi)\in\mathcal P_{L,G}}
\left|
\widehat R\!\left(f,\pi\right)
-
R_c\!\left(f,\pi\right)
\right|
+\varepsilon.
$ 
Following verbatim Step 1-Step 4 in the proof of Theorem \ref{thm:regret}, and using the complexity bound in the proof of Theorem \ref{thm:clustering_regret} (Step 3), we can write 
$$
\mathbb{E}\left[2\sup_{(f,\pi)\in\mathcal P_{L,G}}
\left|
\widehat R\!\left(f,\pi\right)
-
R_c\!\left(f,\pi\right)
\right|\right] \le q_0 \bar p \max\{1,K^2\} B^{3/2}
\sqrt{
\tilde M_u\,
\frac{(v+1)\log G}{\underline\kappa |\mathcal X|}
}, 
$$
where $v$ is the largest VC complexity of $x \mapsto 1\{\alpha(x) = g\}, g \in \{1, \cdots, G\}$.  
Using Markov's inequality, it follows that with probability at least $1- \delta$, 
$$
2\sup_{(f,\pi)\in\mathcal P_{L,G}}
\left|
\widehat R\!\left(f,\pi\right)
-
R_c\!\left(f,\pi\right)
\right| \le \frac{q_0}{\delta} \bar p \max\{1,K^2\} B^{3/2}
\sqrt{
\tilde M_u\,
\frac{(v+1)\log G}{\underline\kappa |\mathcal X|}
}. 
$$
Since $v \le d G\log(G)$ by Example \ref{exmp:tree_entropy}, 
the proof completes after collecting the terms.

\subsection{Proof of Theorem \ref{thm:breakeven}} \label{proof:thm:breakeven}

We will follow similarly the proof of Theorem \ref{thm:breakevena} with changes in notation whenever necessary (and with a slight abuse of notation). 
\paragraph{Notation} Before deriving the theorem, we introduce some notation. Let 
\begin{equation} \label{eqn:xi_x}
\small 
\begin{aligned}
\xi_x^{(m)}
=
\Big\{(\hat f^{(m)}(x)-\hat\tau_{oos}(x))^2 - (\hat{f}^{(0)}(x)-\hat\tau_{oos}(x))^2\Big\}\hat\pi^{(m)}(x)
\end{aligned} 
\end{equation} 
and similarly denote 
\begin{equation} \label{eqn:sigma_m}
\sigma_m^2 := |\mathcal X|\sum_{x\in\mathcal X} p(x)^2
\mathbb{V}\!\left(\xi_x^{(m)}\mid \mathcal{T}\right), \qquad \mathcal{T} := \left( (\hat{f}^{(m')}, \hat{\pi}^{(m')})_{m' \in \mathcal{M}}, \hat{f}^{(0)}\right) 
\end{equation} 
Note that because $p(x)^2$ is of order $1/N_{\mathcal{X}}^2$, we expect $\sigma_m^2(f) = O(1)$. 
We write 
\begin{equation} \label{eqn:delta_hat_f}
\small 
\begin{aligned} 
\widehat{\Delta}(m) = \sum_x p(x) \xi_x^{(m)}, \qquad \Delta(m)
:= \sum_x p(x) \Big\{(\hat f^{(m)}(x)- \tau(x))^2 - (\hat{f}^{(0)}(x)- \tau(x))^2\Big\}\hat\pi^{(m)}(x)
.
\end{aligned} 
\end{equation}
Note that because $\mathbb{E}[(\hat{f}^{(m)}(x) - \hat{\tau}_{oos}(x))^2 | \mathcal{T}] = \eta(x)^2 + (\hat{f}^{(m)}(x) - \tau(x))^2$ and similarly $\mathbb{E}[(\hat{f}^{(0)}(x) - \hat{\tau}_{oos}(x))^2 | \mathcal{T}] = \eta(x)^2 + (\hat{f}^{(0)}(x) - \tau(x))^2$, it follows that $\mathbb{E}[\widehat{\Delta}(m) | \mathcal{T}] =\Delta(m)$. 
To avoid degenerate solutions (i.e., the variance equals zero), we assume that the relevant variance is bounded away from zero.

To establish asymptotic normality, we first derive the following lemma. 

\begin{lem}[Marginal asymptotic normality of the break-even statistic] \label{prop:clt_break_even} Let Assumptions \ref{ass:general_regret}, \ref{ass:entropy_F} hold for the out of sample estimates, and assume the out-of-sample estimates are independent of $\mathcal{T}$ (in-sample estimates).
Take any $m \in \mathcal{M}$ which is $\mathcal{T}$-measurable. Let $\sigma_m^2 > \underline{l}$ for a constant $\underline{l} > 0$ almost surely, with $\sigma_m^2$ as defined in Equation \eqref{eqn:sigma_m}.  Then
\[
\frac{\sqrt{|\mathcal X|} \Big(\widehat\Delta(m)-\mathbb{E}[\widehat{\Delta}(m) | \mathcal{T}]\Big)}{\sigma_m}
\overset{d}{\longrightarrow}N(0,1).
\]
\end{lem}

\begin{proof}[Proof of Lemma \ref{prop:clt_break_even}]
We will write $\xi_x := \xi_x^{(m)}$ whenever clear from the context. 
By Assumption \ref{ass:general_regret} (independence), and independence of the training and out-of-sample estimates, the random variables 
$ 
p(x)\Big(\xi_x-\mathbb E[\xi_x\mid \mathcal{T}]\Big),
$ 
are conditionally independent and centered given $\mathcal{T}$. 

Therefore, we establish the central limit theorem using Lyapunov central limit theorem conditional on $\mathcal{T}$. 
We first verify a uniform third-moment bound. Write
\[
A_x^{(m)}
:=
\Big\{(\hat f^{(m)}(x)-\hat\tau_{oos}(x))^2\Big\}\hat\pi^{(m)}(x), \qquad 
A_x^{(0)}
:=
(\hat{f}^{(0)}(x)-\hat\tau_{oos}(x))^2\hat\pi^{(m)}(x),
\]
so that \(\xi_x=A_x^{(m)}-A_x^{(0)}\). By Lemma \ref{lem:centered_square}, Assumption \ref{ass:general_regret}(ii) (the uniform bound \(|\tau(x)|\le B\)) and the envelope bound \(|\hat f^{(m)}(x)|,|\hat{f}^{(0)}|\le K\) in Assumption \ref{ass:entropy_F}, there exists a finite constant \(c_0<\infty\), that only depends on $M_0, B, K < \infty$, so that 
\[
\max\Big\{
\mathbb E\big[|A_x^{(m)}-\mathbb E[A_x^{(m)}\mid  \mathcal{T}]|^3\mid \mathcal{T}\big],
\mathbb E\big[|A_x^{(0)}-\mathbb E[A_x^{(0)}\mid \mathcal{T}]|^3\mid  \mathcal{T}\big]
\Big\}
\le
c_0 < \infty.
\]
Hence, since \(|a-b|^3\le 4(|a|^3+|b|^3)\), for a finite constant $c_1 < \infty$ (a function of $c_0$)
\[
\mathbb E\Big[\big|\xi_x-\mathbb E[\xi_x\mid \mathcal{T}]\big|^3
\mid \mathcal{T}\Big]
\le
c_1 < \infty. 
\]
By Assumption \ref{ass:general_regret}(iii),
$
\sum_{x\in\mathcal X} p(x)^3
\le
\max_{x\in\mathcal X}p(x)\sum_{x\in\mathcal X}p(x)^2
\le
\frac{\bar p}{|\mathcal X|}\cdot \frac{\bar p}{|\mathcal X|}
=
\frac{\bar p^2}{|\mathcal X|^2}.
$
Therefore
$ 
\sum_{x\in\mathcal X}
\mathbb E\Big[
\big|p(x)\big(\xi_x-\mathbb E[\xi_x\mid \mathcal{T}]\big)\big|^3
\mid \mathcal{T}
\Big]
\le
c_1
\sum_{x\in\mathcal X} p(x)^3
\le
c_1\bar p^2 |\mathcal X|^{-2}.
$ 
Since \(\sigma_m^2 > \underline{l}\) for a positive constant $\underline{l} > 0$ by assumption, it follows that
\[
\frac{
|\mathcal{X}|^{3/2} \sum_{x\in\mathcal X}
\mathbb E\Big[
\big|p(x)\big(\xi_x-\mathbb E[\xi_x\mid \mathcal{T}]\big)\big|^3
\mid \mathcal{T}
\Big]
}{
\sigma_m^{3}
}
\longrightarrow 0.
\]
Thus the Lyapunov condition holds conditionally on $\mathcal{T}$. Therefore, it follows that as we define $Z_{|\mathcal{X}|} = \frac{\sqrt{|\mathcal X|} \Big(\widehat\Delta(m)-\mathbb{E}[\widehat{\Delta}(m) | \mathcal{T}]\Big)}{\sigma_m}$, we have $P(Z_{|\mathcal{X}|} \le t|\mathcal{T}) \rightarrow \Phi(t)$ where $\Phi(t)$ is the Gaussian CDF. Integrating over \(\mathcal T\) over both sides and applying the dominated convergence theorem gives the
unconditional convergence stated in the lemma.
\end{proof}

The rest of the proof mimics the proof of Theorem \ref{thm:breakevena} and omitted for brevity.

\end{document}